\documentclass[sigconf,authorversion,nonacm,dvipsnames]{acmart}

\usepackage{amsmath,amsthm}
\usepackage{subcaption} 
\usepackage{caption} 
\usepackage{graphicx}
\usepackage{dsfont}
\usepackage[capitalise]{cleveref}
\usepackage{varwidth}
\usepackage{multirow}
\usepackage{tikz}
\usepackage{textcomp}
\usepackage{listings}
\usepackage{pifont}
\usepackage{xspace}
\usepackage{colortbl}
\usepackage{wrapfig}
\usepackage{algorithm}
\usepackage{algpseudocode} 
\usepackage{tabularx,booktabs}
\usepackage{makecell}
\usepackage{apxproof}
\usepackage{thmtools,thm-restate}
\theoremstyle{acmplain}

\usepackage{array} 
\usepackage{standalone} 

\usepackage{enumitem}
\setitemize{noitemsep,topsep=0pt,parsep=0pt,partopsep=0pt}
\setenumerate{noitemsep,topsep=0pt,parsep=0pt,partopsep=0pt}

\usetikzlibrary{intersections}

\newcommand{\eat}[1]{}
\newcommand{\cut}[1]{}
\algrenewcommand{\algorithmiccomment}[1]{\hspace{2mm}\blue{$\triangleright$ #1}}
\algrenewcommand\algorithmicrequire{\textbf{Input:}}
\algrenewcommand\algorithmicensure{\textbf{Output:}}

\newcommand\blfootnote[1]{%
  \begingroup
  \renewcommand\thefootnote{}\footnote{#1}%
  \addtocounter{footnote}{-1}%
  \endgroup
}

\newcommand{\distrsensitive}[1]{{\ensuremath{\mathcal{D}_{sens, #1}}}}
\newcommand{\prob}{\textsf{OptSRepr$^2$}}
\newcommand{\subprob}{\textsf{OptSRepair}}

\newcommand{\rsrepair}{RS-repair}
\newcommand{\srepair}{S-repair}

\newcommand{\commonlhs}{common LHS}
\newcommand{\consensusfd}{consensus FD}
\newcommand{\lhschain}{LHS chain}
\newcommand{\lhsmarriage}{LHS marriage}
\newcommand{\rc}{\ensuremath{\rho}}

\newcommand{\fdsatisfies}{\ensuremath{\models}}
\newcommand{\rcsatisfies}{\ensuremath{\models}}

\newcommand{\globalilp}{\textsf{GlobalILP}}

\newcommand{\commonlhsreduction}{\textsf{CommonLHSReduction}}
\newcommand{\consensusreduction}{\textsf{ConsensusReduction}}

\newcommand{\candidateset}{candidate set}
\newcommand{\allsrepair}[1]{\ensuremath{\mathcal{A}_{#1}}}
\newcommand{\candset}[1]{\ensuremath{\mathcal{C}_{#1}}}

\newcommand{\distrdominates}{\succ}

\newcommand{\floor}[1]{{\ensuremath{\lfloor #1 \rfloor}}}
\newcommand{\ceil}[1]{{\ensuremath{\lceil #1 \rceil}}}

\newcommand{\postclean}{\textsf{PostClean}}

\newcommand{\dpalgo}{\textsf{LhsChain-DP}}
\newcommand{\lpgreedy}{\textsf{LP $+$ GreedyRounding}}
\newcommand{\lprepr}{\textsf{LP $+$ ReprRounding}}
\newcommand{\scalableheuristic}{\textsf{FDCleanser}}

\newcommand{\relation}{\ensuremath{R}}
\newcommand{\fdset}{\ensuremath{\Delta}}
\newcommand{\reduce}{\textsf{Reduce}} 
\def\fd#1#2{\mathsf{#1\rightarrow #2}}

\newcommand{\recurrence}[1]{\ensuremath{F_{#1}}}
\newcommand{\repreq}{\ensuremath{=_{Repr}}}
\newcommand{\reprdom}{\ensuremath{\distrdominates_{Repr}}}
\newcommand{\reprinsert}{\ensuremath{\textsf{ReprInsert}}}

\newcommand{\bilp}{\textsc{ILP-Baseline}}
\newcommand{\bapprox}{\textsc{VC-approx-Baseline}}
\newcommand{\bdp}{\textsc{DP-Baseline}}
\newcommand{\bmuse}{\textsc{MuSe-Baseline}}
\newcommand{\bpostclean}{\textsc{PostClean}}

\newcommand{\expglobalilp}{\textsc{\globalilp{}}}
\newcommand{\explpgreedy}{\textsc{\lpgreedy{}}}
\newcommand{\explprepr}{\textsc{\lprepr{}}}
\newcommand{\expdpalgo}{\textsc{\dpalgo{}}}
\newcommand{\expscalableheuristic}{\textsc{\scalableheuristic{}}}

\newcommand{\nretained}[1]{\ensuremath{N_{\text{ret}}(#1)}}
\newcommand{\ndeleted}[1]{\ensuremath{\#_{\text{del}}(#1)}}

\newcommand{\schema}{\ensuremath{\mathcal{S}}}

\newcommand{\dom}{\ensuremath{Dom}}

\def\rrr#1\\{\par
\medskip\hbox{\vbox{\parindent=2em\hsize=6.12in
\hangindent=4em\hangafter=1#1}}}

\definecolor{Gray}{gray}{0.85}

\newcolumntype{a}{>{\columncolor{Gray}}c}
\newcolumntype{e}{>{\columncolor{white}}c}
\newcolumntype{d}{>{\columncolor{Gray}\centering} p{1.9in}}

\newcommand{\red}[1]{{\color{red} {#1}}}

\newcommand{\blue}[1]{{\color{blue} {#1}}}

\newcommand{\paratitle}[1]{\paragraph*{#1}}

\newcommand{\benny}[1]{{\texttt{\color{green} Benny: [{#1}]}}}

\newcommand{\ag}[1]{{\texttt{\color{purple} Amir: [{#1}]}}}

\newcommand{\sr}[1]{{\color{magenta}{\tt SR: #1}}}

\newcommand{\revb}[1]{{\leavevmode\color{Magenta}{#1}}}
\newcommand{\revc}[1]{{\leavevmode\color{red}{#1}}}
\newcommand{\common}[1]{{\leavevmode\color{Plum}{#1}}}

\newtheorem{theorem}{Theorem}
\newtheorem{lemma}{Lemma}
\newtheorem{proposition}[lemma]{Proposition}

\newtheorem{example}[lemma]{Example}

\newtheorem{definition}{Definition}

\def\att#1{\mathsf{#1}}

\newcommand\vldbdoi{XX.XX/XXX.XX}
\newcommand\vldbpages{XXX-XXX}
\newcommand\vldbvolume{18}
\newcommand\vldbissue{2}
\newcommand\vldbyear{2024}

\newcommand\vldbtitle{\shorttitle} 
\newcommand\vldbavailabilityurl{https://github.com/louisja1/RS-repair}
\newcommand\vldbpagestyle{empty}

\newcommand{\yuxi}[1]{{\color{cyan}Yuxi: #1}}
\newcommand{\fs}[1]{{\color{orange}Fangzhu: #1}}

\begin{document}

\title{The Cost of Representation by Subset Repairs}



\author{Yuxi Liu$^{*}$}
\affiliation{Duke University}
\email{yuxi.liu@duke.edu}

\author{Fangzhu Shen$^{*}$}
\affiliation{Duke University}
\email{fangzhu.shen@duke.edu}

\author{Kushagra Ghosh}
\affiliation{Duke University}
\email{kushagra.ghosh@duke.edu}

\author{Amir Gilad}
\affiliation{Hebrew University}
\email{amirg@cs.huji.ac.il}

\author{Benny Kimelfeld}
\affiliation{Technion}
\email{bennyk@cs.technion.ac.il}

\author{Sudeepa Roy}
\affiliation{Duke University}
\email{sudeepa@cs.duke.edu}

\begin{abstract}
Datasets may include errors, and specifically violations of integrity constraints, for various reasons. 
Standard techniques for ``minimal-cost'' database repairing resolve these violations by aiming for minimum change in the data, and in the process, may sway representations of different sub-populations. For instance, the repair may end up deleting more females than males, or more tuples from a certain age group or race, due to varying levels of inconsistency in different sub-populations. Such repaired data can mislead consumers when used for analytics, and can lead to biased decisions for downstream machine learning tasks. 
We study the ``cost of representation'' in subset repairs for functional dependencies. In simple terms, we target the question of how many additional tuples have to be deleted if we want to satisfy not only the integrity constraints but also representation constraints for given sub-populations. We study the complexity of this problem and compare it with the complexity of optimal subset repairs without representations. While the problem is NP-hard in general, we give polynomial-time algorithms for special cases, and efficient heuristics for general cases. We perform a suite of experiments that show the effectiveness of our algorithms in computing or approximating the cost of representation.

%
%
\end{abstract}

\maketitle

\pagestyle{\vldbpagestyle}
\begingroup\small\noindent\raggedright\textbf{PVLDB Reference Format:}\\
Yuxi Liu, Fangzhu Shen, Kushagra Ghosh, Amir Gilad, Benny Kimelfeld, and Sudeepa Roy.
\vldbtitle. PVLDB, \vldbvolume(\vldbissue): \vldbpages, \vldbyear.\\
\href{https://doi.org/\vldbdoi}{doi:\vldbdoi}
\endgroup
\begingroup
\renewcommand\thefootnote{}\footnote{\noindent
This work is licensed under the Creative Commons BY-NC-ND 4.0 International License. Visit \url{https://creativecommons.org/licenses/by-nc-nd/4.0/} to view a copy of this license. For any use beyond those covered by this license, obtain permission by emailing \href{mailto:info@vldb.org}{info@vldb.org}. Copyright is held by the owner/author(s). Publication rights licensed to the VLDB Endowment. \\
\raggedright Proceedings of the VLDB Endowment, Vol. \vldbvolume, No. \vldbissue\ %
ISSN 2150-8097. \\
\href{https://doi.org/\vldbdoi}{doi:\vldbdoi} \\
}\addtocounter{footnote}{-1}\endgroup

\ifdefempty{\vldbavailabilityurl}{}{
\vspace{.3cm}
\begingroup\small\noindent\raggedright\textbf{PVLDB Artifact Availability:}\\
The source code, data, and/or other artifacts have been made available at \url{\vldbavailabilityurl}.
\endgroup
}

\pagestyle{\vldbpagestyle}

\section{Introduction}\label{sec:intro}
\blfootnote{* Both authors contributed equally.}%
Real-world datasets may violate integrity constraints that are expected to hold in the dataset for various reasons such as noisy sources, imprecise extraction, integration of conflicting sources, and synthetic data generation. 
Among the basic data science tasks, repairing such noisy data is considered as one of the most time-consuming and important steps, as it lays the foundation for subsequent tasks that rely on high-quality data~\cite{KandelPHH11,MullerLWPTLDE19}. 
Therefore, the problem of automatic data repairing has been the focus of much prior work~\cite{Afrati2009,RekatsinasCIR17,ChuIP13,BertossiKL13,FaginKK15,GiladDR20,GeertsMPS13-LLunatic,KrishnanWWFG16,livshits2020computing}. 
The existing literature on data repairing typically has the following high-level goal:
given a source database that violates a set of integrity constraints, find the closest database that satisfies the constraints. 
This problem has been studied in many settings, by varying the type of integrity constraints~\cite{bohannon2007conditional,koudas2009metric,ChomickiM05}, changing the database in different forms~\cite{ChuIP13,KolahiL09,BertossiKL13, livshits2020computing,ChomickiM05,LopatenkoB07}, and even relaxing it in various manners~\cite{RekatsinasCIR17,SaIKRR19}. 
Among these, a fundamental and well-studied instance of the problem is that of data repairing with tuple deletions (called {\em subset repair} or {\em \srepair}) \cite{KolahiL09, LivshitsKR18, livshits2020computing, 10.14778/3407790.3407809, DBLP:conf/icdt/MaslowskiW14}, when the integrity constraints are Functional Dependencies (FDs), e.g., a zip code cannot belong to two different cities.  
The aim is to find the minimum number of deletions in the input database so that the resulting database satisfies the FDs. 
In particular, previous work~\cite{LivshitsKR18, livshits2020computing} has 
characterized the tractable and intractable cases in this setting. 

Database repair may drastically sway the proportions of different populations in the data, specifically the proportions of various sensitive sub-populations.
For instance, the repair may end up deleting more females than males, or more tuples from a certain age group, race, or disability status, as illustrated in the sequel in an example. This may happen simply by chance while selecting one of many feasible optimal repairs, or, due to varying levels of inconsistency in different sub-populations in the collected data, which may arise due to varying familiarity with the data collection technology, imputing varying amounts of missing data in different groups due to concerns for ageism and other biases, etc. If data is repaired agnostic to the representations, it can lead to biased decisions for downstream (e.g., ML prediction) tasks \cite{DBLP:conf/edbt/SchelterHKS20, GrafbergerSS21, DBLP:conf/icde/GuhaKSS23}, and mislead consumers when used for analytics.  
Thus, it is only natural to require that the process of data repair that ensures the satisfaction of FDs, will also guarantee desired representation of different groups.  
Recent works have proposed ways of ensuring representation of different sub-populations (especially sensitive ones) and diversity in query results by considering different types of constraints~\cite{ShetiyaSAD22,LiMSJ23}.  
However, to our knowledge, such aspects have not been considered in the context of data repairing. 
\begin{figure}[ht]
\centering
\begin{minipage}{1\linewidth}
    \centering
    \begin{subfigure}{0.45\linewidth}
        \includegraphics[width=0.8\linewidth]{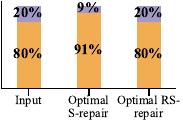}
        \centering
        \caption{Distribution of disability}
        \label{fig:motivating_example_a}
    \end{subfigure}%
    \hfill
    \begin{subfigure}{0.4\linewidth}
        \includegraphics[width=0.6\linewidth]{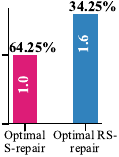}
        \centering
        \caption{Cost of representation}
        \label{fig:motivating_example_b}
    \end{subfigure}
\end{minipage}
\centering
\begin{minipage}{0.8\linewidth}
    \centering
    \includegraphics[width=0.5\linewidth]{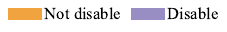}
\end{minipage}
\caption{Disability status and cost of representation for ACS data. In (b), the \% above the bars indicates remaining tuples, and the numbers on the bars show the deletion ratio relative to optimal S-repair (that does not satisfy representation).}
\label{fig:motivating_example}
\Description[]{}
\end{figure}

\cut{
\begin{figure}[t]
        \centering\includegraphics[width = 0.8\columnwidth]{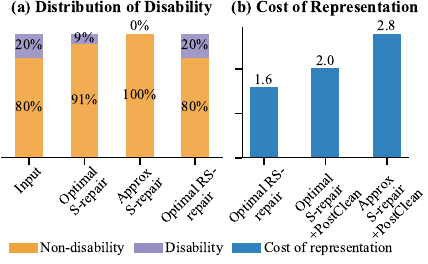}
        \caption{ACS data after repairing}        \label{fig:motivating_example}
        \Description[]{}
\vspace{-0.5cm}
\end{figure}
}
In this work, we introduce the problem of understanding the {cost of representation in data repair}, that considers representations of sub-populations in the repair as a first-class citizen just like the cost of changing the data. 
%
Our framework allows for FDs as well as a novel type of constraint called {\em representation constraint (RC)}. This constraint specifies the proportions of the different values in a sensitive attribute (e.g., gender, race, disability status), such as  ``{\em the percentage of population with disability should be at least 20\%}'', ``{\em the percentages of population with disability vs. non-disability should be exactly 20\%/80\%}'', etc. 
We then formally define the problem of finding an {\em optimal representative \srepair\ (\rsrepair\ for short)} as finding the minimum number of deletions required to satisfy {\em both} the FDs and the RC.
We devise algorithms that consider RCs as an integral part of the repairing process and compare the cost of optimal \srepair\ with the cost of optimal \rsrepair\ to understand the cost of representation that one has to pay for maintaining representations of a sensitive attribute in a dataset after repair.


\def\att#1{\mathsf{#1}}

\begin{example}\label{ex:intro-distribution}  
Consider a noisy dataset constructed from the ACS PUMS, with data collected from the US demographics survey by the Census Bureau. It is expected to satisfy a set of 9 FDs: $\att{Citizenship}$ to $\att{Nativity}$, 
$\att{State}$ to $\att{Division}$, etc.  (more details in \Cref{sec:experiments}). We focus on the sensitive attribute $\att{Disability}$, 
which has a $20\%/80\%$ distribution of disable and non-disable people in the dataset, where the data for the disabled group is 4-times more noisy than the non-disabled group.

%
    Suppose a data analyst wants to repair the dataset by subset repair (\srepair) such that all FDs are satisfied. One can write a simple integer linear program (ILP) to find the maximum \srepair{} so that for each pair of tuples that violate an FD, at least one is removed.  
    Although the 
    ILP method finds an optimal (maximal-size) \srepair,  as shown in \Cref{fig:motivating_example_a},  a side-effect of this repair is the drop in the proportion of people with disabilities 
    from $20\%$ to $9\%$, which makes a minority group less represented further. 
    ILP is not a scalable method for \srepair; if we were to use an efficient approximate algorithm~\cite{livshits2020computing, 10.14778/3407790.3407809}, \cut{($2344$ tuples deleted)} 
    no people with disability would stay in the repaired data. Consequently, both \srepair s may introduce biases against people with disability in downstream tasks that use the repaired datasets.

\end{example}

The above example shows that representation-agnostic \srepair{} methods can badly affect representations of groups. 
Our aim is to answer the following question: {\em Can we both obtain an \srepair{} and enforce a desired representation?} \Cref{fig:motivating_example} shows that ``Optimal \srepair'' retains 64.25\% of the original tuples, but does not satisfy the representation, while it is possible to obtain an \rsrepair{} that retains 34.25\% of the tuples. Although retaining only 34.25\% of the original tuples after the optimal \rsrepair\ looks pessimistic, this cost is necessary in this example. In other words, for this dataset and specific noise, if we insist on preserving the 20\%/80\% ratio of representation among disabled and non-disabled people respectively, we can retain at most 34.25\% of the original tuples, and will be forced to remove
1.6 times the number of tuples than a representation-agnostic optimal \srepair\footnote{More settings varying noise and representations are given in the experiments (\cref{sec:experiments}) and full version \cite{fullversion}.}.
Since computing the optimal \rsrepair\ is not always possible, in this paper, we study multiple other approaches for obtaining the \rsrepair s.

    \paragraph{Contributions.} In this paper, we focus on understanding the cost of representations for \srepair{}\footnote{We note that ``update repair'' methods, e.g., Holoclean~\cite{RekatsinasCIR17} and Nadeef~\cite{DallachiesaEEEIOT13}, have their own limitations or challenges when considering the cost of representations as discussed in 
    \Cref{sec:conclusion}.}, where the number of deleted tuples is used to measure the repair quality. Algorithms and complexity of \srepair{}, which is NP-hard in general, has been studied by multiple prior works~\cite{LivshitsKR18,GiladDR20,10.14778/3407790.3407809,LopatenkoB07}. 
Whenever finding an optimal (largest) S-repair is an intractable problem, so is the problem of finding an optimal RS-repair. The complexity of finding an optimal S-repair has been studied by \citet{LivshitsKR18, livshits2020computing}, who established a complete classification of complexity (dichotomy) based on the structural properties of the input FD set. We show that finding an RS-repair can be hard even if the FDs are such that an optimal \srepair\ can be computed in polynomial time. 

Next we investigate the complexity of computing an optimal \rsrepair{} for special cases of FDs and RC.
We present a polynomial-time dynamic-programming-based algorithm for a well studied class of FDs, namely the \emph{LHS-chains}~\cite{DBLP:journals/jcss/LivshitsKW21,LivshitsK21}, when the domain of the sensitive attribute is bounded (e.g., gender, race, disability status). 
For the general case, we  phrase the problem as an ILP, and devise heuristic algorithms that produce 
    \rsrepair s by rounding the ILP and by using the algorithm for LHS-chains.
    
    Finally, we perform an experimental study using three\cut{two} real-world datasets. We demonstrate the effect of representation-agnostic \srepair{}s on the representations of the sensitive attribute and show that optimal \rsrepair{}s delete $1\times$ to $2\times$ tuples compared to optimal \srepair{}s in \Cref{subsec:exp_cost_repr}, depending on how the noises are distributed. 
    \cut{compare the size of the optimal RS-repairs and the optimal S-repairs.} We conduct a thorough comparison between our algorithms, existing baselines, and a version of these baselines with post-hoc processing that further deletes tuples until the RC is satisfied,  and  analyze the quality and runtime of different approaches. We show that representation-aware subset-repair algorithms can find superior \rsrepair s in terms of the number of deletions.
In summary, our contributions are as follows.
\emph{(1)} We introduce the problem of finding an optimal RS-repair as a way of measuring the cost of representation.
\emph{(2)} We present complexity analysis of the problem.
\emph{(3)} We devise algorithms, including polynomial-time and ILP algorithms with optimality guarantees, and heuristics.
\emph{and (4)} We conduct a thorough experimental study of our solutions and show their effectiveness compared to the baselines.

    Due to space limitation, all proofs appear in the full version~\cite{fullversion}. 

\section{Preliminaries}\label{sec:background}
In this section, we present the background concepts and notations used in the rest of the paper. 

A relation (or table) $\relation{}$ is a set of tuples over a relation schema $\schema = (A_1, \dots, A_m)$, where 
$A_1, \dots, A_m$ are the attributes of $\relation$. 
A tuple $t$ in \relation\ maps each $A_i$, $1\leq i\leq m$, to a value that we denote by $t[A_i]$. This is referred to as a {\em cell} in $R$ in the sequel.  
We use $|\relation|$ to denote the number of tuples in $\relation$. 
The domain of an attribute $A_i$, denoted by $\dom(A_j)$, is the range of values that $A_i$ can be assigned to by $t$. 
Abusing notation, we denote by $t[X]$ the {\em projection} of tuple $t$ over the set $X$ of attributes, i.e., 
$t[X] = \pi_{X}{t}$.

A functional dependency (FD) is denoted as $X \rightarrow Y$, where  $X$ and $Y$ are disjoint sets of attributes. $X$ is termed the left-hand side (LHS), and $Y$ is the right-hand side (RHS) of the FD. A relation \relation{} satisfies $X \rightarrow Y$ if every pair of tuples that agree on $X$ also agree on $Y$. Formally, for every pair $t_i, t_j$ in \relation{}, if $t_i[X] = t_j[X]$, then $t_i[Y] = t_j[Y]$. 
A relation \relation{} satisfies a set \fdset{} of FDs, denoted $ \relation{} \models \fdset{}$, if it satisfies each FD from  $\fdset{}$.
When $\fdset{}$ is clear from the context, we refer to \relation{} as {\em clean} (respectively, {\em noisy}) if it satisfies (respectively, violates) $\fdset{}$. Without loss of generality\cut{(wlog.\benny{No need for period. Or say w.l.o.g. Are we using this shorthand notation at all?})}, we assume that the RHS $Y$ of each FD is a single attribute, otherwise we break the FD into multiple FDs. Note that the LHS $X$ can contain multiple attributes.
We also assume that the FDs $X \rightarrow Y$ are non-trivial, i.e., $Y \notin X$. 



Two types of database repairs have been mainly studied. A {\em subset repair (\srepair)} \cite{LivshitsKR18, KolahiL09, GiladDR20,10.14778/3407790.3407809}  changes a noisy relation by removing tuples, while an {\em update repair} \cite{livshits2020computing, RekatsinasCIR17, GeertsMPS13-LLunatic} changes values of cells. Each of the two has pros and cons. While the update repair retains the size of the dataset, it may generate invalid tuples (as discussed in \Cref{sec:intro}).
An \srepair\ uses original tuples, but at the cost of losing others. The complexity of computing an optimal \srepair\ is well understood \cite{LivshitsKR18}, whereas the picture for update repairs is still quite partial~\cite{LivshitsKR18, KolahiL09}. 
Hence, we focus on \srepair s and leave the update repairs for future work.
%
Formally, given \relation{} and $\fdset{}$, an \srepair\ is a subset of tuples\footnote{We note that in prior work \cite{LivshitsKR18, DBLP:conf/icdt/MaslowskiW14}, an \srepair{} has been defined as a ``{\em maximal}'' subset $R' \subseteq R$ such that $R' \fdsatisfies \fdset$. We consider even non-maximal subsets as valid \srepair s since in our problem, additional tuples from \srepair s may have to be removed to satisfy both FDs and representations.} $R' \subseteq \relation{}$ such that $R' \models \fdset{}$. $R'$ is called an {\em optimal subset repair} (or {\em optimal \srepair})
if for all \srepair{}s $R''$ of \relation{} given $\fdset{}$, $|R'| \geq |R''|$ (there is a weighted version when there is a weight $w(t)$ associated with each input tuple $t$). 


\paragraph{Computing optimal \srepair s.
}\label{sec:opt-s-repair}


\citet{LivshitsKR18} proposed an exact algorithmic characterization (dichotomy) for computing an optimal \srepair{}.
Moreover, it showed that when 
a specific procedure is 
not able to return an answer, the problem is NP-hard for the input $\fdset{}$ in data complexity \cite{Vardi82}. We briefly discuss these concepts and algorithms as we will use them in the sequel. 

A {\em \consensusfd{}} $\emptyset \rightarrow A$ is an FD where the LHS is the empty set, which means that all values of the attribute $A$ must be the same in the relation.
A {\em \commonlhs{}} of an FD set $\fdset{}$ is an attribute $A$ shared by the LHS of all FDs in $\fdset{}$, e.g., $A$ is a \commonlhs{} in $\fdset{} = \{A \rightarrow B, AC \rightarrow D\}$.  
An {\em \lhsmarriage{}} is a pair of distinct left-hand sides $(X_1, X_2)$ such that  every FD in $\fdset{}$ contains either $X_1$ or $X_2$ (or both), and $cl_{\fdset{}}(X_1) = cl_{\fdset{}}(X_2)$, where $cl_{\fdset{}}(X)$ is the {\em closure} of $X$ under $\fdset{}$, i.e. all attributes that can be inferred starting with $X$ using $\fdset{}$. For instance, $(A, B)$ forms a \lhsmarriage{} in $\fdset{} = \{A \rightarrow B, B \rightarrow A, A \rightarrow C\}$ where $cl_{\fdset{}}(A) = cl_{\fdset{}}(B) = \{A, B, C\}$.

Using the above concepts, three simplification methods for $\fdset{}$ are proposed, preserving the complexity of computing an optimal \srepair. This leads to a dichotomy: cases where computing an optimal \srepair\ is either polynomial-time solvable or NP-hard. In the polynomial-time cases, dynamic programming can be used to recursively find an optimal \srepair.

\section{Representative Repairs}
\label{sec:model}
In this section, we formally define the
problem of finding an optimal \rsrepair, and give an overview of  complexity and algorithms.
\subsection{Representation of a Sensitive Attribute}\label{sec:repr-sensitive}
\paratitle{Sensitive Attribute.} Without loss of generality, we denote the last attribute $A_s$ of \schema\ as the sensitive attribute\footnote{This initial study of cost of representations by subset repairs considers representations of sub-populations defined on a single sensitive attribute. Representations on a set of sensitive attributes will be an interesting future work (Section~\ref{sec:conclusion}).}. We denote the 
domain of $A_s$ in \relation{} 
as $\dom(A_s) = \{a_1, \cdots, a_k\}$, therefore $k = |\dom(A_s)|$ representing the size. A special case is when $A_s$ is binary, i.e., the sensitive attribute has two values: (1) a minority or protected group of interest (e.g., female, people with disability, and other underrepresented groups), and (2) the non-minority group or others. 

\begin{definition}[Representation Constraint]\label{def:rep_constraint}
Let $\mathcal{S}$ be a schema and $A_s$ a sensitive attribute. 
%
A \emph{lower-bound constraint} is an expression of the form $\%a \geq p$ where $a$ is a value and $p$ is a number in $[0,1]$. A \emph{Representation Constraint} (RC) 
$\rc{}$ is a finite set of lower-bound constraints.
A relation \relation{} satisfies $\%a \geq p$
if at least $p\cdot|R|$ of the tuples $t \in R$ satisfy $t[A_s] = a$. A relation \relation{} satisfies an RC $\rc{}$, denoted $\relation{} \rcsatisfies \rc{}$, if \relation{} satisfies every 
lower-bound constraint in $\rc{}$.

\end{definition}


We assume that \rc\ contains only constraints $\%a \geq p$ where $a$ is in the active domain of $A_s$; otherwise, for $p>0$ the constraint is unsatisfiable by any \srepair, and for $p=0$ the constraint is trivial and can be ignored. 
We also assume that $\rho$ has no redundancy, that is, it contains at most one lower-bound $\%a \geq p$ for every value $a$. We can also assume that the numbers $p$ in $\rc{}$ sum up to at most one, since otherwise the constraint is, again, infeasible. 

\rc{} is called an {\em exact RC} if $\sum_{i=1}^k p_i = 1$ because the only way to satisfy the individual constraints $\%a_i \geq p_i$ is to match $p_i$ exactly, i.e.,  $\%a_i = p_i$ for all $1 \leq i \leq k$. 
We refer to the lower-bound proportion $p_{\ell}$ of value $a_\ell$ 
in \rc{} by $\rc(a_{\ell})$.

For simplicity, our implementation restricts attention so that each proportion ($p_{\ell}$) in the input is a rational number, represented by an integer numerator and an integer denominator.
Moreover, if one does not specify a lower-bound constraint for some $a_{o}$, then we treat it as an trivial lower-bound constraint in the form of $\%a_o \geq 0$ or formally $\rc{(a_o)} = p_o = 0$. 

\begin{example}\label{eg:satisfy-rc}
Suppose that a relation \relation{} with the schema $\schema{} = (A_1, A_2, A_3)$ and the sensitive attribute $A_3$ contains four tuples $\{(1, a$ $, 3), (2, b, 5), (3, c, 9), (4, d, 3)\}$. The RC is $\rc{} = \{\%3 \geq \frac{1}{3}, \%5 \geq \frac{1}{3}, \%9 \geq \frac{1}{3}\}$. \relation{} does not satisfy \rc{}, but both the subset $R_1 = \{(1, a, 3)$, $(2, b, 5)$, $(3, c, 9)\}$ and the subset $R_2 = \{(2, b, 5), (3, c, 9), (4, d, 3)\}$ satisfy it. 
\end{example}

Next we define an \rsrepair\ 
for a set of FDs and an RC:
\begin{definition}[\rsrepair
]\label{def:repr_optsrepair}
Given a relation \relation{} with the sensitive attribute $A_s$, a set \fdset{} of FDs, and an RC \rc{}, a subset $R' \subseteq \relation{}$ is called an \rsrepair{} (representative subset repair) 
w.r.t. \fdset{} and \rc{} if:
\begin{itemize}[leftmargin=*]
\itemsep0em
    \item $R'$ is an S-repair of $R$, i.e., $R' \models \fdset{}$,
    \item $R'$ satisfies the RC $\rc{}$ on $A_s$, i.e., $R' \models \rc$.
\end{itemize}
We call $R'$ an {\em optimal \rsrepair} of $R$ w.r.t. $\Delta$ and $\rc$ 
if for all \rsrepair s $R''$ of $R$, we have $|R''| \leq |R'|$. 
\end{definition}


\begin{example}
    Continuing \Cref{eg:satisfy-rc}, if \fdset{} = $\{A_1 \rightarrow A_2\}$, then either $R_1$ or $R_2$ in can be an optimal RS-repair of $\relation{}$ w.r.t. \fdset{} and \rc{}.  
\end{example}


We study the problem of computing an optimal RS-repair. For our complexity analysis, we assume that $\mathcal{S}$ and $\fdset$ are fixed, and the input consists of the relation \relation{} and the RC $\rc{}$. 

    \paratitle{Choice of repair model for cost of representation} In this initial study of repairs with representation, we consider \srepair{} (deletion) as the repair model. 
    Multiple prior works have theoretically analyzed  \srepair{}s (without the cost of representation)~\cite{LivshitsKR18, KolahiL09, 10.14778/3407790.3407809, GiladDR20}, and the complexity of achieving an optimal \srepair\ based on the structure of the input FD set is well understood \cite{LivshitsKR18}. We focus on the framework and theoretical analysis for this simpler variant of the repair model with the representation criteria. Additionally, \srepair{}s keep tuples from the original dataset and do not introduce new combinations of values within the same tuple while repairing the data. We use the number of deletions to define the optimal \rsrepair\ as an extension for defining \srepair\ models from prior work. Extensions to other repair models and other cost functions (e.g., based on the effects on downstream tasks) are important and challenging future work (see \Cref{sec:conclusion}).

\subsubsection{NP-Hardness of Computing Optimal RS-Repairs}\label{sec:hardness}
In this section, we consider the relation $R$ and the RC \rc\ as inputs, while the schema \schema\ and the FD set $\fdset$ are fixed. 
We first note that since 
the problem of finding an optimal \rsrepair\ 
is a generalization of 
the problem of finding an optimal \srepair, 
as expected, in all cases where 
finding an optimal \srepair\
is NP-hard, 
finding an optimal \rsrepair\
is also NP-hard. 
\cut{
(formal proof in extended version \cite{fullversion}). \benny{This is very trivial. No need for proof, and no need for any formal result. It is enough is to say it by text.}
\sr{@Benny: the proof is not complex, but needs some care. for the simplest proof, you can simply set the lower bounds to zero in RC, but the RC will be a trivial RC that does not care about representation. The proof is for a non-trivial RC, where equal proprtions of 0 and 1 for senstive attribute is maintained. I think the proof can be extended to any valid RC with arbitrary proportions, at least any valid exact RC where the sum is 1. }
}
\Cref{thm:hardness} shows that computing an optimal \rsrepair\ is NP-hard even for a single FD. 
We prove this theorem by a reduction from 3-SAT
(proof in \cite{fullversion}).
\begin{restatable}{theorem}{hardnessthm}\label{thm:hardness}
     The problem of finding an optimal RS-repair is NP-hard already for $\mathcal{S}=(A,B,C)$ and $\Delta=\{A\rightarrow B\}$.
\end{restatable}

It is important to note that if \fdset\ contains a single FD, computing an optimal \srepair\ can be done in polynomial time
\cite{LivshitsKR18}.
The key distinction is on (the size of) the active domain of the sensitive attribute. \Cref{thm:poly-time-chain} in the next subsection describes a tractable scenario when the active domain of the sensitive attribute is bounded.

\subsection{Overview of Our Algorithms for Computing Optimal RS-Repairs}\label{sec:algo-reprepair}
As shown by Livshits et al. \cite{LivshitsKR18}, computing an optimal \srepair{} is poly-time solvable if $\fdset$ can be reduced to $\emptyset$ by repeated applications of three simplification processes: (i) consensus FDs (remove FDs of the form $\emptyset \rightarrow Y$), (ii) common LHS (remove attribute $A$ from $\fdset$, such that $A$ belongs to the LHS of all FDs in $\fdset$), and (iii) LHS marriage, which is slightly more complex. 
We will show that reduction to $\emptyset$ only by the first two simplification processes entails a polynomial time algorithm for computing optimal \rsrepair{}s when the sensitive attribute $A_s$ has a fixed number of distinct values (e.g., for common sensitive attributes gender, race, disability status, etc.). Before we formally state the theorem, we take a closer look at the class of FD sets that reduces to $\emptyset$ by the first two simplifications.

\begin{definition}\label{def:lhs-chain}
    An FD set $\Delta$ is an {\em LHS-chain} \cite{DBLP:conf/pods/LivshitsK17, LivshitsKR18} if for every two FDs $X_1 \rightarrow Y_1$ and $X_2 \rightarrow Y_2$, either $X_1 \subseteq X_2$ or $X_2 \subseteq X_1$ holds. 
\end{definition}
For instance, the FD set $\fdset_1 = \{A \rightarrow B, AC \rightarrow D\}$ is an LHS-chain. LHS-chains have been studied for \srepair{}s in prior work 
\cite{LivshitsKR18, DBLP:conf/pods/LivshitsK17}. \cite{DBLP:conf/pods/LivshitsK17} showed that the class of LHS-chains consists of precisely the FD sets where the \srepair{}s can be counted in polynomial time (assuming {\tt P $\neq$ $\#$P}). 
\cite{LivshitsKR18} observed that FD sets that form an LHS-chain can be simplified to the empty set by repeatedly applying simplifications on only the common LHS and the consensus FD. We show in the following proposition that the converse also holds: 
\begin{restatable}{proposition}{lhschainreductionprop}\label{prop:lhs-chain-reduction}
A set \fdset{} of FDs reduces to the $\emptyset$ by repeated applications of consensus FD and common LHS simplifcations if and only if $\fdset$ is an LHS-chain. 
\end{restatable}

The following theorem states our main algorithmic result.
\begin{restatable}{theorem}{polytimechainthm}\label{thm:poly-time-chain}
    Let $\schema{}$ be a fixed schema and $\fdset{}$ be a fixed FD set that forms an \lhschain{}.  
    Suppose that the domain size of the sensitive attribute $A_s$ is fixed. Then, an optimal \rsrepair{} can be computed in polynomial time.
\end{restatable}
\cut{
    \yuxi{this is the old version of our theorem}
    \begin{restatable}{theorem}{polytimechainthm}\label{thm:poly-time-chain}
        Given relation $R$ where the domain of sensitive attribute $\dom(A_s)$ is fixed, a fixed set of FDs $\fdset$, and a representative constraint $\rc$, the problem of computing an optimal \rsrepair{} of \relation{} for \fdset{} and \rc{} can be solved in polynomial time when $\fdset$ forms an LHS-chain, i.e., if \fdset{} can be reduced to $\emptyset$ by repeated applications of consensus FD and common LHS simplifications. 
    \end{restatable}
}

We present and analyze a dynamic programming (DP)-based algorithm \dpalgo{$(\relation{}, \fdset{}, \rc{})$} in \Cref{sec:poly-algo-chain} to prove the above theorem. \dpalgo{} not only gives an optimal algorithm for the special case of LHS-chains, but will also be used in \Cref{sec:general} as a procedure to obtain efficient heuristics for general FD sets where computing an optimal \rsrepair{} can be NP-hard. We give another (non-polynomial-time) optimal algorithm and several polynomial-time heuristics for cases with general FD sets in \Cref{sec:general}.

\par
\subsection{Can We Convert an \srepair\ to an \rsrepair?}
\label{subsec:postclean}
As discussed in the introduction for \Cref{ex:intro-distribution}, an intuitive heuristic to compute a RS-repair is (i) first compute an \srepair{} $R'$ (optimal or non-optimal) of $R$ w.r.t. $\fdset$, and (ii)
then delete additional tuples from $R'$ to obtain $R''$ that also satisfies the RC \rc. Following this idea, we present the \postclean{} algorithm,
which takes a relation $R$ and an RC \rc, and returns a maximum subset $R'$ of $R$ such that $R' \rcsatisfies \rc$.
\postclean{} has a dual use in this paper. First, it is used as a subroutine in several algorithms in the later sections when an \srepair{} of $R$ is used as the input relation to \postclean{}. Second, in \Cref{sec:experiments}, we also compose \postclean\ with several known approaches for computing \srepair s to create baselines for our algorithms. 
\paratitle{Overview of the \postclean{} algorithm} 
The PostClean algorithm 
intuitively works as follows (pseudo-code and analysis are in \cite{fullversion}).
Recall from \Cref{sec:repr-sensitive} that $\rc(a_\ell)$ denotes the lower bound on the fraction of the value $a_\ell$ in the tuples retained by the \rsrepair. Further note that sum of $\rc(a_\ell)$ may be smaller than $1$, i.e., the RC \rc{} may only specify the desired lower bounds for a subset of the sensitive values, and the rest can have arbitrary proportions as long as a minimum set of tuples is removed to obtain the optimal \rsrepair{}. Moreover, the fractions are computed w.r.t. the final repair size $|R'|$ and not w.r.t. the input relation size $|R|$. \postclean{} iterates over all possible sizes $T$ of $R'$ 
from $|R|$ to 1. For each $T$, it checks if the lower bound on the number of tuples with $a_\ell$, i.e., $\tau_{\ell} = \ceil{T \cdot \rc{(a_{\ell})}}$, is greater than the number of tuples with value $a_\ell$ in the original relation $R$. If yes, then no repair $R'$ of size $T$ can satisfy $\rc$, and it goes to the next value of $T$ (or $T \gets T - 1$). Otherwise, if there are sufficient tuples for all sensitive values $a_\ell$, and if the sum of the lower bounds on numbers, formally $T_0 = \sum{\tau_\ell}$ is $\leq T$, then we have a feasible $T$. Finally, the algorithm arbitrarily fills $R'$ with more tuples from $R$ if $T_0 < T$ and returns the final $R'$. Note that if all values of $T$ between $|\relation{}|$ and $1$ are invalid, then an $\emptyset{}$ is returned because it is the only subset of $\relation{}$ that (trivially) satisfies the $\rc$. 
The following states the optimality and runtime of \postclean{}. 


\begin{restatable}{proposition}{postcleanprop}\label{prop:post_clean}
Given $\relation$ and \rc{}, \postclean{$(R, \rc{})$} 
returns in polynomial time a maximum subset $R'$ of $R$ such that $R' \rcsatisfies \rc$.
\end{restatable}

Applying \postclean\ on an optimal \srepair{} may not lead to an optimal \rsrepair{} as illustrated below (and in \Cref{ex:intro-distribution}).

\begin{example}
Consider a relation \relation{} with $\schema{} = (A, B, \att{sex})$, a set \fdset{} of FDs $\{A \rightarrow B\}$, and an exact RC $\rc{} = \{\%\texttt{male}=\frac{1}{2}, \%\texttt{female}= \frac{1}{2}\}$. Let $\relation =$ $\{(1, a, \texttt{male})$, $(1, b, \texttt{female})$, $(2, c, \texttt{male})$, $(2, d, \texttt{female})\}$. An optimal \srepair{} of \relation{} w.r.t. \fdset{} is $R' = \{(1, a, \texttt{male}), (2, c, \texttt{male})\}$. However, $\postclean(R', \rc{})$ returns $\emptyset$ since $R'$ does not have any female tuples. Conversely, $\{(1, a, \texttt{male}), (2, d, \texttt{female})\}$ is an optimal RS-repair, which satisfies both \fdset{} and \rc{}.
\end{example}
\section{Dynamic Programming for LHS-chains}\label{sec:poly-algo-chain}

We now prove Theorem~\ref{thm:poly-time-chain} by presenting a DP-based optimal algorithm, \dpalgo{} (\Cref{alg:dpalgo}), that finds an optimal RS-repair for LHS-chains (\Cref{sec:algo-reprepair}, 
\Cref{def:lhs-chain}) in polynomial time when the sensitive attribute has a fixed domain size. 
By the property of an LHS-chain, $\fdset$ can be reduced to $\emptyset$ by repeated application of only consensus FD simplification and common LHS simplification. 
\par
{\em Overview of \dpalgo{}.}
   For a relation \relation{} and a set \fdset{} of FDs, let $\allsrepair{\relation{}, \fdset} = \{R' \subseteq R ~\mid~ R' \fdsatisfies \fdset\}$ be the set of all S-repairs of \relation{} for \fdset{}.
Intuitively, if we could enumerate all S-repairs $R'$ from $\allsrepair{R, \fdset}$, we could compute $R'' = \postclean(R', \rc)$ (\Cref{subsec:postclean}) for each of them and return the $R''$ with the maximum number of tuples. Since \postclean{} optimally returns a maximum subset satisfying $\rc$ for every $R'$, and since any RS-repair w.r.t. $\fdset$ and $\rc$ must be an S-repair w.r.t. $\fdset$, such an $R''$ is guaranteed to be an optimal \rsrepair{}. 
\begin{algorithm}
\caption{\dpalgo{$(\relation{}, \fdset{}, \rc{})$}}\label{alg:dpalgo}
    \begin{algorithmic}[1]
        \Require{a relation \relation{}, an LHS-chain FD set \fdset{}, and a RC \rc{}}
        \Ensure{an optimal RS-repair of $(\relation{}, \fdset{}, \rc{})$}
        \State $\candset{\relation{}, \fdset{}} \gets \reduce{(\relation{}, \fdset{})}$;
        \Comment{\Cref{alg:reduce} in \Cref{subsec:reduce}}
        \State $S \gets \{\postclean{(R', \rc{})} \mid  R' \in \candset{\relation{}, \fdset{}}\}$;
        \Comment{\Cref{subsec:postclean}}
        \State \Return $\arg\max\limits_{s \in S}{|s|}$;
    \end{algorithmic}
\end{algorithm}
However, even when the domain size of the sensitive attribute $A_s$ is fixed, the size of $\allsrepair{\relation{}, \fdset{}}$ can be exponential in $|\relation{}|$. Therefore, it is expensive to enumerate the set of all S-repairs. Hence, we find a \textit{candidate set} $\candset{\relation{}, \fdset{}} \subseteq \allsrepair{\relation{}, \fdset{}}$ of S-repairs that is sufficient to inspect. Then, we apply \postclean{} to each element of $\candset{\relation{}, \fdset{}}$, and return the final solution having the maximum size. 
We formally define candidate set $\candset{\relation{}, \fdset{}}$ in \Cref{subsec:candidates}, along with associated definitions. The basic idea is that there are no two S-repairs in $\candset{\relation{}, \fdset{}}$ 
where one is ``clearly better'' than the other or that the two ``are equivalent to each other''. Further, 
for any S-repair that is not in the candidate set i.e., $R'' \in \allsrepair{\relation{}, \fdset} \setminus \candset{\relation{}, \fdset{}}$, there is an S-repair $R' \in \candset{\relation{}, \fdset{}}$ that 
is ``clearly better'' or ``equivalent to'' $R''$. 
We prove (in \Cref{lemma:candset}) that an optimal RS-repair can be obtained by applying \postclean{} to each \srepair\ in $\candset{R, \fdset{}}$ and returning the one with maximum size. Moreover, the size of the candidate set is 
$O(|R|^k)$, 
when 
the domain size 
$|\dom(A_s)| = k$ 
is fixed (proofs in  \cite{fullversion}).

\Cref{alg:dpalgo} has two steps: Line 1 computes the candidate set $\candset{\relation{}, \fdset{}}$ by the recursive  \reduce\ procedure (\Cref{alg:reduce} in \Cref{subsec:reduce}), that divides the problem into smaller sub-problems by DP. Then Line 2 applies \postclean{} to each S-repair in $\candset{\relation{}, \fdset{}}$ and returns the maximum output as an optimal RS-repair in Line 3. 
\Cref{subsec:reduce} describes the \reduce{} procedure. Since $\fdset$ is an LHS-chain, it reduces to $\emptyset$ by repeated reductions of consensus FD (\Cref{subsec:consensusreduction}) and common LHS (\Cref{subsec:commonlhsreduction}). 
%
The correctness of \dpalgo{} follows from \Cref{lemma:candset,lem:reduce-correct} stated later.

\begin{restatable}{lemma}{dpalgocomplexitylem}\label{lem:dpalgo_complexity}
     \dpalgo{} terminates in $O(m \cdot |\fdset{}| \cdot k \cdot |\relation{}|^{3k + 2})$ time, where $m$ is the number of attributes in $R$, $|\fdset{}|$ is the number of FDs, and $k = |\dom(A_s)|$ is the domain size of the sensitive attribute $A_s$. 
\end{restatable}




\subsection{Candidate Set for Optimal RS-Repairs}\label{subsec:candidates}
Recall that $\allsrepair{\relation{}, \fdset}$ denotes the set of all S-repairs \relation{} w.r.t. \fdset{}. We define a \candidateset{} as the subset of $\allsrepair{R, \fdset}$ such that every \srepair{} in the \candidateset{} is neither {\em representatively dominated} by nor {\em representatively equivalent to} other S-repairs in terms of the sensitive attribute as defined below.
\begin{definition}\label{def:candset}
For a relation \relation{}, FD set \fdset{}, and $R_1, R_2 \in \allsrepair{\relation{}, \fdset}$:
\begin{itemize}[leftmargin=*]
\itemsep0em
    \item 
$R_1$ is {\em representatively equivalent to} $R_2$, denoted $R_1 \repreq R_2$,  if for all $a_\ell \in \dom(A_s)$, $|\sigma_{A_s = a_\ell}  R_1| = |\sigma_{A_s = a_\ell} R_2|$, i.e. the same number of tuples for each sensitive value.
\item $R_1$ {\em representatively dominates} $R_2$, denoted $R_1 \reprdom R_2$, if 
for all $a_\ell \in \dom(A_s)$,  $|\sigma_{A_s = a_\ell}  R_1| \geq |\sigma_{A_s = a_\ell}  R_2|$, and there exists $a_c \in \dom(A_s)$,  $|\sigma_{A_s = a_c}  R_1| > |\sigma_{A_s = a_c}  R_2|$.
\end{itemize}
\end{definition}

\begin{example}\label{example:repreq_and_reprdom}
    Consider three S-repairs for the relation $\relation{}$ with the schema $(A, B, \texttt{race})$ and FD set $\fdset{} = \{A \rightarrow B\}$: (1) $R'_1 = \{(1, 2, \texttt{white})$, $(2, 3, \texttt{black})\}$; (2) $R'_2 = \{(1, 3, \texttt{black})$, $(1, 3, \texttt{white})\}$; and (3) $R'_3 = $ $\{(1, 1, \texttt{black}), (2, 2, \texttt{white})$, $(3, 3, \texttt{asian})\}$. Here $R'_1 \repreq R'_2$ since they have the same number of \texttt{black} and \texttt{white} tuples. $R'_3 \reprdom R'_1$ and $R'_3 \reprdom R'_2$ since $R'_3$ has one more \texttt{asian} than $R'_1$ and $R'_2$.
\end{example}
    
\begin{definition}[\candidateset]\label{def:candidate_set}
    Given a relation \relation{} and any FD set $\fdset$, 
    a {\em \candidateset{}} denoted by \candset{\relation{}, \fdset{}} is a subset of $\allsrepair{R, \fdset}$ such that
    \begin{enumerate}
        \item For all $R_1, R_2 \in \candset{\relation{}, \fdset{}}$, 
        $R_1 \not\repreq R_2$, $R_1 \not\reprdom R_2$, and $R_2 \not\reprdom R_1$
        \item For any $R'' \in \allsrepair{R, \fdset} \setminus \candset{\relation, \fdset}$, there exists $R' \in \candset{\relation, \fdset}$ such that $R' \repreq R''$ or $R' \reprdom R''$.
        
    \end{enumerate}
    Each S-repair $R' \in \candset{\relation{}, \fdset{}}$ is called a {\em candidate}. 
\end{definition}

For the correctness of \dpalgo{}, we use the following lemma.

\begin{restatable}{lemma}{candsetlem}\label{lemma:candset}
    

    For any relation $R$ and any FD set $\fdset$, if $\candset{R, \fdset}$ is computed correctly in Line 1, \dpalgo{} (Algorithm~\ref{alg:dpalgo}) returns an optimal \rsrepair{} of \relation{} w.r.t. \fdset{} and \rc{}. \cut{
        \yuxi{this is separated from this lemma}
        Moreover, the size of the \candidateset{} $|\candset{\relation, \fdset}|$ is $O(|\relation|^k)$, where $k = |\dom{(A_s)}|$ is the domain size of $A_s$ and $|\relation{}|$ is the size of input relation.
    }
\end{restatable}


\paratitle{\reprinsert{}} A subroutine \reprinsert{$(\candset{}, R_0)$} (pseudocode in \cite{fullversion}) will be used in the following subsections. Intuitively, it safely processes an insertion to a set of candidates and maintains the properties in \Cref{def:candidate_set}. Specifically, it takes a set of candidates $\candset{}$, where no representative equivalence or representative dominance exists, and an S-repair $R_0 \in \allsrepair{R, \fdset}$ as inputs. \reprinsert{} compares $R_0$ with every $R' \in \candset{}$. If there is an $R'$ such that $R' \reprdom R_0$ or $R' \repreq R_0$, it returns the existing $\candset{}$. Otherwise it removes all $R'$ from $\candset{}$ where $R_0 \repreq R'$ and returns $\candset{} \cup \{R_0\}$.

\subsection{Recursive Computation of  Candidate Set}\label{subsec:reduce}
The procedure \reduce{} (\Cref{alg:reduce}) computes a \candidateset{} recursively when $\fdset$ is an LHS-chain. Since an LHS-chain $\fdset$ can be reduced to $\emptyset$ by repeated applications of consensus FD and common LHS, 
\reduce{} calls \consensusreduction{} (\Cref{subsec:consensusreduction}) and  \commonlhsreduction{} (\Cref{subsec:commonlhsreduction}) until $\fdset$ is empty. When $\fdset$ is empty, it returns $\{\relation{}\}$ as the singleton candidate set since \relation{} itself is an S-repair and representatively dominates all other S-repairs. The following lemma states the correctness of the \reduce{} procedure.

\begin{algorithm}[t]
\caption{\reduce{$(\relation{}, \fdset{})$}}\label{alg:reduce}
    \begin{algorithmic}[1]
        \Require{a relation \relation{}, a FD set \fdset{} that forms an \lhschain{}}
        \Ensure{a \candidateset{} $\candset{\relation{}, \fdset{}}$ w.r.t. \relation{} and \fdset{}}
        \If{$\fdset{}$ is empty}\label{line:}
            \State \Return $\candset{\relation{}, \fdset{}} := \{\relation{}\}$;
         \ElsIf{Identify a consensus FD $f: \emptyset \rightarrow Y$}
            \State \Return $\consensusreduction{(\relation{}, \fdset{}, f)}$;
            \Comment{\Cref{alg:consensus_reduction}}
        \ElsIf{Identify a common LHS $X$}
            \State \Return $\commonlhsreduction{(\relation{}, \fdset{}, X)}$;
            \Comment{\Cref{alg:common_lhs_reduction}} 
        \EndIf
    \end{algorithmic}
\end{algorithm}

\begin{restatable}{lemma}{reducecorrectlem}\label{lem:reduce-correct}
    Given relation $R$ and FD set $\fdset$ that forms an LHS-chain, \reduce{$(\relation{}, \fdset{})$} correctly computes the candidate set $\candset{\relation{}, \fdset}$.
\end{restatable}

\subsubsection{Reduction for Consensus FD}\label{subsec:consensusreduction}
Consider a consensus FD $f: \emptyset \rightarrow Y$. Within an S-repair, all values of $Y$ should be the same. Suppose that $\dom(Y) = \{y_1, \cdots, y_n\}$ and $\relation{}_{y_{\ell}} = \sigma_{Y = y_\ell} R$ denotes the subset of \relation{} that has the value $Y = y_\ell$. The procedure  \consensusreduction{} (Algorithm~\ref{alg:consensus_reduction}) computes the candidate set $\candset{R_{y_\ell}, \fdset - f}$ by calling $\reduce(\relation{}_{y_\ell}, \fdset - f)$ for every $y_{\ell}$. Note that any S-repair $R' \in \candset{R_{y_\ell}, \fdset - f}
$ is also an S-repair of \relation{} for \fdset{}, i.e., $R' \in \allsrepair{\relation{}, \fdset}$, since the FD $f$ is already taken care of in $\relation{}_{y_\ell}$. Hence, Line 5 combines these sets from smaller problems (i.e., $\reduce(\relation{}_{y_\ell}, \fdset - f)$ for every $y_{\ell} \in \dom{(Y)}$) by inserting their candidates into $\candset{\relation{}, \fdset}$ one by one 
so that the properties of a candidate set is maintained in $\candset{\relation{}, \fdset}$.

\subsubsection{Reduction for Common LHS}\label{subsec:commonlhsreduction}
Consider a common LHS attribute $X$ that appears on the LHS of all FDs in $\fdset$. Suppose that $\dom{(X)} = \{x_1, x_2, \dots, x_n\}$, $\relation{}_{x_\ell} = \sigma_{X = x_\ell}{\relation{}}$ denotes the subsets of $\relation{}$ that have value $X = x_{\ell}$, and $\relation{}_{x_1, \dots, x_{\ell}} = \sigma_{X = x_1 \lor \dots \lor X = x_{\ell}}{\relation{}}$ as an extension.
Also suppose $\fdset_{-X}$ denotes that the common LHS attribute $X$ is removed from all FDs in $\fdset$. If we were to consider S-repairs, we could optimally repair each $R_{x_\ell}$ w.r.t. $\fdset_{-X}$ independently (and recursively), and then take the union of their optimal S-repairs to obtain an optimal S-repair for \relation{} w.r.t. $\fdset$ (as done in \cite{LivshitsKR18}). 
\par
\cut{
    For RS-repairs w.r.t. FD set $\fdset$ and RC $\rc$, we cannot take a simple union of the optimal RS-repairs of $R_{x_\ell}$ w.r.t. $\fdset_{-X}$ and some RC, since repairs $S_{x_\ell}$ of $\relation_{x_\ell}$ will have different distributions of the sensitive attributes $A_m$, and while combining them, we have to ensure that we are satisfying the final RC $\rc$. 
}

Yet, this is not the case for computing RS-repairs, since maximum size is not the only requirement---while we know the final repair will satisfy \rc{}, we do not know what the value distribution of $A_s$ should be in each disjoint piece (e.g., some $\relation_{x_{\ell}}$) of the final repair before we get one. Therefore, in each step of the recursion (either \commonlhsreduction{} here or \consensusreduction{} above), the \candidateset{} preserves all possible distributions of $A_s$ from the S-repairs that could provide the final optimal solution.

\par
\commonlhsreduction{} (\Cref{alg:common_lhs_reduction}) constructs the candidate set $\candset{\relation{}, \fdset}$ recursively from smaller problems by building solutions cumulatively in $n$ stages (the outer loop).
In particular, after stage $\ell$, the algorithm obtains the candidate set $\candset{\relation{}_{x_1, \dots, x_{\ell}}, \fdset{}}$ by combining S-repairs for $\relation{}_{x_1}, \cdots, \relation{}_{x_\ell}$. Note that the union of S-repairs for $\relation{}_{x_1}, \cdots, \relation{}_{x_\ell}$ 
is an S-repair for $\relation{}_{x_1, \dots, x_{\ell}}$ and consequently \relation{} 
w.r.t.~$\fdset$, but we have to ensure that the properties of a candidate set are maintained while combining these S-repairs. Line 3 computes the candidate set $\candset{\relation{}_{x_\ell}, \fdset_{-X}}$ recursively by calling $\reduce(\relation{}_{x_\ell}, \fdset_{-X})$. Since $\candset{\relation{}_{x_1, \dots, x_{\ell-1}}, \fdset{}}$ is already formed in the previous stage, in Lines 4-7, it goes over all combinations of $R' \in \candset{\relation{}_{x_1, \dots, x_{\ell-1}}, \fdset{}}$ and $R'' \in \candset{\relation{}_{x_\ell}, \fdset_{-X}}$, takes their union $R_0 = R' \cup R''$, and inserts it to $\candset{\relation{}_{x_1, \dots, x_{\ell}}, \fdset{}}$ by \reprinsert{} ensuring that the property of a candidate set is maintained. Finally, $\candset{\relation{}_{x_1, \dots, x_{n}}, \fdset{}}$ is returned as the final set $\candset{\relation{}, \fdset}$.

\begin{algorithm}[t]
\caption{\consensusreduction{$(\relation{}, \fdset{}, f)$}}\label{alg:consensus_reduction}
    \begin{algorithmic}[1]
        \Require{a relation \relation{}, a FD set \fdset{}, a consensus FD $f: \emptyset \rightarrow Y$ in \fdset{}}
        \Ensure{A \candidateset{} $\candset{\relation{}, \fdset{}}$}
        \State $\candset{\relation{}, \fdset{}} \gets \emptyset$;
        \For{each value $y_{\ell} \in \dom{(Y)}$} 
        
            \State $\candset{\relation{}_{y_{\ell}}, \fdset{} - f} \gets \reduce{(\relation{}_{y_{\ell}}, \fdset{} - f)}$;\Comment{\Cref{alg:reduce}}
            \For{all $R'$ in $\candset{\relation{}_{y_{\ell}}, \fdset{} - f}$}
                \State $\candset{\relation{}, \fdset{}} \gets \reprinsert{(\candset{\relation{}, \fdset{}}, R')}$; \Comment{\Cref{subsec:candidates}}
            \EndFor
        \EndFor
        \State \Return $\candset{\relation{}, \fdset{}}$.
    \end{algorithmic}
\end{algorithm}
{
\begin{algorithm}[b]
\caption{\commonlhsreduction{$(\relation{}, \fdset{}, X)$}}\label{alg:common_lhs_reduction}
    \begin{algorithmic}[1]
        \Require{A relation \relation{}, a FD set \fdset{}, a \commonlhs{} $X$ for all FDs in \fdset}
        \Ensure{a \candidateset{} $\candset{\relation{}, \fdset{}}$}
        \For{$\ell = 1$ to $n$}\Comment{Suppose $\dom{(X)} = \{x_1, x_2, \dots, x_n\}$}
            \State $\candset{\relation{}_{x_1, \dots, x_{\ell}}, \fdset{}} \gets \emptyset$;\Comment{Initialize a candidate set for $\Delta$ that only considers values $x_1, \cdots, x_\ell$ of $X$}
 
            \State $\candset{\relation{}_{x_{\ell}}, \fdset_{- X}} \gets \reduce{(\relation{}_{x_{\ell}}, \fdset_{- X})}$; \Comment{\Cref{alg:reduce}} 
              \For{all $R'$ in $\candset{\relation{}_{x_1, \dots, x_{\ell - 1}}, \fdset{}}$ and all $R''$ in $\candset{\relation{}_{x_{\ell}}, \fdset_{- X}}$}
                \State{$R_0 \gets R' \cup R''$}; \Comment{Combine prior and current \srepair{}s}
                \State $\candset{\relation{}_{x_1, \dots, x_{\ell}}, \fdset{}}$ $\gets$ $\reprinsert{}(\candset{\relation{}_{x_1, \dots, x_{\ell}}, \fdset{}}, R_0)$; \Comment{\Cref{subsec:candidates}}
            \EndFor
        \EndFor
        \State $\candset{\relation{}, \fdset{}} \gets \candset{\relation{}_{x_1, \dots, x_n}, \fdset{}}$
        \State \Return $\candset{\relation{}, \fdset{}}$.
    \end{algorithmic}
\end{algorithm}
}

\section{Algorithms for the General Case} 
\label{sec:general}
Computing optimal \rsrepair{}s for arbitrary $\fdset$ and $\rc$ is NP-hard (\Cref{sec:hardness}). 
We now present a collection of end-to-end algorithms capable of handling general inputs. We begin with an exact algorithm based on integer linear programming (ILP) and then present a heuristic utilizing LP relaxation and rounding. Next, we present another heuristic using procedures from the previous section for LHS-chains as a subroutine (\Cref{subsec:scalable_heuristic}).

\cut{
In the subsequent sub-sections, we specify the details for each end-to-end-algorithm:
\begin{itemize}
    \item Section~\ref{subsec:ilp}: ILP 
    \item Section~\ref{subsec:lp}: LP relaxation.
    \item Section~\ref{subsec:scalable_heuristic}: a scalable heuristic based on dynamic programmings called \scalableheuristic{}
    \item Section~\ref{subsec:end_to_end_framework}: an end-to-end framework combining the techniques.
\end{itemize}
\fs{@Sudeepa: please check the name. Do we need to rename?}
}
\subsection{LP-based Algorithms}\label{subsec:ilp}
\paratitle{ILP-based Optimial Algorithm and LP-based Heuristic} We use $|\relation{}|$ binary random variables $x_1, x_2, \dots, x_{|\relation{}|}$, where $x_i \in \{0,1\}$ denotes whether tuple $t_i \in \relation$ is retained (if $x_i = 1$) or deleted (if $x_i=0$) in the \rsrepair{}. From the result of the following ILP we take the tuples with $x_i = 1$. We refer to this algorithm as \globalilp{}.

\begin{small}
\hrule
\begin{align}\label{ilp:global}
    \textbf{Maximize } \sum_{i \in [1, |\relation{}|]}\!\!{x_i} \quad &\textbf{subject to:}\\
    x_i + x_j \leq 1 \quad & \text{for all conflicting $t_i$ and $t_j$} \nonumber\\
    \sum_{i: t_i[A_s] = a_\ell}{\hspace{-1em}x_i} ~~\geq~~ p_\ell \times \sum_{i}{x_i} \quad &\text{for all } a_\ell \in \dom(A_s) \nonumber\\
    x_i \in \{0,1\} \quad & \text{for all } i \in [1,|\relation{}|] \nonumber 
\end{align}
\hrule
\end{small}


The objective maximizes the number of tuples retained. The first constraint ensures that the solution does not violate $\fdset{}$. The following set of constraints correspond to the RC $\rc$, by ensuring each lower-bound constraint is satisfied, i.e. $\%a_{\ell} \geq p_{\ell}$, where $p_\ell = \rc(a_\ell)$, is satisfied for every $a_{\ell} \in \dom(A_s)$. 

\cut{
\begin{algorithm}
\caption{\globalilp{}$(\relation{}, \fdset{}, \rc{})$}\label{alg:global_ilp}
    \begin{algorithmic}[1]
        \State Build and optimize ILP in \cref{ilp:global};
        \If{Find a solution $\Vec{x} \in \{0, 1\}^{|\relation{}|}$}
            \State \Return $R' \coloneq \{t_i \mid x_i = 1, \forall i \in [1, |\relation{}|]\}$; 
        \EndIf
        \State\Return $\emptyset$;
    \end{algorithmic}
\end{algorithm}

The algorithm is outlined in \Cref{alg:global_ilp}. After solving the ILP in \Cref{ilp:global}, if a feasible solution is found, the algorithm constructs the repair $R'$ by selecting all tuples $t_i$ such that $x_i=1$ and returns $R'$. If no solution is found, it returns an empty set, indicating that no valid representative repair exists. 
}

The ILP in \Cref{ilp:global} can be relaxed to an LP by replacing the integrality constraints $x_i \in \{0,1\}$ with $x_i \in [0,1]$, for every $i \in [1,|\relation{}|]$.\cut{ It can be solved in polynomial time, but may return fractional variable values.}
We propose rounding procedures to derive an integral solution from fractional $x_i$s and refer to the heuristic as \lprepr{} (pseudocode and running time analysis in \cite{fullversion}). 

\paratitle{Limitations} While \globalilp{} provides an exact optimal solution, its scalability is limited by the size 
of the ILP. And ILP in general is not poly-time solvable. Each pair of tuples ($t_i, t_j$) that conflict on some FD introduces a constraint $x_i + x_j \leq 1$ to the ILP, therefore it can have $O(|\relation{}|^2)$ constraints, leading to a large program that does not scale. 
In \Cref{sec:experiments}, we show that \globalilp{} finds optimal \rsrepair{}s but does not scale to large datasets. For \lprepr{}, we observe in \Cref{sec:experiments} that, even with the state-of-the-art LP solvers, solving our LP can be slow and sometimes encounter out-of-memory issues due to large number of constraints. Hence, we propose a DP-base heuristic using ideas from \Cref{sec:poly-algo-chain} to explore the possibility of avoiding solving the large LP.

\subsection{\scalableheuristic{}: A DP-based Algorithm}\label{subsec:scalable_heuristic}


The combinatorial DP-based \scalableheuristic{} algorithm is motivated by the ideas behind the \commonlhsreduction{} and \consensusreduction{} procedures in Section~\ref{sec:poly-algo-chain}. 
An FD set $\fdset{}$ with only one FD 
can be reduced to the empty set using \dpalgo{} by first applying \commonlhsreduction{} and then \consensusreduction{}, and hence is poly-time solvable by Theorem~\ref{thm:poly-time-chain}.  \scalableheuristic{} therefore calls \dpalgo{} with one FD at a time from $\Delta$ until all FDs are taken care of. Further, it prioritizes the FD which has a most frequent LHS column $X$ (among all the columns) in its LHS. This 
greedy approach may not be optimal as demonstrated empirically in \Cref{sec:experiments}.
\begin{algorithm}
\caption{\scalableheuristic{$(\relation{}, \fdset{}, \rc{})$}}\label{alg:scalable_heuristic}
    \begin{algorithmic}[1]
        \While{$\fdset{}$ is not empty}
            \State Select the most frequent LHS column $X$ in \fdset{};
            \State Choose one arbitrary FD $f$ whose LHS contains $X$;
            \State $\relation{} \gets \dpalgo{(\relation{}, \{f\}, \rc{})}$;
            \State $\fdset{} \gets \fdset{} - f$
        \EndWhile
        \State \Return $\relation{}$;
    \end{algorithmic}
\end{algorithm}
As highlighted, 
Since $\fdset$ is fixed and each call to \dpalgo{} runs in polynomial time for fixed $\dom(A_s)$, \scalableheuristic{} terminates in polynomial time for any $(\relation{}, \fdset{}, \rc{})$ where $\dom{(A_s)}$ is fixed. 
\cut{
Nevertheless, \scalableheuristic{} may not always find the optimal representative repair due to the following reasons: 
\begin{itemize}
    \item Most frequent LHS does not necessarily exhibit the property of being a \commonlhs{}, so the optimization base on $f$ alone may lead to a local optimum.
    \item The locally optimal repair computed by \dpalgo{} for a single FD $f$ may not be globally optimal when considering the entire FD set \fdset{}.
\end{itemize}
Despite these limitations, 
}
\scalableheuristic{} provides a practical and scalable heuristic approach for handling large instances of the problem of computing \rsrepair{}s, by leveraging the efficient subroutine, \dpalgo{}.
Its effectiveness and efficiency will be empirically evaluated in Section~\ref{sec:experiments}.


In our implementation, the heuristics described in \Cref{subsec:scalable_heuristic} first employ \consensusreduction{} and \commonlhsreduction{} until no feasible reductions are possible.
This decomposes the problem into sub-problems.
Then the heuristics are applied to the sub-problems and return a single \rsrepair{} (and consequently a singleton candidate set) for each of them. These candidate sets are later merged during the backtracking stage of reductions. Finally, all the end-to-end algorithms will return a candidate set as what \dpalgo{} does, and rely on \postclean{} to ensure satisfying the RC.

\section{Experiments}\label{sec:experiments}
\cut{
    Our work is driven by the observation that prevailing data repairing methods usually overlook the representation of sensitive attributes, simply because their main objectives typically do not prioritize representation. This observation prompt us to incorporate representation consideration into the procedure of data repairing. 
}
In this section, we evaluate the deletion overhead to preserve representation in subset repairs, 
and the quality and performance of our algorithms. 
In particular, we study the following questions:
\begin{enumerate}[leftmargin=*]
    \item \Cref{subsec:exp_cost_repr}: How many additional tuple deletions are required (i.e., deletion overhead) for computing optimal \rsrepair{}s compared to computing optimal \srepair{}s? 
    \item \Cref{subsec:exp_quality}: How effective is each algorithm in minimizing tuple deletions compared to an optimal \rsrepair{} algorithm (i.e., \rsrepair{} quality)? 
    \item \Cref{subsec:exp_scalability}: What is the runtime cost of our algorithms?
    \item \Cref{subsec:exp_rc}: What is the impact of considering non-exact RCs on the number of deletions, i.e., $\% a \geq p$ instead of $\% a = p$?
\end{enumerate}

\cut{
For example, if the input relation has $80\%$ 'Native' and $20\%$ 'Foreign-born' for the sensitive attribute NATIVITY, the RC would be $\{\%\text{'Native-born'} = \frac{80}{100}, \%\text{'Foreign-born'} = \frac{20}{100}\}$. The sum of the required proportions $=1$.
we set \rc{} according to the value distribution of the sensitive attribute in the input relation.} 
\cut{We gradually lower the required lower-bound proportion for some sensitive values.}

\paratitle{Summary of findings} First, the deletion overhead is high when the noise distribution is imbalanced in the input (more noise in one subgroup), especially when the under-represented group defined by the sensitive attribute has relatively more noise than the majority group.
Second, the DP-based algorithms proposed in this paper, \expscalableheuristic{} for general FD sets (\Cref{sec:general}), and \expdpalgo{} to which the former reduces to for chain FD sets (\Cref{sec:poly-algo-chain}), present the best trade-off of high \rsrepair{} quality and runtime compared to the other alternatives that we have examined. Third, as the constraint on exact RC is relaxed, fewer deletions are needed. 

\subsection{Setup}\label{subsec:exp_basic}
We implement\footnote{Code publicly available at \url{https://github.com/louisja1/RS-repair}.} our algorithms in Python 3.11 and experiment on a machine with a commodity EPYC CPU 
(AMD EPYC 7R13 48-Core Processor @2.6GHz, Boost 3.73GHz, 192MB L3 cache; 256GiB DDR4-3200 memory). 
We use Gurobi\cite{gurobi} as the LP/ILP solver.
\cut{We use pandas.Dataframe\footnote{https://pandas.pydata.org/docs/reference/api/pandas.DataFrame.html} to maintain the repair candidates. Although finer data structures (e.g. bitmap\footnote{https://pypi.org/project/bitmap/}) or dumping/loading candidates to/from can reduce the memory cost from DP, a more memory-friendly implementation for \dpalgo{} goes beyond the scope of our paper and remains as future work.} 


\subsubsection{Datasets}\label{subsubsec:exp_dataset} We use samples of different sizes from three real datasets (ACS, COMPAS, and Flight) that are commonly used for the study of fairness or data repair. For ACS and COMPAS, the noise is injected (\Cref{sec:expt-noise}), whereas for the Flight dataset, the noise is real and inherent to the dataset. 
Additional experiments with a fourth dataset on  Credit Card Transactions \cite{choksi_credit_2023} appear in the full version~\cite{fullversion} due to space limitations.
\begin{itemize}[leftmargin=*]
    \item ACS-PUMS (in short ACS)~\cite{ding2022retiring}: The American Community Survey Public Use Microdata Sample dataset. We use 9 attributes and samples from 2K to 1.5M in our experiments (described below).
    \item COMPAS~\cite{larson2016compas}: The ProPublica COMPAS recidivism dataset. We use 10 attributes and samples from 4K to 30K in our experiments. 
    \item Flight~\cite{li2015truth,boeckling2024cleaning}\footnote{We download and use the version of data from \citet{boeckling2024cleaning}.}: The Flight data contains information on scheduled and actual departure time and arrival time from different sources, and has been used in prior work of data fusion~\cite{li2015truth} and data cleaning~\cite{RekatsinasCIR17,boeckling2024cleaning}. We use 6 attributes and samples from 2K to 8K.
 \end{itemize}

 \subsubsection{Functional dependencies}\label{sec:expt-fd} 
 We consider two types of FD sets in our experiments: (1) chain FD sets (i.e., LHS-chains, \Cref{def:lhs-chain}, \Cref{sec:poly-algo-chain}),  and (2) non-chain FD sets (\Cref{sec:general}). 
 For ACS and COMPAS, FDs are inferred from the data description documents and verified in the original clean datasets \cite{bureau_pums_nodate}. 
 The FD sets used in all three datasets are shown in \Cref{table:fds}.\footnote{$\att{ST, DIV, CIT, RAC1P, RAC2P, POBP}$, and $\att{WAOB}$ are in short for state code, division, citizenship status, race code 1, race code 2, place of birth, world area of birth respectively for dataset ACS. $\att{DOB, RSL}$, and $\att{RSLT}$ are short for date of birth, recommended level of supervision, and its text respectively for dataset COMPAS. $\att{DF}$ is in short for date collected + flight number for dataset Flight.} For example, the FD $\fd{ST}{DIV}$ for ACS implies that each state has a unique division. 
 For the Flight dataset, FDs are from prior work \cite{boeckling2024cleaning}.  
 For ACS and COMPAS, the set of FDs is a non-chain, and we also use a subset of the FDs that forms a chain for these two datasets. For Flight, the set of FDs forms a chain, so only the chain FD set is considered.
 


\begin{table}[t]
\begin{footnotesize}
\caption{FD sets used in experiments}\label{table:fds}

\begin{tabular}{|c|c|c|}
\hline
{\bf Dataset} & \textbf{Chain FD Set} & \textbf{Non-chain FD Set} \\ 
\hline
\textbf{ACS} & \begin{tabular}{@{}l@{}}\{$\fd{ST}{DIV}$,\\$\fd{DIV}{Region}$\}\end{tabular} & \begin{tabular}{@{}c@{}}\{$\fd{CIT}{Nativity}$,$\fd{ST}{DIV} $, \\$\fd{DIV}{Region}$, \\$\fd{POBP}{WAOB}$,$\fd{RAC2P}{RAC1P}$\}\end{tabular}\\ 
\hline
\textbf{COMPAS} & \begin{tabular}{@{}l@{}}\{\textsf{DecileScore} $\rightarrow$\\ \textsf{ScoreText}\}\end{tabular} & 
\begin{tabular}{@{}c@{}}\{$\fd{DecileScore}{ScoreText}$,\\ $\fd{ScaleID}{DisplayText}$,$\fd{RSL}{RSLT}$,\\  $\fd{DisplayText}{ScaleID}$,$\fd{RSLT}{RSL}$,\\ $\fd{FirstName,LastName,DOB}{Sex}$\}\end{tabular}\\ 
\hline

\multirow{2}{*}{\textbf{Flight}} & \multicolumn{2}{c|}{\textbf{Chain FD Set}} \\
\cline{2-3}
& \multicolumn{2}{c|}{\begin{tabular}{@{}l@{}}\{$\fd{DF}{ActualDeparture}$, $\fd{DF}{ActualArrival}$,\\ $\fd{DF}{ScheduledDeparture}$, $\fd{DF}{ScheduledArrival}$\}\\
\end{tabular}} \\

\hline
\end{tabular}
\end{footnotesize}
\end{table}
    

\subsubsection{Sensitive attribute selection}\label{sec:expt-sensitive}
For the ACS dataset, we consider $\att{Nativity}$ as the sensitive attribute (with the values \texttt{Native-born} and \texttt{Foreign-born}) for the non-chain FD set, and $\att{Region}$ (with the values \texttt{Region-1-2} and \texttt{Region-3-4}) for the chain FD set, so that the 
sensitive attribute is always included in the FD set. ($\att{Nativity}$ is not in the chain FD set.) 
For the COMPAS dataset, we use $\att{Sex}$ as the sensitive attribute with values \texttt{Male} and \texttt{Female} that appear in the dataset.
For the Flight dataset, there is no natural sensitive attribute so we choose the $\att{Source}$ of the data as the sensitive attribute with two values \texttt{wunderground} and \texttt{flightview}. We use binary values of the sensitive attribute in our experiments since we vary both the data and noise distribution of two groups (\Cref{sec:expt-noise}), but our algorithms can handle multiple values of the sensitive attribute. 





\subsubsection{Obtaining noisy input data}\label{sec:expt-noise} 
The Flight dataset has violations of the FD, so we only generate uniform random samples of 2K, 4K, 6K, and 8K tuples as the noisy input data. These samples have slightly different value distributions of $\att{Source}$, ranging from $\{\%\texttt{wunderground} = 58\%, \%\texttt{flightview} = 42\%$\} to $\{\%\texttt{wunderground}$ $ = 54\%, \%\texttt{flightview} = 46\%\}$.
\par
The ACS and COMPAS datasets satisfy their FDs, and  we inject random noise (FD violations) 
similarly to prior work~\cite{GeertsMPS13-LLunatic, DallachiesaEEEIOT13, RekatsinasCIR17}.
We vary two parameters: (1) {\em value distribution of the sensitive attribute}, and (2) {\em relative noise distribution of the two groups defined by the sensitive attribute}, as we describe next.

{\em (1) Value distribution of the sensitive attribute for ACS and COMPAS:} We consider binary values of their sensitive attributes denoted by \texttt{Group-1} and \texttt{Group-2}. 
The notation ``X\%-Y\%'' denotes the value distribution of these two groups, i.e., 
$\frac{\%\text{no. of tuples from {\tt Group-1}}}{\%\text{no. of tuples from {\tt Group-2}}} = \frac{X}{Y}$
where  
$X + Y = 100$. 
We generate stratified samples for these two groups in the ACS and COMPAS datasets where the percentages ``X\%-Y\%'' are varied as ``80\%-20\%'', ``60\%-40\%'' to ``50\%-50\%'' for different sizes of the datasets. Therefore, {\tt Group-1} is called the majority group ({\tt Native-born, Region-1-2, Male} in \Cref{sec:expt-sensitive}) having possibly higher percentage of tuples, and {\tt Group-2} ({\tt Foreign-born, Region-3-4, Female} in \Cref{sec:expt-sensitive}) is called the minority group.
\par
{\em (2) Relative noise distribution of the sensitive attribute in ACS and COMPAS: } 
First, we choose an overall noise level. 
We only change the values of the attributes that appear in the LHS or RHS of FDs (\Cref{table:fds}).
$x \%$ noise level means that $x \%$ of the cells (attribute values) belonging to these attributes from all tuples in the data generated in the previous step have been changed to another value from the domain of the corresponding attribute. Then, we choose a relative noise distribution of the values {\tt Group-1} and {\tt Group-2} of the binary sensitive attribute. The distribution ``x\%-y\%'' means that $\frac{\%\text{noise level of {\tt Group-1}}}{\%\text{noise level of {\tt Group-2}}} = \frac{x}{y}$, where $x+y = 100$. The values of $x\%$-$y\%$ are chosen from ``20\%-80\%'', ``40\%-60\%'', ``50\%-50\%'', ``60\%-40\%'', ``80\%-20\%'' to study different levels of noise in the majority and the minority groups (\Cref{subsec:exp_cost_repr}). For example, 
``20\%-80\%'' means that the noise level of {\tt Group-2} (the minority group, e.g., {\tt Foreign-born}, {\tt Female}, etc.) is four times that of {\tt Group-1} (the majority group, e.g., {\tt Native-born}, {\tt Male}, etc.). Since the sensitive attribute values are changed, the data distribution in the previous step of the majority and the minority group may change following the noise injection. 

\subsubsection{Representation constraints}\label{sec:expt-rc}
\cut{For the ACS dataset, we use RCs on sensitive attributes \texttt{NATIVITY} and \texttt{REGION} for the non-chain FD set and chain FD sets, respectively; for the COMPAS dataset, we consider the RC on \texttt{SEX} for both chain and non-chain FD sets. As stated previously, we use exact RCs (\Cref{def:rep_constraint}) in \Cref{subsec:exp_cost_repr,subsec:exp_quality,subsec:exp_scalability}, and we examine general RCs in \Cref{subsec:exp_rc}. The specific proportions of each RC are reported in full version\cite{fullversion}.}
We use exact RCs in \Cref{subsec:exp_cost_repr,subsec:exp_quality,subsec:exp_scalability} to preserve the original distribution of the two groups of the sensitive attribute, and examine RCs with inequality in \Cref{subsec:exp_rc} to preserve the original distribution of the minority group. 

\begin{figure}[t]
\small
\def\figin{$\relation{}, \fdset{}, \rc{}$}
    \centering
    \scalebox{0.8}{
    \input{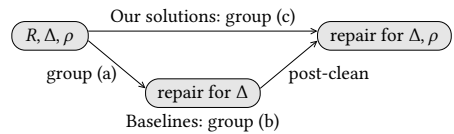}}
    \caption{Our solutions vs. baselines}
    \label{fig:exp_cmp}
    \Description[]{}
\end{figure}

\begin{table*}[!htbp]
\centering
\caption{Deletion overhead for varying representations (80\%-20\% and 50\%-50\%) and relative noise distributions (ACS data, 5\% overall noise, 10k tuples). \cut{The sensitive group distribution "X-Y" means ``$\frac{\%\texttt{Native-born}}{\%\texttt{Foreign-born}} = \frac{X}{Y}$'' for chain FDs set, and ``$\frac{\%\texttt{Region-1-2}}{\%\texttt{Region-3-4}} = \frac{X}{Y}$'' for non-chain FD set. The relative noise distribution "x-y" means $\frac{\%\text{noise of the first group}}{\% \text{noise of the second group}} = \frac{x}{y}$.} ``S-rep. \%'' and ``RS-rep. \%'' stand for deletion percentage of \srepair{} and and for \rsrepair{}, respectively.}
\label{tab:del_overhead_comparison}
\setlength{\tabcolsep}{2pt} 
\footnotesize
\begin{subtable}[!htp]{0.24\textwidth}
    \centering
    \caption{Chain FD set: 80\%-20\%}
    \label{tab:del_overhead_a}
    \begin{tabularx}{\textwidth}{|>{\centering\arraybackslash}p{0.9cm}|XXX|}
        \hline
        noise & del. ratio & S-rep. \% & RS-rep. \% \\
        \hline
        20\%-80\% & 2.112 & 13.52 & 28.55 \\
        40\%-60\% & 1.254 & 13.88 & 17.40 \\
        50\%-50\% & 1.003 & 13.89 & 13.92 \\
        60\%-40\% & 1.044 & 13.91 & 14.52 \\
        80\%-20\% & 1.134 & 13.87 & 15.73 \\
        \hline
    \end{tabularx}
\end{subtable}%
\hfill
\begin{subtable}[!htp]{0.24\textwidth}
    \centering
    \caption{Chain FD set: 50\%-50\%}
    \label{tab:del_overhead_b}
    \begin{tabularx}{\textwidth}{|>{\centering\arraybackslash}p{0.9cm}|XXX|}
        \hline
        noise (\%) & del. ratio & S-rep. \% & RS-rep. \% \\
        \hline
        20\%-80\% & 1.514 & 13.60 & 20.70 \\
        40\%-60\% & 1.152 & 13.80 & 15.90 \\
        50\%-50\% & 1.004 & 13.90 & 13.90 \\
        60\%-40\% & 1.153 & 13.90 & 16.00 \\
        80\%-20\% & 1.511 & 13.60 & 20.60 \\
        \hline
    \end{tabularx}
\end{subtable}%
\hfill
\begin{subtable}[!htp]{0.24\textwidth}
    \centering
    \caption{Non-chain FD set: 80\%-20\%}
    \label{tab:del_overhead_c}
    \begin{tabularx}{\textwidth}{|>{\centering\arraybackslash}p{0.9cm}|XXX|}
        \hline
        noise (\%) & del. ratio & S-rep. \% & RS-rep. \% \\
        \hline
        20\%-80\% & 2.039 & 25.91 & 52.83 \\
        40\%-60\% & 1.197 & 28.65 & 34.30 \\
        50\%-50\% & 1.026 & 28.81 & 29.55 \\
        60\%-40\% & 1.009 & 28.71 & 28.97 \\
        80\%-20\% & 1.199 & 28.02 & 30.80 \\
        \hline
    \end{tabularx}
\end{subtable}
\hfill
\begin{subtable}[!htp]{0.24\textwidth}
    \centering
    \caption{Non-chain FD set: 50\%-50\%}
    \label{tab:del_overhead_d}
    \begin{tabularx}{\textwidth}{|>{\centering\arraybackslash}p{0.9cm}|XXX|}
        \hline
        noise (\%) & del. ratio & S-rep. \% & RS-rep. \% \\
        \hline
        20\%-80\% & 1.423 & 28.63 & 40.74 \\
        40\%-60\% & 1.091 & 30.02 & 32.75 \\
        50\%-50\% & 1.005 & 29.94 & 30.08 \\
        60\%-40\% & 1.071 & 29.56 & 31.64 \\
        80\%-20\% & 1.395 & 27.97 & 39.02 \\
        \hline
    \end{tabularx}
\end{subtable}
\end{table*}

\subsubsection{Algorithms}\label{sec:expt-algo} We classify the repair algorithms into three groups (\Cref{fig:exp_cmp}): (a) \srepair{}s, (b) \srepair{}s with \postclean\ (\Cref{subsec:postclean}) to satisfy the RC, and (c) \rsrepair{} algorithms from \Cref{sec:poly-algo-chain,sec:general}. 
Group (a) includes the following:

\begin{itemize}[leftmargin=*]
    \item \bilp~\cite{LivshitsKTIKR21} (for both chain and non-chain FD sets): Formulates 
    an ILP (with only the first constraint for FD in \Cref{ilp:global}) and computes the {\em optimal} \srepair{} satisfying a given FD set. 
    \item \bapprox~\cite{BARYEHUDA1981198} (for non-chain FD sets only): An 2-approximation algorithm that translates the problem to 
    finding a minimum vertex cover of the conflict graph.
    \item \bdp~\cite{livshits2020computing} (for chain FD sets only): A dynamic programming (DP) algorithm that computes the {\em optimal} \srepair{} in polynomial time when the FD set forms an LHS-chain. 
    \item \bmuse~\cite{GiladDR20} (for both chain and non-chain FD sets): An \srepair{} framework with deletion rules. 
    We use step semantics. 
\end{itemize}

Group (b) is defined as using \postclean{} to post-process the \srepair{}s obtained by the Group (a) approaches. These baselines are denoted as baseline+\bpostclean{}, e.g. \bilp{} $+$ \bpostclean{}. Group (c) in \Cref{fig:exp_cmp} consists of the algorithms proposed in \Cref{sec:poly-algo-chain,sec:general}, which incorporate the RC in their core procedures.

For chain FD sets, we examine the polynomial-time optimal algorithm \expdpalgo{}{} (\Cref{sec:poly-algo-chain}). For non-chain FD sets, we examine the end-to-end heuristic algorithms \cut{\explpgreedy{},}\explprepr{} and \expscalableheuristic{} (\Cref{sec:general}). 
Recall that when applied to chain FD sets, \expscalableheuristic{} reduces to \expdpalgo{}.
We repeat each experiment twice for each data point and take the average.
\subsection{Deletion Overhead of \rsrepair s}\label{subsec:exp_cost_repr}

In this section, we investigate how many additional deletions are required to satisfy the representation of the sensitive attribute by defining {\em deletion overhead} as follows:
\small
\begin{equation}
\text{deletion overhead} = \frac{\ndeleted{\text{an optimal \rsrepair{} of \relation{} w.r.t. \fdset{} and \rc{}}}}{\ndeleted{\text{an optimal \srepair{} of \relation{} w.r.t. \fdset{}}}}
\end{equation}
\normalsize
Here $\ndeleted{}$ is the number of tuple deletions of a repair.
This ratio quantifies the additional deletions required by an optimal \rsrepair{} compared to an optimal \srepair{} for the same $(\relation{}, \fdset{})$. This ratio is at least $1$ as an optimal \rsrepair{} may have to delete more tuples also to satisfy $\rc$. 
A large ratio indicates a high overhead, i.e., many extra deletions are needed to satisfy $\rc$, whereas a ratio close to 1 indicates that an optimal \srepair{} is likely to satisfy $\rc$. We use \bilp{} and \expglobalilp{}  for optimal \srepair{}s and \rsrepair{}s respectively.


\begin{figure}[ht]
    \centering
    \includegraphics[width=\columnwidth]{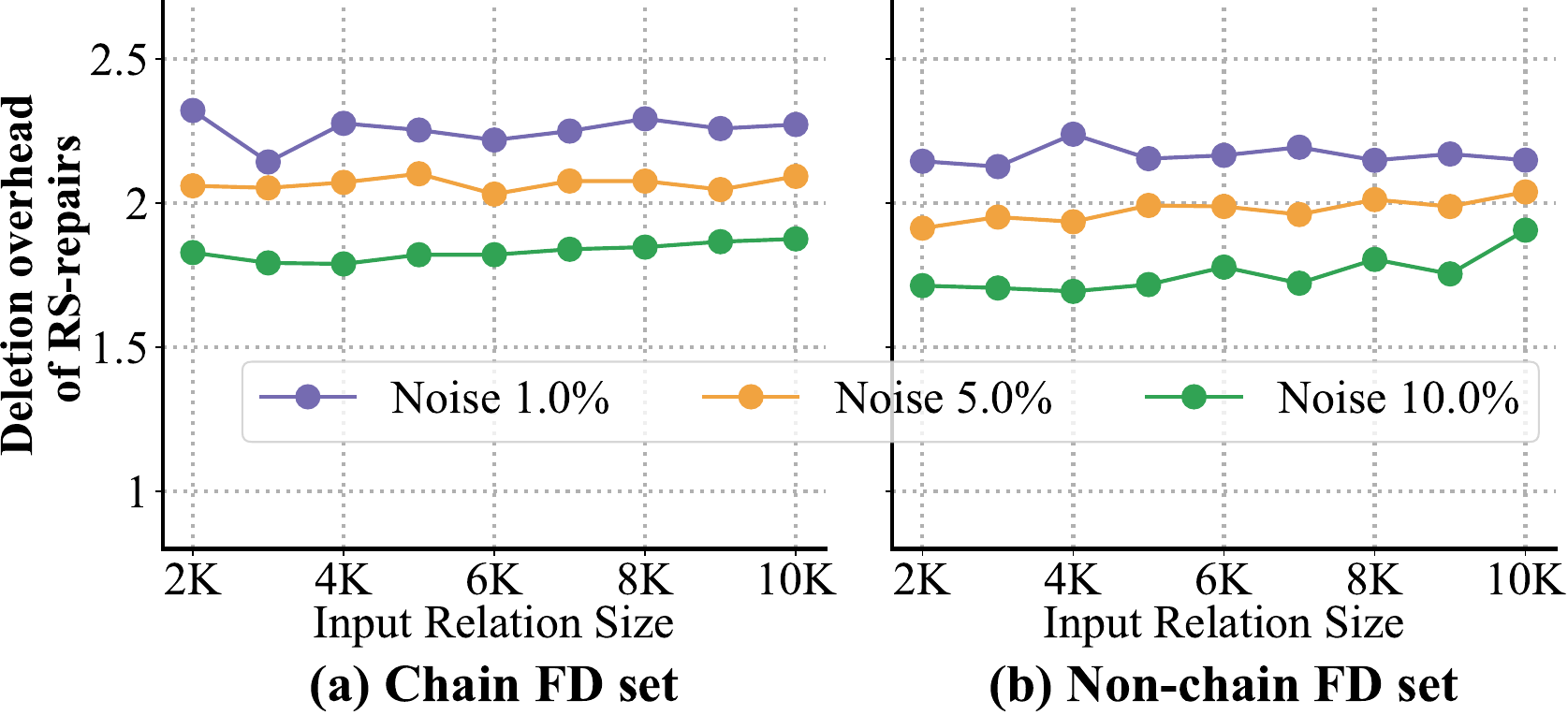}
    \caption{Deletion overhead varying overall noise and input size (ACS, 80\%-20\% value distr., 20\%-80\% relative noise distr.)
    }
    \label{fig:cost_repr_relative}
    \Description[]{}
\end{figure}

\begin{table}[ht]
\centering
\footnotesize
\caption{Deletion overhead and value distribution of sensitive attribute for Flight data with real noise}
\label{tab:del_overhead_flight}
\begin{tabular}{|l|ccc|ccc|}
\hline
\multirow{2}{*}{Size}  & \multicolumn{3}{c|}{Deletion overhead} & \multicolumn{3}{c|}{Value distribution}\\
    \cline{2-7}
      & del. ratio & S-rep. \% & RS-rep. \% & Original & S-rep. & RS-rep.\\
    \hline
    2K & 1.005 & 47.25 & 47.50 & 58\%-42\% & {\bf 69\%-31\%} & 58\%-42\%\\
    4K & 1.006 & 48.20 & 48.50 & 55\%-45\% & {\bf 61\%-39\%} & 55\%-45\%\\
    6K & 1.035 & 48.30 & 50.00 & 54\%-46\% & {\bf 57\%-43\%} & 54\%-46\%\\
    8K & 1.034 & 48.36 & 50.00 & 54\%-46\% & {\bf 65\%-35\%} & 54\%-46\%\\  
    \hline
\end{tabular}

\end{table}

First, we examine the deletion overhead in various scenarios for ACS in \Cref{tab:del_overhead_comparison}.  (Results for COMPAS appear in full version \cite{fullversion} due to space limitations.) \Cref{tab:del_overhead_comparison} also shows the deletion overhead and the percentage of tuples deleted by the optimal \srepair\ and optimal \rsrepair{} for ACS for both chain and non-chain FD sets  (5\% overall noise and 10K tuples). 
We vary the data distribution of the majority and minority groups as 80\%-20\% and 50\%-50\%, and for each vary their relative noise distribution 20\%-80\%, 40\%-60\%, 50\%-50\%, 60\%-40\%, 80\%-20\%, investigating the cases when the minority group has more, equal, and less noise than the majority group. 

We observe the following. (1) When noise is uniform (50\%-50\%), irrespective of the data distribution in \Cref{tab:del_overhead_a,tab:del_overhead_b,tab:del_overhead_c,tab:del_overhead_d}, \rsrepair{} does not delete many additional tuples, hence the deletion overhead is close to 1. As the noise distribution becomes unbalanced, deletion overhead increases as more tuples are deleted for representation.  (2) In \Cref{tab:del_overhead_a,tab:del_overhead_c}, when the data distribution of the majority and minority groups is 80\%-20\%, the overhead for relative noise distribution 20\%-80\% is larger than the overhead for 80\%-20\%. Intuitively, the fraction of noisy tuples from the minority group is much larger requiring more deletions from the minority group compared to the majority group. Hence, many tuples from the majority groups need to be removed to get the data distribution back to 80\%-20\% in the optimal \rsrepair. (3) Although both chain and non-chain FD sets show similar trends on the deletion overhead, the deletion percentages of tuples (S-rep.\% and RS-rep\%) are lower for the chain FD sets. This is because a lower number of attributes is involved in the chain FD set (\Cref{table:fds}), less noise is injected.

Next, in \Cref{fig:cost_repr_relative}, we vary the overall noise level ($1\%, 5\%, 10\%$) and the input size (2K to 10K) for ACS data with 80\%-20\% value distribution and 20\%-80\% relative noise distribution.
The relation size does not significantly impact the deletion overhead of \rsrepair{}s. However, higher noise levels lead to lower deletion overhead, since as noise increases, \srepair{} tends to delete more tuples from both majority and minority groups, almost matching the deletions by \rsrepair{} maintaining their representation, which reduces the ratio. 

\Cref{tab:del_overhead_flight} presents the deletion overhead and percentage for the optimal \srepair{} and \rsrepair{} on Flight, alongside the value distribution of the majority and minority groups before and after repair. (1) Across input sizes (2K to 8K), the deletion overhead remains slightly above 1, indicating significant tuple deletions in both \srepair{} and \rsrepair. (2) Notably, while  \srepair{} does not preserve the original value distribution and disproportionately deletes from the minority group, \rsrepair{} maintains the distribution when the representation constraint is applied.

\begin{table}[t]
\caption{Value distributions by optimal \srepair{}s varying overall noise (ACS, 6K tuples, 20\%-80\% relative noise distr.)}\label{tab:dist_sensitive_attr}
    \centering

\footnotesize
\begin{tabular}{|c|c|c|c|c|c|}
\hline
\multirow{2}{*}{FD Set} & \multirow{2}{*}{Original} & \multicolumn{4}{c|}{Overall noise level}\\
\cline{3-6}
& & 1\% & 5\% & 10\% & 15\% \\
\hline
Chain FD set & 80\%-20\% & 81\%-19\% & 84\%-16\% & 88\%-11\% & 92\%-8\% \\
\hline
Non-chain FD set & 80\%-20\% & 82\%-18\% & 89\%-11\% & 96\%-4\% & 100\%-0\% \\
\hline
\end{tabular}%
\end{table}
In \Cref{tab:dist_sensitive_attr}, 
we report the value distributions of the majority and minority groups after optimal \srepair{} varying the noise level for ACS data with 80\%-20\% value and 20\%-80\% relative noise distribution (one example is in \Cref{ex:intro-distribution}).
At $15\%$ noise level, the 
minority group drops from $20\%$ to $8\%$ for chain FDs, and to $0\%$ for non-chain FDs. This shows that \srepair{}s may significantly change the representation of a sensitive attribute and highlights the importance of considering the representation constraint in the repair algorithms.


\subsection{\rsrepair\ Quality of Various Approaches}\label{subsec:exp_quality}

\begin{table*}[t]
\centering
\caption{Repair quality of different algorithms for ACS and COMPAS data ($10\%$ noise level) with chain and non-chain FD sets. Numbers are boldfaced if they are from \expdpalgo{}, \expscalableheuristic{}, or other algorithms that beat them in the same setup.}
\label{tab:quality_acs_compas}
\footnotesize
\setlength{\tabcolsep}{0pt} 
\renewcommand{\arraystretch}{1} 
\begin{subtable}{\textwidth}
    \centering
    \begin{tabularx}{\textwidth}{|l|*{4}{>{\centering\arraybackslash}X}| *{4}{>{\centering\arraybackslash}X}|*{4}{>{\centering\arraybackslash}X}|}
    \hline
    \multicolumn{13}{|c|}{\cellcolor{gray!50}\textbf{(a) ACS: chain FD set}}\\\hline
    {\bf Sensitive attribute distribution} & \multicolumn{4}{c|}{\textbf{80\%-20\%}} & \multicolumn{4}{c|}{\textbf{60\%-40\%}} & \multicolumn{4}{c|}{\textbf{50\%-50\%}} \\
    \hline
    \textbf{Algorithm / Size (K)} & \textbf{4} & \textbf{6} & \textbf{8} & \textbf{10} & \textbf{4} & \textbf{6} & \textbf{8} & \textbf{10} & \textbf{4} & \textbf{6} & \textbf{8} & \textbf{10} \\
    \hline
    \expglobalilp{} & 100 & 100 & 100 & 100 & 100 & 100 & 100 & 100 & 100 & 100 & 100 & 100 \\
    \rowcolor{gray!30} \textbf{\expdpalgo{} (= \expscalableheuristic{})} & \textbf{100} & \textbf{100} & \textbf{100} & \textbf{100} & \textbf{100} & \textbf{100} & \textbf{100} & \textbf{100} & \textbf{100} & \textbf{100} & \textbf{100} & \textbf{100} \\
    \bilp{}+\bpostclean{} & 76.54 & 79.57 & 80.60 & 80.90 & 96.95 & 97.60 & 97.52 & 97.63 & 99.79 & 99.93 & 99.98 & 99.93 \\
    \bdp{}+\bpostclean{} & 76.54 & 79.57 & 80.60 & 80.90 & 96.95 & 97.60 & 97.52 & 97.63 & 99.79 & 99.93 & 99.98 & 99.93 \\
    \bmuse{}+\bpostclean{}& 77.80 & 80.64 & - & - & 96.68 & 65.54 & - & - & 99.66 & 99.91 & - & - \\ 
\hline
\end{tabularx}
\end{subtable}
\begin{subtable}{\textwidth}
    \centering
    \begin{tabularx}{\textwidth}{|l|*{4}{>{\centering\arraybackslash}X}| *{4}{>{\centering\arraybackslash}X}|*{4}{>{\centering\arraybackslash}X}|}
    \hline  
    \multicolumn{13}{|c|}{\cellcolor{gray!50}\textbf{(b) ACS: non-chain FD set}}\\\hline
    {\bf Sensitive attribute distribution} & \multicolumn{4}{c|}{\textbf{80\%-20\%}} & \multicolumn{4}{c|}{\textbf{60\%-40\%}} & \multicolumn{4}{c|}{\textbf{50\%-50\%}} \\
    \hline
    \textbf{Algorithm / Size (K)} & \textbf{4} & \textbf{6} & \textbf{8} & \textbf{10} & \textbf{4} & \textbf{6} & \textbf{8} & \textbf{10} & \textbf{4} & \textbf{6} & \textbf{8} & \textbf{10} \\
    \hline
    \expglobalilp{} & 100 & 100 & 100 & 100 & 100 & 100 & 100 & 100 & 100 & 100 & 100 & 100 \\
    \rowcolor{gray!30} \textbf{\expscalableheuristic{}} & \textbf{94.25} & \textbf{95.41} & \textbf{96.19} & \textbf{99.14} & \textbf{94.45} & \textbf{96.05} & \textbf{96.12} & \textbf{95.60} & \textbf{96.45} & \textbf{97.22} & \textbf{96.65} & \textbf{97.83} \\
    \explprepr{} & 64.94 & 66.39 & 69.25 & 79.14 & 81.33 & 82.26 & 75.33 & 87.75 & \textbf{98.15} & 90.74 & 90.21 & 84.29 \\
    \bilp{}+\bpostclean{} & 29.60 & 45.57 & 46.78 & 73.62 & 84.03 & 86.03 & 85.79 & 85.79 & 95.7 & 96.30 & \textbf{97.57} & \textbf{98.10} \\
    \bapprox{}+\bpostclean{} & 0 & 0 & 0 & 0 & 5.00 & 4.14 & 3.81 & 4.51 & 18.15 & 15.73 & 13.93 & 15.11 \\
    \hline
    \end{tabularx}
\end{subtable}
\begin{subtable}{\textwidth}
    \centering
    \begin{tabularx}{\textwidth}{|l|*{4}{>{\centering\arraybackslash}X}|*{4}{>{\centering\arraybackslash}X}|*{3}{>{\centering\arraybackslash}X}|}
    \hline
    \multicolumn{12}{|c|}{\cellcolor{gray!50}\textbf{(c) COMPAS: chain FD set}}\\\hline
    {\bf Sensitive attribute distribution} & \multicolumn{4}{c|}{\textbf{80\%-20\%}} & \multicolumn{4}{c|}{\textbf{60\%-40\%}} & \multicolumn{3}{c|}{\textbf{50\%-50\%}} \\
    \hline
    \textbf{Algorithm / Size (K)} & \textbf{4} & \textbf{10} & \textbf{20} & \textbf{30} & \textbf{4} & \textbf{10} & \textbf{20} & \textbf{30} & \textbf{4} & \textbf{10} & \textbf{20}\\
    \hline
    \expglobalilp{} & 100 & 100 & 100 & 100 & 100 & 100 & 100 & 100 & 100 & 100 & 100 \\
    \rowcolor{gray!30} \textbf{\expdpalgo{} (= \expscalableheuristic{})} & \textbf{100} & \textbf{100} & \textbf{100} & \textbf{100} & \textbf{100} & \textbf{100} & \textbf{100} & \textbf{100} & \textbf{100} & \textbf{100} & \textbf{100} \\
    \bilp{}+\bpostclean{} & 97.96 & 98.60 & 99.19 & 99.01 & 99.92 & 100 & 100 & 100 & 100 & 100 & 100 \\
    \bdp{}+\bpostclean{} & 97.96 & 98.60 & 99.19 & 99.01 & 99.92 & 100 & 100 & 100 & 100 & 100 & 100 \\
    \bmuse{}+\bpostclean{}& 98.93 & - & - & - & - & - & - & - & - & - & - \\ 
    \hline
    \end{tabularx}
\end{subtable}
\begin{subtable}{\textwidth}
    \centering
    \begin{tabularx}{\textwidth}{|l|*{4}{>{\centering\arraybackslash}X}|*{4}{>{\centering\arraybackslash}X}|*{3}{>{\centering\arraybackslash}X}|}
    \hline
    \multicolumn{12}{|c|}{\cellcolor{gray!50}\textbf{(d) COMPAS: non-chain FD set}}\\\hline
    {\bf Sensitive attribute distribution} & \multicolumn{4}{c|}{\textbf{80\%-20\%}} & \multicolumn{4}{c|}{\textbf{60\%-40\%}} & \multicolumn{3}{c|}{\textbf{50\%-50\%}} \\
    \hline
    \textbf{Algorithm / Size (K)} & \textbf{4} & \textbf{10} & \textbf{20} & \textbf{30} & \textbf{4} & \textbf{10} & \textbf{20} & \textbf{30} & \textbf{4} & \textbf{10} & \textbf{20} \\
    \hline
    \expglobalilp{} & 100 & 100 & 100 & 100 & 100 & 100 & 100 & 100 & 100 & 100 & 100\\
    \rowcolor{gray!30} \textbf{\expscalableheuristic{}} & \textbf{98.94} & \textbf{94.21} & \textbf{97.39} & \textbf{96.45} & \textbf{95.99} & \textbf{96.90} & \textbf{97.39} & \textbf{96.37} & \textbf{97.11} & \textbf{99.24} & \textbf{99.19} \\
    \explprepr{} & 98.58 & \textbf{97.92} & \textbf{98.04} & \textbf{98.44} & \textbf{99.61} & \textbf{99.54} & \textbf{98.25} & \textbf{98.02} & \textbf{100} & \textbf{100} & \textbf{100} \\
    \bilp{}+\bpostclean{} & 96.81 & \textbf{94.36} & 90.01 & 85.36 & \textbf{99.10} & \textbf{98.32} & 96.95 & 95.86 & \textbf{99.55} & 99.17 & \textbf{99.28}\\
    \bapprox{}+\bpostclean{} & 26.24 & 25.96 & 22.81 & 20.85 & 9.95 & 9.81 & 8.73 & 9.01 & 11.12 & 10.40 & 8.91 \\
    \hline
    \end{tabularx}
\end{subtable}
\end{table*}


The algorithm \expglobalilp{} computes an optimal \rsrepair{},
but does not scale (e.g., it took $18$ hours to repair ACS data for the non-chain FD set with $10\%$ noise and 10K tuples). 
In this section, we report the \textit{\rsrepair{} quality} ({\em repair quality} in short) defined below to assess how our proposed algorithms perform 
compared to the optimal \rsrepair. 
The experiments are conducted on relatively small datasets where  \expglobalilp{} terminated in a reasonable time. 
\cut{To find out an alternative algorithm when \expglobalilp{} does not scale, }
\small
\begin{equation}\label{eq:repair-quality}
\text{repair quality} =\frac{|\relation{}| - \ndeleted{\relation'}}{|\relation{}| - \ndeleted{\text{an optimal \rsrepair{} of \relation{} w.r.t. \fdset{} and \rc{}}}}
\end{equation}
\normalsize

The expression compares the number of tuples retained by an \rsrepair{} $R'$ to that of an optimal \rsrepair{} (via \expglobalilp{}), indicating how well $R'$ approximates the optimum.
This quality is upper-bounded by $100\%$, 
and a ratio close to $100\%$ indicates high quality when the output of an \rsrepair{} algorithm is close to optimal.

In \Cref{tab:quality_acs_compas},  we examine the repair quality for varying value distributions of the sensitive attribute and FD types for ACS and COMPAS data. Since there are many potential scenarios, we present a set of results (more results are in the full version~\cite{fullversion}) for overall 10\% noise level, value distributions of the majority and minority groups X\%-Y\% (``80\%-20\%'', ``60\%-40\%'', and ``50\%-50\%'') with relative noise distributions Y\%-X\% (i.e., the same number of cells are changed in both majority and minority groups). 
The results by \expscalableheuristic{} proposed in \Cref{subsec:scalable_heuristic} for general FD sets, which is identical to the optimal \expdpalgo{} in \Cref{sec:poly-algo-chain}, are {\bf boldfaced} in all tables. The repair quality by other algorithms that are better than \expscalableheuristic{} (non-chain) or \expdpalgo{} (chain) is also boldfaced. 
\par
We conclude the following.
(1) \expdpalgo{} (=\expscalableheuristic{}) is optimal for chain FDs, so has 100\% repair quality in all settings for chain FDs in ACS and COMPAS (\Cref{tab:quality_acs_compas} (a) and (c)). (2) For non-chain FDs, \expscalableheuristic{} has consistently high repair quality ($> 94\%$) in all settings for both ACS and COMPAS  (\Cref{tab:quality_acs_compas} (b) and (d)). (3) For the baselines that apply \postclean{} after an \srepair{} is obtained (i.e., group (b) in \Cref{fig:exp_cmp}), repair quality significantly varies in different scenarios. They perform relatively better for chain FD sets where \expdpalgo{} already gives optimal results (\bmuse{}+\bpostclean{} exceeded 12 hours for larger data denoted by ``-''), but may give poor quality for non-chain FD sets. For both chain and non-chain FDs, their quality mostly degrades when the value distribution of the sensitive attribute is imbalanced. While \bilp{}+\bpostclean{} has better quality than \expscalableheuristic{} in some non-chain FD settings in \Cref{tab:quality_acs_compas} (b) and (d), ILP is not a scalable method (\Cref{subsec:exp_scalability}). The efficient approximation  \bapprox{}+\bpostclean{} has poor quality. (4) The other LP-rounding-based algorithm \explprepr{} from \Cref{subsec:ilp} also performs better than \expscalableheuristic{} in some settings in \Cref{tab:quality_acs_compas} (b) and (d), especially when the value distribution of the sensitive attribute is close to 50\%-50\%, but is not scalable (\Cref{subsec:exp_scalability}). \Cref{tab:quality_acs_compas} shows that \expscalableheuristic{} gives high quality across datasets and settings with better scalability as shown in \Cref{subsec:exp_scalability}.

\begin{table}[!ht]
\caption{Repair quality for Flight data (chain FD set)}\label{tab:flight_quality}
    \centering
\footnotesize
\begin{tabular}{|l|c|c|c|c|}
\hline
\textbf{Algorithm / Size (K)} & \textbf{2} & \textbf{4} & \textbf{6} & \textbf{8} \\
\hline
\expglobalilp{} & 100 & 100 & 100 & 100 \\
{\bf \expdpalgo{}(= \expscalableheuristic{})} & {\bf 100} & {\bf 100} & {\bf 100} & {\bf 100} \\
\bdp{}+\bpostclean{} & 57.14 & 77.67 & 66.67 & 2.93 \\
\bilp{}+\bpostclean{} & 52.38 & 70.87 & 66.67 & 50.00 \\

\hline
\end{tabular}
\end{table}

\Cref{tab:flight_quality} presents the repair quality for the Flight dataset.
\expglobalilp{} and \expdpalgo{} provide $100\%$ repair quality consistently as expected, while the baselines \bdp{} + \bpostclean{} and \bilp{}+\bpostclean{} have a significantly lower quality across all input sizes (e.g., for $8K$, the quality of \bdp{}+\bpostclean{} is $3.5\%$ and the quality of \bilp{}+ \bpostclean{} is $60\%$).
\subsection{Running Time Analysis}\label{subsec:exp_scalability}
\paragraph{Comparison of runtime of different \rsrepair{} methods:} First, we evaluate the runtime performance of the \rsrepair{} algorithms from groups (b) and (c) in \Cref{fig:exp_cmp}. 
\Cref{fig:time_acs_compas_10} shows the runtime results of the ACS and COMPAS for 80\%-20\% value distribution of the sensitive attribute and 10\% overall noise level (runtime for more settings and Flight are in the full version~\cite{fullversion}). If no data points are shown, execution took longer than 12 hours.  (1) The optimal \expglobalilp{} for \rsrepair{} has a high runtime in all settings as ILP is not scalable. (2) For chain FD sets in \Cref{fig:time_b,fig:time_f} the optimal \expdpalgo{} (= \expscalableheuristic{}) is much faster than \expglobalilp{}, e.g., \expglobalilp{} takes $1,046$s for size 10K ACS data, while \expdpalgo{} takes only $12$s, (3) For non-chain FD sets in \Cref{fig:time_d,fig:time_h}, \expscalableheuristic{} has a good scalability. While \bapprox{}+\bpostclean{} has a slightly better runtime than \expscalableheuristic{} in \Cref{fig:time_f}, \Cref{tab:quality_acs_compas} shows its poor quality.
Additionally, for COMPAS and non-chain FD set () in size 30K due to a lagged rounding step. 

\begin{table}[b]
\centering 
\footnotesize
\caption{Runtime (mins) for ACS ($5\%$ overall noise, 80\%-20\% value distr. and 20\%-80\% relative noise distr.). ``OOM'' =  out-of-memory issues.}
\label{tab:acs_time_largescale}

\begin{tabularx}{\columnwidth}{l|rrrrrr}
\toprule
\multicolumn{6}{c}{\textbf{Chain FD set}} \\
\midrule
Name/DB Size & 10K & 100K &500K & 1M & 1.5M\\
\midrule
\bdp{}+\bpostclean{} & 0.01 & 0.02 & 0.07 & 0.14 & 0.20 \\
\hline
\textbf{\expdpalgo{}} & 0.14 & 3.52 & 84.76 & \textbf{479.50} & \textbf{2873.09} \\
\midrule
\multicolumn{6}{c}{\textbf{Non-chain FD set}} \\
\midrule
Name/DB Size & 10K & 100K & 500K & 1M & 1.5M \\
\midrule
\bapprox{}+\bpostclean{} & 0.29 & OOM & OOM & OOM & OOM\\
\hline
\explprepr{} & 3.18 & OOM & OOM & OOM & OOM\\
\hline
\textbf{\expscalableheuristic{}} & 0.15 & 10.56 & 142.80 & \textbf{351.79} & \textbf{2049.64}\\
\bottomrule
\end{tabularx}
\end{table}

{\em Scalability for bigger data:} Since ILP is known to be unscalable, in \Cref{tab:acs_time_largescale} we evaluate the scalability of the other methods over bigger samples up to 1.5M\cut{1M} tuples of the ACS dataset.
We see that \expscalableheuristic{} for non-chain FD sets, which is identical to the optimal \dpalgo{} for chain FD sets, can repair larger datasets, while the other methods face out-of-memory issues. While \bdp{}+\bpostclean{} has much better scalability for chain FD sets, as discussed in \Cref{subsec:exp_quality}, it suffers from poor quality. 
For the chain FD set, \expdpalgo{} provides the optimal \rsrepair{} in $8$ hours for 1M tuples and in $48$ hours for 1.5M tuples. For the non-chain FD set, \expscalableheuristic{} takes 5.9 hours for 1M tuples, and 34 hours for 1.5M tuples. Although it takes a long time to repair the data, it returns results, which may be permissible for offline data repair tasks. In contrast, \bapprox{} $+$ \bpostclean{} and \explprepr{} encounter issues with memory usage over 100K tuples, which is as expected given that the number of constraints can be as large as $|R|^2$. 
Yet, both \expdpalgo{} and \expscalableheuristic{} are DP-based algorithms that require significant space and do not scale well. Developing more space- and time-efficient algorithms for \rsrepair{} remains a promising direction for future research. 

\subsection{Varying Representation Constraints}\label{subsec:exp_rc}
In this section, we vary the RCs to study the effect on the number of tuples remaining in the database after repair in ACS data. We relax the RC by gradually reducing the proportion of the minority group 
from the original ``$= 20\%$''
to ``$\geq 18\%$'' and `` $\geq15\%$''. 


\cut{
    in repair quality and the percentage of remaining tuples by the proposed methods and baselines that when the constraint to the minority group "Foreign-born" is relaxed to \%Foreign-born $\geq 18\%$ and $15\%$. We consider two scenarios: (1) chain \fdset{} with $20\%$ noise and $8000$ tuples; (2) non-chain \fdset{} with $10\%$ noise and $4000$ tuples.
}

\begin{figure}[t]
\centering
\begin{minipage}{1\linewidth}
\centering
\begin{subfigure}{0.5\linewidth}
\includegraphics[width=\linewidth]{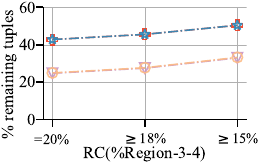}
\centering
\caption{Chain (8K tuples, 15\% noise)}
\label{fig:relax_chain}
\end{subfigure}%
\begin{subfigure}{0.5\linewidth}
\includegraphics[width=\linewidth]{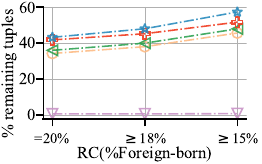}
\centering
\caption{Non-chain (4K tuples, 5\% noise)}
\label{fig:relax_nonchain}
\end{subfigure}%
\end{minipage}
\centering
\begin{minipage}{1\linewidth}
\centering
\includegraphics[width=1\linewidth]{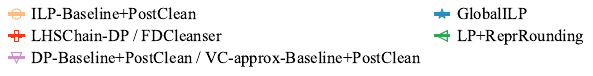}
\end{minipage}
\caption{Percentage of remaining tuples varying the RC (ACS, 80\%-20\% value distr. and 20\%-80\% relative noise distr.)}
\label{fig:relax_rc_remain}
\Description[]{}
\end{figure}

\begin{figure*}[t]
    \centering
    \begin{minipage}{\textwidth}
        \begin{subfigure}{0.25\textwidth}
        \captionsetup{aboveskip=0pt}
            \includegraphics[width=\textwidth]{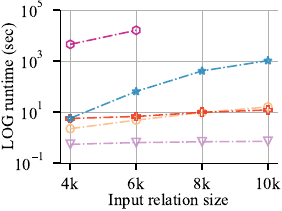}
            \caption{ACS: chain FD set (10\%)}
            \label{fig:time_b}
        \end{subfigure}%
        \begin{subfigure}{0.25\textwidth}
        \captionsetup{aboveskip=0pt}
            \includegraphics[width=\textwidth]{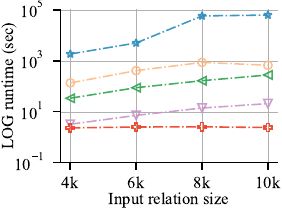}
            \caption{ACS: non-chain FD set (10\%)}
            \label{fig:time_d}
        \end{subfigure}%
        \begin{subfigure}{0.25\textwidth}
        \captionsetup{aboveskip=0pt}
            \includegraphics[width=\textwidth]{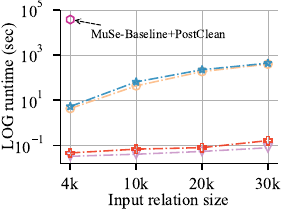}
            \caption{COMPAS: chain FD set (10\%)}
            \label{fig:time_f}
        \end{subfigure}%
        \begin{subfigure}{0.25\textwidth}
        \captionsetup{aboveskip=0pt}
            \includegraphics[width=\textwidth]{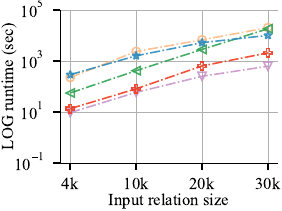}
            \caption{COMPAS: non-chain FD set (10\%)}
            \label{fig:time_h}
        \end{subfigure}%
    \end{minipage}
    \begin{minipage}{0.9\textwidth}
        \centering
        \includegraphics[width = 0.9\textwidth]{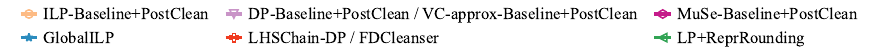}
    \end{minipage}
    \caption{Runtimes for ACS and COMPAS data ($10\%$ noise level, 80\%-20\% value distribution) with chain and non-chain FD sets}
    \label{fig:time_acs_compas_10}
    \Description[]{}
\end{figure*}

\Cref{fig:relax_rc_remain} shows that for both chain and non-chain FD sets, the percentage of remaining tuples 
increases as the 
constraint is relaxed.
The percentage of tuples deleted for non-chain FDs is high even for 5\% noise, since the FDs have 9 attributes (\Cref{table:fds}), 
possibly changing a large number of tuples in noise injection 
(\Cref{sec:expt-noise}), and the 
20\%-80\% 
relative noise has high cost of \rsrepair{} 
(\Cref{subsec:exp_cost_repr}).
\cut{
\revb{Moreover, the RCs force us to delete additional tuples from the majority group to satisfy the required representations, resulting in a substantial percentage of deletions.}
In the end, although algorithms may delete fewer tuples by relaxing RCs, there is a trade-off between the strictness of the RC and the required number of deletions.
}

\section{Related Work}\label{sec:related}


\paragraph*{Data repairing and constraints}
Past work on data repairing~\cite{Afrati2009,RekatsinasCIR17,ChuIP13,BertossiKL13,FaginKK15} aimed at minimizing the number of changes in the database in terms of tuple deletions~\cite{Afrati2009,livshits2020computing,GiladDR20,ChomickiM05,LopatenkoB07,10.14778/3407790.3407809} and value updates~\cite{ChuIP13,RekatsinasCIR17,DallachiesaEEEIOT13,GeertsMPS13-LLunatic,KolahiL09,BertossiKL13}. 
The complexity of the former is better understood~\cite{livshits2020computing}.
We intend to study extensions of our approach to the value-change model as well in future work.
%
The literature 
considered a wide variety of constraints, such as FDs and versions thereof \cite{bohannon2007conditional,koudas2009metric}, denial constraints \cite{ChomickiM05}, tuple-generating dependencies \cite{FaginKMP05,Fagin09a}, and equality-generating dependencies \cite{DBLP:journals/siamcomp/BeeriV84}. 
Different from these, our representation constraints considers the statistical proportions of an attribute {\em over the entire database.}

{\em Representation constraints.} 
Various types of constraints involving probability distributions on the data have been proposed in the literature. This vein of literature focuses on fairness in light of {\em a specific downstream task}.
A predominant part of the work there focuses on {\em algorithmic fairness}~\cite{FeldmanFMSV15,CalmonWVRV17,SalimiRHS19,SalazarNA21,pitoura2022fairness}: given a classification algorithm, and sensitive attributes, determine whether the classification is {\em fair} according to some definition of {\em fairness} \cite{HardtPNS16,equal1990uniform}. 
Specifically, they consider the outcome or predicted attribute in the data and some sensitive attributes (like race or sex), and aim to maintain some equality of conditional probabilities for different values of these attributes. 
Several fairness definitions can be expressed as comparisons of (combinations of) conditional probabilities, 
e.g.,  demographic parity~\cite{DworkHPRZ12,KleinbergMR17} and true positive rate balance~\cite{simoiu2017problem,Chouldechova17}. 
One of our hypotheses is that maintaining representations of different sub-populations 
while repairing the data is a precursor to reducing biases in all data-dependent tasks. Yet, capturing some fairness definitions using complex RCs is an intriguing subject for future work.
\section{Discussion and Future Work}\label{sec:conclusion}


In this paper, we proposed a novel framework for estimating the cost of representation for \srepair{}s, i.e., how many extra tuples have to be deleted in order to satisfy both the FDs as well as a representation constraint (RC) on the values of a sensitive attribute. We studied the complexity of computing an optimal \rsrepair{}, presented poly-time optimal algorithms for FD sets that form LHS-chains for bounded domain of the sensitive attribute, devised efficient heuristics for the general cases, and evaluated our approaches experimentally. 
This paper conducts an initial study of data repair with representations using \srepair, and there are many interesting future directions.
\par
First, one can study the complexity and algorithms for generalization of representations to {\em multiple sensitive attributes with arbitrary overlap with the FDs}.
While the problem will remain NP-hard and ILP can give a solution for smaller data and some settings, a closer study of the complexity, algorithms (e.g., the \postclean{} process), and impact in practice will be challenging. 
We have a detailed discussion on multiple sensitive attributes in the full version \cite{fullversion}.
\par
Second, it will be interesting to study the cost of representations for other repair models, and in particular, for  {update repairs}. 
Update repair cleans the data by updating values of some cells (attributes of tuples), which preserves the original data size unlike \srepair. However, update repair approaches also have their own challenges with or without representations. In  \Cref{ex:intro-distribution} discussed in the introduction, a popular update repair method  Holoclean~\cite{RekatsinasCIR17} does not make any changes in the data, and therefore does not eliminate any violations, since it treats the FDs as soft constraints while preserving the data distribution. On the other hand, another update repair method Nadeef~\cite{DallachiesaEEEIOT13} provides a repair satisfying the FDs, and since the sensitive attribute $\att{Disability}$ does not participate in the FDs, it preserves its representations. However, in our experiment, Nadeef changed a {\tt (Asian, foreign-born)} person to {\tt (White, native)} while keeping the other attribute values the same {\tt (female, born in Philippines, lives in CA, DISABLE, $\cdots$)}. This is not faithful to the reality (People born in Philippines should be {\tt ``foreign born''} not {\tt ``native''} and are likely to be Asian (not White)). Data repair under updates might be inherently connected to preserving distributions of multiple attributes discussed above.

\par
Third, this paper considers the cost of repair preserving representation on the data without considering any tasks where the data is used. A natural extension of our model and a future work is to consider the (still task-independent) weighted cost of tuple deletion where different tuples have different weights or different costs, which can be helpful to the downstream tasks after repairing. 
A more challenging future work is to study both the cost and effect of data repair with and without representation on {\em downstream tasks}, e.g., when the repaired data is used for ML-based tasks like predictions and classification. It will be interesting to see the implication of preserving representations in the input dataset to preserving fairness of the ML tasks for different sub-populations.
\begin{acks}
 The work of Benny Kimelfeld was funded by the Israel Science Foundation (ISF) under grant 768/19 and the German Research Foundation (DFG) under grant KI 2348/1-1. The work of Amir Gilad was funded by the Israel Science Foundation (ISF) under grant 1702/24 and the Alon Scholarship.
\end{acks}

\clearpage

\bibliographystyle{ACM-Reference-Format}
\bibliography{bibtex}

\clearpage

\appendix
\section{Details in Section~\ref{sec:model}}
We present the missing details of \Cref{sec:model} 
in this section in the following order:
\begin{itemize}
    \item Proof of \Cref{thm:hardness}  (\Cref{sec:repr-sensitive}),
    \item Proof of \Cref{prop:lhs-chain-reduction} (\Cref{sec:algo-reprepair}),
    \item Simple LHS marriage is at least as hard as color-balanced matching (\Cref{sec:algo-reprepair})
    \item Pseudocode of \postclean{} (\Cref{subsec:postclean}),
    \item Proof of \Cref{prop:post_clean} (\Cref{subsec:postclean}).
\end{itemize}
The proof of \Cref{thm:poly-time-chain} is given in \Cref{subsec:proof_thm_poly_time_chain}, since it uses the lemmas proved in \Cref{sec:appendix_polyalgo}.

\subsection{Proof of Theorem~\ref{thm:hardness}}
\hardnessthm*
\begin{proof} 
    Recall the decision problem of computing an optimal \rsrepair{}. Suppose \schema{} and \fdset{} are fixed. Given a relation \relation{}, a RC \rc{}, and a non-negative integer $T$, the answer is yes if and only if there exists an \rsrepair{} that retains at least $T$ tuples.
    \par
    We show a reduction from 3-satisfiability (3-SAT) to the decision problem of computing the optimal \rsrepair{}.
    
    Consider an input $\phi$ of 3-CNF with clauses $c_1, c_2, \cdots, c_m$ and variables $v_1, v_2, \cdots, v_n$. The goal is to decide if there is an assignment of each variable to $\{0, 1\}$ (i.e., false or true) so that $\phi$ evaluates to true.

    We construct an instance for our decision problem as follows. The relation $\relation{(A, B, C)}$ contains three attributes $A, B, C$ where the sensitive attribute is $C$. For a variable $v_i$ and clause $c_j$, if the positive literal $v_i$ appears in $c_j$ we add a tuple $(v_i, 1, c_j)$ to \relation{}; otherwise if the negative literal $\overline{v_i}$ appears in $c_j$, we add a tuple $(v_i, 0, c_j)$ to \relation{}. Overall, there are $3m$ tuples in \relation{}. The FD set $\fdset{} = \{A \rightarrow B\}$. The RC $\rc{}$ contains the lower bounds $\%c_j \geq \frac{1}{m}$, for all $c_j \in \dom(C)$. 
    Note that $\rc$ is an exact constraint since there are $m$ clauses $c_j$ and $m \cdot \frac{1}{m} = 1$, i.e., $\%c_{j} = \frac{1}{m}$ must hold for every $C$-value if the repair satisfies \rc{}.
    We set $T = m$.  
    
    We show that $\phi$ has a satisfying assignment if and only if \relation{} has a R\srepair{} $R'$ for $\Delta$ and $\rc$ whose size is at least $T = m$.

    \paragraph*{(If)} Consider any \rsrepair{} $R'$ of $R$ w.r.t. $\Delta, \rc$ such that $|R'| \geq m$.
    $R'$ satisfies \rc{}, so there are at least $\floor{\frac{|R'|}{m}} \geq 1$ tuples in $R'$ such that $t[C] = c_{j}$ for each distinct $c_{j}$. Pick any such tuple $t$. If $t = (v_i, u, c_j)$, we assign the Boolean value $u$ to $v_i$. 
    Since $R'$ satisfies $\fdset{}$, so the assignment to every variable $v_i$ is unique. The variables that are not retained by $R'$ can be assigned to either $0$ or $1$. Such assignment of the variables is a satisfying assignment for $\phi$, since clause $c_j$ is satisfied for every $j \in [1, m]$, and consequently $\phi$ evaluates to true.

    \paragraph*{(Only if)} Suppose that we have a satisfying assignment for variables $v_1, \dots, v_n$ so that $\phi$ evaluate to true, represented by Boolean values $u_1, \dots, u_n$. We form a relation $R'$ with tuples of the form $(v_i, u_i, *)$ where $*$ can take any value. Note that $R'$ satisfies $\Delta = \{A \rightarrow B\}$ since each $v_i$ is associated with a unique $u_i$. Since the assignments $v_i = u_i$ for all $i \in [1,n]$ form a satisfying assignment for $\phi$, every clause $c_\ell$ is satisfied by at least one literal $v_i$, hence each $C$-value $c_{\ell}$ appears at least once in $R'$ and consequently $|R'| \geq m$. For the cases where $|R'| > m$, we ensure that $R'$ has exactly one tuple with value $C = c_{j}$ for every $c_{j} \in \dom{(C)}$ by picking one arbitrary tuple and discarding the rest. Then every value $C = c_j$ for $c_{j} \in \dom{(C)}$ appears exactly once, i.e., $\%c_{j} = \frac{1}{m}$ holds, and therefore $R' \rcsatisfies \rc{}$. 
    
\end{proof}

\subsection{Reduction from Color-Balanced Matching to Simple LHS Marriage}
\label{sec:lhs-marriage}
 As mentioned in \Cref{sec:algo-reprepair}, in this section we show that when the FDs $\fdset$ is a \textit{simple LHS marriage} as defined below, which has been shown to be poly-time solvable for optimal \srepair, is as hard as the color-balanced matching problem studied in the algorithms literature. 

\begin{itemize}
    \item 
{\bf Simple LHS Marriage:} Given relation $R(A, B, C)$, the FD set is $\fdset = \{A \rightarrow B, B \rightarrow A, A \rightarrow C\}$, $C$ is the sensitive attribute with values $c_1, \cdots, c_k$ ($k$ is constant), and the RC is $\rc{}=\{\rc_1, \rc_2, \cdots, \rc_k\}$,  where $\sum_{i}\rc_i = 1$ (an exact RC). We assume set semantics, i.e., $R$ does not have any duplicate tuples.
\end{itemize}
The  \srepair{} problem for $R$ and $\fdset$ has been shown to be polytime-solvable in \cite{LivshitsKR18}.

The {\em exact-color bipartite matching problem} is defined as follows:
\begin{definition}[Exact Colored Matching ~\cite{ka2022optimal-AAMAS22}]\label{def:ecm} 
    Given an instance of an unweighted bipartite graph $G=(U \cup V,E)$ and $k$ colors, where each edge is colored by one color. The color constraint is $\phi = (\phi_1, \phi_2, \cdots, \phi_k)$, where $\phi_i$ is the number of edges in color $i$ in a matching $M$ and $\phi_i\geq0$. The number of colors $k$ is a constant. A matching $M$ is an exact colored matching in $G$ if it satisfies the color constraint $\phi$.
\end{definition}


The exact colored matching problem for two colors is known to be in RNC~\cite{mulmuley1987matching}. A recent study by \citet{ka2022optimal-AAMAS22} shows that the problem for a constant number of colors $k$ has a randomized polynomial-time algorithm by generalizing the algorithm presented in \cite{mulmuley1987matching}, and hence it is not NP-hard unless {\tt NP = RP}.

The following lemma shows that the simple LHS marriage \rsrepair{} is at least as hard as the exact colored matching problem. 


\cut{
We propose a randomized polynomial-time algorithm for the simple LHS marriage \rsrepair{}. First, we prove Item 1 and 2 to show their equivalence. Then, we provide a randomized polynomial-time algorithm to solve the exact colored matching.
}

\cut{
\begin{definition}[Simple LHS Marriage]\label{def:simple_lhs}
    For a \prob{$(\relation{}, \fdset{}, \rc{})$}, a Simple LHS Marriage satisfies that: 1) $\fdset{} = \{X \rightarrow Y, X \rightarrow Y\}$, where $X$ and $Y$ are disjoint sets of attributes; 3) the representative constraint on sensitive attribute $A_{\ell}$ is $\rc{}=\{\rc_1, \rc_2, \cdots, \rc_k\}$, where the number of atomic representative constraint $k$ is constant 3) $\pi_{X \cup Y \cup {A_\ell}}R$ has no duplicates.
\end{definition}
}

\begin{lemma}\label{lem:ecm_to_simple_lhs}
The exact colored matching problem with color constraint \cut{$\phi = (\phi_1, \phi_2, \cdots, \phi_k)$ }can be reduced to the decision problem of finding an optimal \rsrepair{} for simple LHS marriage in polynomial time.
\end{lemma}
\begin{proof}
    Recall the decision problem of computing an optimal \rsrepair{} from \Cref{thm:hardness}: suppose \schema{} and \fdset{} are fixed. Given a relation \relation{}, a RC \rc{}, and a non-negative integer $T$, the answer is yes if and only if there exists an \rsrepair{} that contains at least $T$ tuples.
    
    Given an instance of the exact colored matching problem with an unweighted bipartite graph $G=(U \cup V,E)$ and a color constraint $\phi = (\phi_1, \phi_2, \cdots, \phi_k)$, we construct an instance of the decision problem of finding an \rsrepair{} for simple LHS marriage as follows:
    \begin{itemize}
        \item Create a relation $R(A, B, C)$, where $C$ is the sensitive attribute with domain $\{c_1, c_2, \cdots, c_k\}$. For each edge $(u, v) \in E$ with color $c_i$, add a tuple $(u, v, c_i)$ to $R$.
        \item Set the FD set $\fdset{} = \{A \rightarrow B, B \rightarrow A, A \rightarrow C\}$.
        \item Set the RC $\rc{} = \{\rc_1, \rc_2, \cdots, \rc_k\}$, where $\rc_i = \frac{\phi_i}{\sum_{i \in [1,k]}\phi_i}$ for $i \in \{1, 2, \cdots, k\}$.
        \item Set $T = \sum_{i \in [1,k]}\phi_i$
    \end{itemize}
    
    We claim that $G$ has an exact colored matching $M$ satisfying $\phi$ if and only if there is an \rsrepair{} $\relation'$ w.r.t. \fdset{} and \rc{} such that $|\relation'| \geq T$.

    \paragraph*{(If)} Given an \rsrepair{} $\relation'$ of $\relation$ w.r.t. \fdset{} and \rc{} with size at least $T$, we can construct an exact colored matching $M$ in $G$ as follows: for each tuple $(u, v, c_i)$ in $\relation'$, add the edge $(u, v)$ with color $c_i$ to $M$. Since $\relation'$ satisfies the FDs in $\fdset{}$, for any two tuples $(u_1, v_1, c_{i_1})$ and $(u_2, v_2, c_{i_2})$ in $\relation'$, if $u_1 = u_2$, then $v_1 = v_2$, and if $v_1 = v_2$, then $u_1 = u_2$. This implies that no two edges in $M$ share a common endpoint, and thus $M$ is a matching. As $\relation'$ satisfies $\rc{}$, the number of edges in $M$ with color $c_i$ is at least $|\relation'| \cdot \rc_i \geq \sum_{i \in [1,k]}\phi_i \cdot \frac{\phi_i}{\sum_{i \in [1,k]}\phi_i} = \phi_i$. If the number of edges in $M$ with $c_i$ is strictly greater than $\phi_i$, we can remove arbitrary edges with color $c_i$ from $M$ to satisfy the exact color constraint $\phi_i$. After this process, $M$ becomes an exact colored matching satisfying $\phi$.

    \paragraph*{(Only if)} Given an exact colored matching $M$ in $G$ satisfying $\phi$, we can construct an \rsrepair{} $\relation'$ of $\relation$ as follows: for each edge $(u, v)$ in $M$ with color $c_i$, add the tuple $(u, v, c_i)$ to $\relation'$. Since $M$ is a matching, for any two tuples $(u_1, v_1, c_{i_1})$ and $(u_2, v_2, c_{i_2})$ in $\relation'$, we have $u_1 \neq u_2$ and $v_1 \neq v_2$. Therefore, $\relation'$ satisfies the FDs in $\fdset{}$. As $M$ satisfies $\phi$, the number of tuples in $\relation'$ with sensitive value $c_i$ is exactly $\phi_i$ for each $c_i$, hence \relation{} satisfies that $\%c_i = \frac{\phi_i}{\sum_{i \in [1,k]}\phi_i} = \rc_i$. Therefore $\relation'$ satisfies both \fdset{} and \rc{}. Moreover,  $|R'| = \sum_{i=1}^k \phi_i = T$, so $R'$ is an \rsrepair{} with size at least $T$.

    The reduction takes polynomial time as it scans the edges in $E$ once to construct the relation \relation{} and the RC $\rc{}$ are directly set based on the given color constraint $\phi$.
\end{proof}





\subsection{Pseudocode and Walkthrough of \postclean{}}\label{sec:appendix-postclean-code}
In this section, \Cref{example:post_clean_running} gives a simple scenario about how \postclean{} works and the pseudocode of the $\postclean(R, \rc)$ procedure from Section~\ref{subsec:postclean} is shown in \Cref{alg:post_clean} followed by a detailed walk-through  of the algorithm.

\begin{example}\label{example:post_clean_running}
    Assume that an input $(\relation{}, \rc{})$ of \postclean{} satisfies: $|\relation{}| = 12, |\sigma_{A_s = \texttt{red}}{\relation{}}| = 2, |\sigma_{A_s = \texttt{blue}}{\relation{}}| = |\sigma_{A_s = \texttt{green}}{\relation{}}| = 5$, and $\rc{} = \{\%\texttt{red}\geq\frac{1}{4}, \%\texttt{blue}\geq \frac{1}{4}, \%\texttt{green}\geq\frac{1}{4}\}$. 
    \postclean{} iterates over $T$ from $12$ to $0$. When $T$ is between 12 and 9, it requires the subset $R'$ has at least $\tau_{\texttt{red}} = \ceil{\frac{T}{4}} = 3$ red tuples, but $R$ cannot fulfill that number. Next, when $T = 8$, $\tau_{\texttt{red}} = \tau_{\texttt{blue}} = \tau_{\texttt{green}} = 2$, which can be supplied by $R$ and thus $T_0 = \tau_{\texttt{red}} + \tau_{\texttt{blue}} + \tau_{\texttt{green}} = 6$. $T_0 < T$ indicates that we have to add two extra tuples that can have arbitrary $A_s$-values from \relation{}. Moreover, adding these two tuples
    will not violate the RC  \rc{}, since $\tau_{\texttt{red}}, \tau_{\texttt{blue}}, \tau_{\texttt{green}}$ are already larger than or equal to $\frac{T}{4}$. Hence, \postclean{} will collect two more blue tuples, or two more green tuples, or one more blue and one more green from $R$.
\end{example}

\begin{algorithm}
    \caption{\postclean{}$(R, \rc{})$}\label{alg:post_clean}
    \begin{algorithmic}[1]
        \Require{A relation $R$ and a RC \rc{}}
        \Ensure{A subset $R' \subseteq R$ of largest size such that $R' \rcsatisfies \rc$}
        \For{Size $T$ from $|R|$ to $1$}
            \State $b \gets $ True;
            \For{Every value $a_{\ell}$ in $\dom{(A_s)}$}
                \State $\tau_{\ell} \gets \ceil{T \cdot \rc{(a_{\ell})}}$;
                \If{$\tau_{\ell} > |\sigma_{A_s = a_{\ell}}{\relation{}}|$}
                    \State 
                    $b \gets$ False;
                    \State \textbf{break}
                \EndIf
            \EndFor
            \State $T_0 \gets \sum_{\ell \in [1,k]}\tau_{\ell}$;
            \If{$b$ is True and $T_0 \leq T$}
                \While{$T_0 < T$}
                    \State Arbitrarily choose an $a_{\ell}$ from $\dom{(A_s)}$ where the corresponding $\tau_{\ell} < |\sigma_{A_s = a_{\ell}}{\relation{}}|)$;
                    \State $\tau_{\ell} \gets \tau_{\ell} + 1$;
                    \State $T_0 \gets T_0 + 1$;
                \EndWhile
                \State $R' \gets \bigcup_{a_{\ell} \in \dom{}(A_s)}$ an arbitrary subset of $\sigma_{A_s = a_{\ell}} R$ of size $\tau_{\ell}$;
                \State \Return $R'$;
            \EndIf
        \EndFor
        \State \Return $\emptyset$;
    \end{algorithmic}
\end{algorithm}

The algorithm computes the size (denoted by $T$) of the maximum subset (denoted by $R'$) of \relation{} that satisfies the \rc{}. To achieve this, the main part of \Cref{alg:post_clean} (lines 1-20) is a loop that enumerates the size $T$ from $|\relation{}|$ to $1$, which terminates once a feasible $R'$ is found. A Boolean variable $b$ is initialized in line 2. In lines 3-9, we check if there are sufficient number of tuples for every value $A_s = a_{\ell}$ in \relation{}. Specifically, $\tau_{\ell}$ is assigned to the lower-bound value for the number of tuples with $A_s = a_{\ell}$.
 Then $\tau_{\ell}$ is compared with $|\sigma_{A_s = a_{\ell}}{\relation{}}|$, where the latter denotes the number of tuples in \relation{} with value $A_s = a_{\ell}$. If the lower-bound requirement is larger than the number of available tuples with $A_s = a_\ell$ (line 5), then the verification fails, $b$ is falsified and the inner loop is stopped early. In this case, the current $T$ is not valid, therefore lines 11-19 will be skipped and the procedure will go to next value of $T$ if $T > 1$; otherwise,  the loop (lines 1-20) is finished and it is the case where $\emptyset$ is the only valid subset of $\relation{}$ that trivially satisfies all the lower-bound constraints. At that point, the $\emptyset{}$ is returned in line 21.
 \par
Next, we discuss the case when $b$ is true for all values $A_s = a_{\ell}$ for the current value of $T$. In line 10, $T_0$ is computed as the sum of $\tau_{\ell}$s, representing the minimum size of $R'$ that satisfies every lower-bound constraint. Note that $T_0$ can be larger than $T$, because $\tau_{\ell}$ takes the ceiling for every $a_{\ell}$. In that case, it is impossible to derive a $R'$ of size $T$ and the algorithm fails in the condition check in line 11. On the other hand, $T_0$ is smaller than or equal to $T$, representing that $R$ does have a subset that satisfies \rc{}. If $T_0 < T$ (line 12), then there is a slack between the expected size $T$ and the $T_0$ tuples by concatenating $\tau_{\ell}$ tuples for every $a_{\ell}$. This slack is fulfilled by retaining additional tuples with arbitrary value $a_{\ell}$ as long as \relation{} has more tuples with $A_s = a_{\ell}$ than what it already took by that $T_0$ tuples (lines 12-16). In line 17, the algorithm forms $R'$ by concatenating the samples of tuples with $A_s = a_{\ell}$ for every $a_{\ell}$ and of size $\tau_{\ell}$.

Here is an example for how \postclean{} works:

\subsection{Proof of Proposition~\ref{prop:post_clean}}
\postcleanprop*

\begin{proof}
We first prove the optimality, and then the polynomial running time.
   \paragraph*{(Optimality)} By contradiction, suppose that there is a $R'' \subseteq \relation{}$ where $R'' \rcsatisfies \rc{}$ and $|R''| > |R'|$. 
   Note that $T = |R''|$ is checked earlier than $T = |R'|$ in the outermost for-loop. Hence, 
    $T = |R''|$ failed the validation in \Cref{alg:post_clean} 
    in line 11, either by $b = \text{false}$ or by $T_0 > T$.    
    
    (i) If $b = \text{false}$,
    then there exists an $a_{\ell}$ where $\tau_{\ell} > |\sigma_{A_s = a_{\ell}}{\relation{}}|$. Furthermore, since $R'' \subseteq \relation{}$ indicates that $|\sigma_{A_s = a_{\ell}}{\relation{}}| \geq |\sigma_{A_s = a_{\ell}}{\relation{}''}|$, we have 
    \begin{equation}\label{equn:postclean-1}
    \tau_{\ell} > |\sigma_{A_s = a_{\ell}}{\relation{}''}|
    \end{equation}
    Additionally, $R'' \rcsatisfies \rc{}$ indicates that
    \begin{equation}\label{equn:postclean-2}
        |\sigma_{A_s = a_{\ell}}{\relation{}''}| \geq |\relation{}''| \cdot \rc{(a_{\ell})}.
    \end{equation}
    According to the value assignment of $\tau_{\ell}$, we have     
    \begin{equation}\label{equn:postclean-3}
        \tau_{\ell} = \ceil{T \cdot \rc{(a_{\ell})}} = \ceil{|R''| \cdot \rc{(a_{\ell})}}
    \end{equation}
    Combining \Cref{equn:postclean-2} and \Cref{equn:postclean-3}, we have
    \begin{equation}\label{equn:postclean-4}
        \tau_{\ell} \leq \ceil{|\sigma_{A_s = a_{\ell}}{\relation{}''}|}.
    \end{equation}
    Combining \Cref{equn:postclean-1} and \Cref{equn:postclean-4} that are related to $\tau_{\ell}$, we have 
    \begin{equation}
        \ceil{|\sigma_{A_s = a_{\ell}}{\relation{}''}|} > |\sigma_{A_s = a_{\ell}}{\relation{}''}|.
    \end{equation} However, $|\sigma_{A_s = a_{\ell}}{\relation{}''}|$ is exactly an integer representing the number of tuples with value $a_{\ell}$ in $R''$, so we have (by removing the ceiling function)
    \begin{equation}
        |\sigma_{A_s = a_{\ell}}{\relation{}''}| > |\sigma_{A_s = a_{\ell}}{\relation{}''}|,
    \end{equation}
    which leads to a contradiction. 
    
    (ii) Consider the other case that $T_0 > T$. We express $T_0$ and $T$ by $R''$ and $\rc{}$:
    \begin{equation}
        T_0 = \sum_{\ell \in [1,k]} \tau_{\ell} = \sum_{\ell \in [1,k]}{\ceil{|R''| \cdot \rc{(a_{\ell})}}},
    \end{equation}
    \begin{equation}
        T = |\relation{}''|.
    \end{equation}
    By substituting them into the inequality $T_0 > T$, we have
    \begin{equation}\label{equn:postclean-10}
        \sum_{\ell \in [1,k]}{\ceil{|R''| \cdot \rc{(a_{\ell})}}} > |R''|.
    \end{equation}
    Again, $R'' \rcsatisfies \rc$ indicates that $|\sigma_{A_s = a_{\ell}}{\relation{}''}| \geq |\relation{}''| \cdot \rc{(a_{\ell})}$. Combining it with \Cref{equn:postclean-10}, we have
    \begin{equation}\label{equn:postclean-11}
        \sum_{\ell \in [1,k]}{\ceil{|\sigma_{A_s = a_{\ell}}{\relation{}''}|}} > |R''|.
    \end{equation} 
    Since $|\sigma_{A_s = a_{\ell}}{\relation{}''}|$ is exactly an integer, we can simplify \Cref{equn:postclean-11} into $\sum_{\ell \in [1,k]}{|\sigma_{A_s = a_{\ell}}{\relation{}''}|} > |R''|$, which implies $|R''| > |R''|$, which is a contradiction.

    This completes the proof of Optimality.

    \paragraph*{(Time complexity)}
    First, it is important to know that $|R'|$ is bounded by $|\relation{}|$. In line 1, we iterate over $T$ in $O(|\relation{}|)$. Within the loop, we enumerate all distinct $a_{\ell}$ in the active domain of $A_s$, which is $O(k) = O(|R|)$. In lines 4-8, we compute $\tau_{\ell}$ and validate it in $O(1)$. The second part within the loop is the process of distributing the slack, which can be bounded in $O(|\relation{}|)$. Finally in line 19, we form $R''$ by sequentially taking samples from $R'$ and concatenations, this step is $O(|\relation{}|)$. Overall, the complexity is $O(|\relation{}|^2)$.
\end{proof}

\subsection{Proof of Proposition~\ref{prop:lhs-chain-reduction}}
\lhschainreductionprop*
\begin{proof}
    Since the `if' direction (an LHS chain reduces to $\emptyset$ by repeated applications of consensus FD and common LHS) was observed in \cite{LivshitsKR18}, we show the `only if' direction here, i.e., show that, if an FD set reduces to $\emptyset$ by repeated applications of consensus FD and common LHS, it must be an LHS chain.
    \par
    We prove `only if' by contradiction. Consider an FD set $\fdset$ that is not an LHS chain but was reduced to $\emptyset$ by consensus FD and common LHS simplifications. Since $\fdset$ is not an LHS chain, it has two FDs $f_1: X_1 \rightarrow Y_1$ and $f_2: X_2 \rightarrow Y_2$ such that $X_1 \not\subseteq X_2$ and $X_2 \not\subseteq X_1$. Hence there is an attribute $A_1 \in X_1$ such that $A_1 \notin X_2$, and an attribute $A_2 \in X_2$ such that $A_2 \notin X_1$. Apply consensus FD and common LHS simplifications on $\fdset$ until no more of either of these two simplifications can be applied. Without loss of generality, suppose $f_1$ was removed altogether from $\fdset$ first by the simplifications, and consider the last step $j$ that removed the last attribute from $X_1$.  Since $A_1 \in X_1 \setminus X_2$, $A_1$ cannot be removed by common LHS simplification in any step $\leq j$. Since at least $A_1$ remains in the LHS of $f_1$ in step $j$, $f_1$ cannot be removed from $\fdset$ by consensus FD since the LHS is not $\emptyset$. This contradicts the assumption that $\fdset$ was reduced to $\emptyset$, hence proved. 
\end{proof}

\section{Details in Section~\ref{sec:poly-algo-chain}}\label{sec:appendix_polyalgo}
We present the missing details of Section~\ref{sec:poly-algo-chain} 
as follows:
\begin{itemize}
    \item Size of Candidate Sets: \Cref{lemma:candset_size}
    \item Proof of \Cref{lem:dpalgo_complexity} (\Cref{sec:poly-algo-chain})
    \item Proof of \Cref{lemma:candset} (\Cref{subsec:candidates})
    \item Pseudocode of \reprinsert{} (\Cref{subsec:candidates})
    \item Proof of \Cref{lem:reduce-correct} (\Cref{subsec:reduce})
    \item Proof of \Cref{thm:poly-time-chain} (\Cref{sec:algo-reprepair})
\end{itemize}
\subsection{Size of Candidate Sets}
{
\begin{lemma}\label{lemma:candset_size}
    Given a relation \relation{}, an FD set \fdset{}, and the domain size $k$ of the sensitive attribute $A_s$ of \relation{}, the size of the \candidateset{} $\candset{\relation{}, \fdset{}}$ is $O(|\relation{}|^k)$.
\end{lemma}
\begin{proof}
    With the property that there are no two candidates in \candset{\relation{}, \fdset{}} that are representatively equivalent to each other and a candidate in \candset{\relation{}, \fdset{}} must be a subset of \relation{}, the size of \candset{\relation{}, \fdset{}} is bounded by the maximum number of subsets of \relation{} such that no two subsets have the same number of tuples for every $a_{\ell} \in \dom{(A_s)}$. By the rule of product\footnote{Given a relation with 3 red tuples, 2 blue tuples, and 1 green tuple. For any subest of this relation, it has 0 red, 1 red, 2 red, or 3 red (4 possibilities), independently 0 blue, 1 blue, or 2 blue (3 possibilities), and independently 0 green, 1 green (2 possibilities). By taking the product $4 \times 3 \times 2 = 24$, there are at most $24$ subsets such that each of them has a unique distribution of the color.}, it is 
    \begin{equation}
        \prod\limits_{a_{\ell} \in \dom{}(A_s)}(|\sigma_{A_s = a_{\ell}}{\relation{}}| + 1)    
    \end{equation}
    Specifically, given that the input relation \relation{} has $|\sigma_{A_s = a_{\ell}}{\relation{}}|$ tuples with value $a_{\ell}$ in attribute $A_s$ for every $a_{\ell}$, for every $a_{\ell} \in \dom{(A_s)}$, $(|\sigma_{A_s = a_{\ell}}{\relation{}}| + 1)$ means the possible numbers of tuples with $A_s = a_{\ell}$ from any subsets of \relation{}, i.e. no tuple with $A_s = a_{\ell}$, 1 tuple with $A_s = a_{\ell}$, and so on. 

    Next, we utilize the AM-GM inequality:
    \begin{equation}
        \begin{aligned}
        \prod\limits_{a_{\ell} \in \dom{}(A_s)}(|\sigma_{A_s = a_{\ell}}{\relation{}}| + 1) &\leq (\frac{\sum\limits_{a_{\ell} \in \dom{(A_s)}}(|\sigma_{A_s = a_{\ell}}{\relation{}}| + 1)}{k})^k \\
        &= (\frac{|R| + k}{k})^k \\
        &= (\frac{|R|}{k} + 1)^k
        \end{aligned}
    \end{equation}

    where $|\relation{}|$ is the input relation size and $k = |\dom{(A_s)}|$ is the active domain size of $A_s$, which is fixed. When $|\relation{}| \geq 2$ and $k \geq 2$, i.e., any relation with at least two tuples and at least two distinct sensitive values of $A_s$, we have $(\frac{|R|}{k} + 1)^k = O(|\relation{}^k|)$.
\end{proof}
}

\subsection{Proof of \Cref{lem:dpalgo_complexity}}\label{subsec:proof_lem_dpalgo_complexity}
The time complexity of \reduce{} (\Cref{lem:dpalgo_complexity}) and \dpalgo{} depends on the time complexity of \consensusreduction{} (\Cref{lemma:complexity_consensus}) and that of \commonlhsreduction{} (\Cref{lemma:complexity_commonlhs}). Since \reduce{} and \consensusreduction{} (resp. \commonlhsreduction{}) call each other recursively, the time complexity analysis is based on recurrence relations. The recurrence relations for \reduce{}, \consensusreduction{}, and \commonlhsreduction{} are \recurrence{\text{RDCE}}, \recurrence{\text{COSNS}}, \recurrence{\text{COLHS}} respectively.

\subsubsection{Time Complexity of \consensusreduction{}}
We analyze the time complexity of \consensusreduction{} through a recurrence relation in \Cref{lemma:complexity_consensus}.
\begin{lemma}\label{lemma:complexity_consensus}
    The recurrence relation for \consensusreduction{} is given by the following expression:
    \begin{equation}\label{eq:cosns}
        \recurrence{\text{COSNS}}(\relation{}, \fdset{}) = \sum\limits_{y_{\ell} \in \dom{(Y)}} \recurrence{\text{RDCE}}(\relation{}_{y_{\ell}}, \fdset{} - f) + O(k \cdot |\relation{}|^{2k + 1})
    \end{equation}
    , where $\relation{}_{y_{\ell}} = \sigma_{Y = y_{\ell}}{\relation{}}$, $f$ is the consensus FD and $Y$ is its RHS, $\fdset{} - f$ is removing $f$ from \fdset{}, $|\relation{}|$ is the relation size, $k$ is the domain size of the sensitive attribute $A_s$, and \recurrence{RDCE}$(\cdot, \cdot)$ is the recurrence relation of \reduce{}.
\end{lemma}
\begin{proof}
    First of all, the algorithm involves \candidateset{}s on $\relation{}_{y_{\ell}}$, and $\relation{}_{y_{\ell}}$ is a subset of $\relation{}$, so $|\relation{}_{y_{\ell}}| \leq |\relation{}|$. Without loss of accuracy, the size of the candidate set of $\relation{}_{y_{\ell}}$ is also $O(|\relation{}|^k)$ as implied in \Cref{lemma:candset_size}.

    In lines 2-7, \consensusreduction{} iterates over each distinct $y_{\ell}$ in $O(|\relation{}|)$. Within the loop, it first computes $\candset{\relation{}_{y_{\ell}}, \fdset{} - f}$, which is the \candidateset{} of relation $\relation{}_{y_{\ell}}$ and FD set $\fdset{} - f$ in line 3. Next, in line 4-6, the algorithm enumerates each candidate $R'$ in $\candset{\relation{}_{y_{\ell}}, \fdset{} - f}$ in $O(|\relation{}|^k)$, since the size of a candidate set is $O(|\relation{}|^k)$ (\Cref{lemma:candset_size}). For each candidate $R'$, in line 5, the algorithm utilizes \reprinsert{} to insert $R'$ in $O(k \cdot |\relation{}|^k)$ (\Cref{lem:complexity_of_reprinsert}). Combining them, we obtain \Cref{eq:cosns}.
    \cut{
        Overall, we have 
        \begin{equation}
            \recurrence{\text{COSNS}}(\relation{}, \fdset{}) = \sum\limits_{y_{\ell} \in \dom{(Y)}} \recurrence{\text{RDCE}}(\relation{}_{y_{\ell}}, \fdset{} - f) + O(k \cdot |\relation{}|^{2k + 1}).
        \end{equation}
    }
\end{proof}

\subsubsection{Time Complexity of \commonlhsreduction{}}
The time complexity of \commonlhsreduction{} is analyzed using a recurrence relation in Lemma~\ref{lemma:complexity_commonlhs}.

\begin{lemma}\label{lemma:complexity_commonlhs}
    The recurrence relation for \commonlhsreduction{} is given by the following expression:
    \begin{equation}\label{eq:colhs}
        \recurrence{\text{COLHS}}(\relation{}, \fdset{}) = \sum\limits_{x_{\ell} \in \dom{(X)}} \recurrence{\text{RDCE}}(\relation{}_{x_\ell}, \fdset_{- X}) + O(k \cdot |\relation{}|^{3k + 1})
    \end{equation}
    , where $\relation_{x_{\ell}} = \sigma_{X = x_{\ell}}{\relation{}}$, $X$ is the \commonlhs{} attribute, $\fdset{}_{-X}$ is removing column $X$ from \fdset{}, $|\relation{}|$ is the relation size, $k$ is the domain size of the sensitive attribute $A_s$, and \recurrence{RDCE}$(\cdot, \cdot)$ is the recurrence relation of \reduce{}.
\end{lemma}
\begin{proof}
    First of all, the algorithm involves \candidateset{}s on $\relation{}_{x_{\ell}}$, $\relation{}_{x_1, \dots, x_{\ell - 1}}$ and $\relation{}_{x_1, \dots, x_{\ell}}$. All of them are subsets of $\relation{}$, so their sizes are bounded by $|\relation{}|$. Without loss of accuracy, the size of the candidate set of each of them is also $O(|\relation{}|^k)$ as implied in \Cref{lemma:candset_size}.

    In lines 1-8, \commonlhsreduction{} iterates over each distinct value $x_{\ell}$ in $O(|\relation{}|)$. Within the loop, it computes $\candset{\relation{}_{x_{\ell}}, \fdset{}_{-X}}$, which is the \candidateset{} of relation $\relation{}_{x_{\ell}}$ and FD set $\fdset{}_{-X}$ in line 3. Next, in lines 4-7, it enumerates each candidate $R'$ from $\candset{\relation{}_{x_1, \dots, x_{\ell - 1}}, \fdset}$ and each candidate $R''$ from $\candset{\relation{}_{x_{\ell}}, \fdset{}}$. Note that the size of candidate set is $O(|\relation{}|^{k})$ (\Cref{lemma:candset_size}), so this enumeration will be $O(|\relation{}|^{2k})$. In line 5, it compute the union $R_0$ of  $R'$ and $R''$ in $O(|\relation{}|)$. In line 6, it inserts $R_0$ into the candidate set $\candset{\relation{}_{x_1, \dots, x_{\ell}}, \fdset{}}$ by \reprinsert{} in $O(k \cdot |\relation{}|^k)$ (\Cref{lem:complexity_of_reprinsert}). Combining them, we obtain \Cref{eq:colhs}.
\end{proof}
\subsubsection{Time Complexity of \reduce{} and \dpalgo{}}
In \Cref{lem:dpalgo_complexity}, we present the time complexity of \dpalgo{}, including the time complexity of \reduce{}, which serves as the core sub-procedure of \dpalgo{}.
\dpalgocomplexitylem*

{
\begin{proof}
    The complexity of \dpalgo{} consists of three parts:
    \begin{enumerate}
        \item \reduce{} computes a \candidateset{} in $O(m \cdot |\fdset{}| \cdot k \cdot |\relation{}|^{3k + 2})$ (proof below).
        \item \postclean{} computes a maximum subset that satisfies the \rc{} in $O(|\relation{}|^2)$ (\Cref{prop:post_clean}) for every candidate from a candidate set of size $O(|\relation{}|^k)$ (\Cref{lemma:candset_size}). In total, \postclean{} runs in $O(|\relation{}|^{k + 2})$.
        \item The selection of the maximum \rsrepair{} after post-cleaning is linear in the size of the \candidateset{}, $O(|\relation{}|^k)$.
    \end{enumerate}

    For part (1), \reduce{} is an if/else statement with three branches. Therefore, the complexity of \reduce{} consists of two parts: if-condition checking and branching. The if-condition checking part checks:
    \begin{itemize}
        \item if \fdset{} is empty in $O(1)$;
        \item otherwise, if \fdset{} has a \consensusfd{} by going ever each FD in \fdset{} in $O(|\fdset{}|)$;
        \item otherwise, if \fdset{} has a \commonlhs{} by checking each attribute for common columns in the LHS of each FD in \fdset{} in $O(m \cdot |\fdset{}|)$.
    \end{itemize}
    Depending on the results of condition checking, the recurrence relation of \reduce{} has three possibilities: 
    \begin{itemize}
        \item if \fdset{} is empty, then the algorithm simply returns $\{\relation{}\}$ and 
        \begin{equation}\label{eq:rdce_eq_r}
            \recurrence{\text{RDCE}}(\relation{}, \fdset{}) = O(|\relation{}|);
        \end{equation}
        \item otherwise, if \fdset{} has a \consensusfd{}, the sub-procedure \consensusreduction{} is called, as implied by \Cref{lemma:complexity_consensus} and \Cref{eq:cosns}, 
        \begin{equation}\label{eq:rdce_eq_cosns}
            \begin{aligned}
                            \recurrence{\text{RDCE}}(\relation{}, \fdset{}) &= \recurrence{\text{COSNS}}(\relation{}, \fdset{}) \\
            &= \sum\limits_{y_{\ell} \in \dom{(Y)}} \recurrence{\text{RDCE}}(\relation{}_{y_{\ell}}, \fdset{} - f) + O(k \cdot |\relation{}|^{2k + 1});
            \end{aligned}
        \end{equation}
        \item otherwise, if \fdset{} has a \commonlhs{}, the sub-routine \commonlhsreduction{} is called, implied by \Cref{lemma:complexity_commonlhs} and \Cref{eq:colhs}, 
        \begin{equation}\label{eq:rdce_eq_colhs}
        \begin{aligned}
            \recurrence{\text{RDCE}}(\relation{}, \fdset{}) &= \recurrence{\text{COLHS}}(\relation{}, \fdset{}) \\
            &= \sum\limits_{x_{\ell} \in \dom{(X)}} \recurrence{\text{RDCE}}(\relation{}_{x_\ell}, \fdset{}_{-X}) + O(k \cdot |\relation{}|^{3k + 1}).
        \end{aligned}           
        \end{equation}
    \end{itemize}
    Intuitively, \reduce{} can be viewed as a tree traversal, where each node is a call of \reduce{} and the root corresponds to $\reduce{(\relation{}, \fdset{})}$. Note that the call of $\reduce{}$ by each leaf node is on a clean sub-relation and empty \fdset{}. The number of leaves is bounded by $|\relation{}|$, as the associated sub-relations are disjoint. The tree height, equivalent to the number of reductions applied, is bounded by $O(m \cdot |\fdset{}|)$ as at least one attribute is removed from the \fdset{} per reduction. The number of nodes in each level must be larger than the number of nodes in the parent level, so the tree size is bounded by $O(m \cdot |\fdset{}| \cdot |\relation{}|)$ and therefore $\recurrence{\text{RDCE}}(\relation{}, \fdset{})$ is equal to the sum of the Big O terms of the associated $\recurrence{\text{RDCE}}(\cdot, \cdot)$ of all non-root tree nodes. Therefore, by combining \Cref{eq:rdce_eq_r,eq:rdce_eq_cosns,eq:rdce_eq_colhs}, $\recurrence{\text{RDCE}}(\relation{}, \fdset{})$ is bounded by $O((m \cdot |\fdset{}| \cdot |\relation{}|) \cdot (k \cdot |\relation{}|^{3k + 1}))$.
    To sum up, the branching step is $O(m \cdot |\fdset{}| \cdot k \cdot |\relation{}|^{3k + 2})$. Compared to this, the time complexity of if-condition checking, $O(m \cdot |\fdset{}|)$, can be ignored.

    Overall, the time complexity of \dpalgo{} is dominated by \reduce{}, which is $O(m \cdot |\fdset{}| \cdot k \cdot |\relation{}|^{3k + 2})$.
\end{proof}
}
\subsection{Proof of Lemma~\ref{lemma:candset}}\label{subsec:proof_lem_candset}
\candsetlem*
\begin{proof}
\Cref{alg:dpalgo} returns $\arg\max\limits_{s \in S}{|s|}$ as the output \rsrepair{} where $S$ is defined as follows:
        \begin{equation}\label{eq:dpalgo-correct}
        S \coloneq \{\postclean{(R', \rc{})} \mid \forall R' \in \candset{\relation{}, \fdset{}}\},
        \end{equation}
    
    Since $\candset{\relation{}, \fdset{}} \subseteq \allsrepair{\relation{}, \fdset{}} $, each $R' \in \candset{\relation{}, \fdset{}}$ is an \srepair{}. Moreover, since the final output of \Cref{alg:dpalgo} is obtained by applying \postclean{} on such \srepair{}s, by \Cref{prop:post_clean}, the output satisfies the RC $\rc$ and thus is an \rsrepair{}. Next, we prove the optimality of the final answer when $\candset{\relation{}, \fdset{}}$ is correctly computed.

    \par
    We argue that applying \Cref{eq:dpalgo-correct} on $\allsrepair{\relation{}, \fdset{}}$ instead of $\candset{\relation{}, \fdset{}}$ yields the optimal \rsrepair{}. To see this, note that any \rsrepair{} is also an \srepair{}, so an optimal \rsrepair{} $S^*$ must belong to $\allsrepair{\relation{}, \fdset{}}$ that already satisfies the RC $\rc$. If we return an \rsrepair{} $S^0$ by applying \postclean{} on every \srepair{} in $\allsrepair{\relation{}, \fdset{}}$ and choosing the one with the maximum size, $|S^0| \geq |S^*|$, and since $S^*$ is the optimal \rsrepair{}, $|S^0| = |S^*|$. Hence the final output $S^0$ would be an optimal \rsrepair{}.
    
    Finally, we argue that \srepair{}s in $\allsrepair{\relation{}, \fdset{}} \setminus \candset{\relation{}, \fdset{}}$ cannot derive an \rsrepair{} that deletes strictly fewer tuples. Consider any \srepair{} $R'' \in \allsrepair{\relation{}, \fdset{}} \setminus \candset{\relation{}, \fdset{}}$. By the candidate set definition (\Cref{def:candidate_set}), $R'' \not\in \candset{\relation{}, \fdset{}}$ due to one of the following reasons:
    \begin{enumerate}
        \item\label{item:repreq} There exists an $R' \in \candset{\relation{}, \fdset{}}$ where $R' \repreq R''$. In this case, we argue that $|\postclean{(R'', \rc{})}| = |\postclean{(R', \rc{})}|$. Because $R''$ and $R'$ will obtain the same $\{\tau_1, \dots, \tau_{k}\}$, which are the minimum numbers required for each sensitive value to satisfy the RC, in \postclean{} regardless of which $T$ is enumerated, therefore \postclean will return a \rsrepair{} with the same number of tuples for each distinct value of the sensitive attribute $A_s$ and consequently the same $T$ for $R'$ and $R''$ respectively.   
        \item There exists an $R' \in \candset{\relation{}, \fdset{}}$ where $R' \reprdom R''$. In this case, we argue that $|\postclean{(R'', \rc{})}| \geq |\postclean{(R', \rc{})}|$. Since there exists some $a_{c}$ where 
        \begin{equation}
            |R'| \cdot |\sigma_{A_s = a_{c}}{R'}| > |R''| \cdot |\sigma_{A_s = a_{c}}{R''}|,
        \end{equation}
         we first take tuples with $A_s$-value $a_{c}$ in $R'$ away as a backup, until $R'$ has the same value distribution of $A_s$ as $R''$. As implied by \Cref{item:repreq} above, now $R'$ and $R''$ will generate \rsrepair{}s with the same $T$. Moreover, $R'$'s backup offers $R'$ a chance to retain more tuples (consequently deleting fewer tuples) while not conflicting \rc{}. 
    \end{enumerate}
    
    Therefore, no matter how we satisfy the RC by additional deletions (through \postclean{}) on $R''$, we can do the same or even better on some candidate in $\candset{\relation{}, \fdset{}}$. Hence, $\candset{\relation{}, \fdset{}}$ is sufficient to compute the \rsrepair{} as the entire \allsrepair{\relation{}, \fdset{}} does.
\end{proof}

\cut{
A special note is that we do not consider the existence of representative dominance in this bound: for example, we a gap between the real size and the upper bound for a scenario in Example~\ref{example:loose_bound_for_candset}. However, we do not present a tighter bound which goes beyond the scope of this paper. 
    \begin{example}\label{example:loose_bound_for_candset}
        Suppose that relation $\relation{}$ has $3$ tuples with \texttt{'male'} in the sensitive attribute $A_s$ and another $3$ tuples with \texttt{'female'} (totally $6$ tuples). The largest possible size of a candidate set is $4$ (i.e. $3$ \texttt{'male'} + $0$ \texttt{'female'}, $2$ \texttt{'male'} + $1$ \texttt{'female'}, $1$ \texttt{'male'} + $2$ \texttt{'female'}, $0$ \texttt{'male'} + $3$ \texttt{'female'}), which is much less than $(\frac{6}{2} + 1)^2 = 16$.
    \end{example}
}

\subsection{Pseudocode of \reprinsert{}} \label{subsec:pseudo_reprinsert}
\begin{algorithm}
\caption{\reprinsert{$(\candset{}, R^*)$} 
}\label{alg:reprinsert}
    \begin{algorithmic}[1]
        \Require{a set \candset{} of candidates where no representative dominance or representative equivalence exist, a candidate $\relation{}^*$ to be inserted}
        \Ensure{the set of candidates after the insertion}
        \For{every $R' \in \candset{}$}
            \If{$R' \reprdom \relation{}^{*}$ or $R' \repreq R^{*}$}
                \State \Return $\candset{}$;
            \ElsIf{$\relation{}^* \reprdom R'$}
                \State $\candset{} \gets \candset{} \setminus \{R'\}$;
            \EndIf
        \EndFor
        \State \Return $\candset{} \cup \relation{}^{*}$;
    \end{algorithmic}
\end{algorithm}
    \Cref{alg:reprinsert} considers an insertion of a candidate $R^*$ into a set \candset{} of candidates and eliminates any representative dominance and representative equivalence. \cut{\Cref{alg:reprinsert} eliminates any representative dominance that might occur with $R^*$ and avoids the possibility that $R^*$ is representatively equivalent to another candidate through these operations:} First, in line 1, the algorithm enumerates each candidate $R'$ in \candset{}. Then in line 2, it checks if $R'$ representatively dominates $R^*$ or is representatively equivalent to $R^*$. If either condition is true, indicating that $R^*$ must be redundant or sub-optimal, \reprinsert{} returns the current \candset{}. Conversely, if $R^*$ representatively dominates $R'$ (line 4), $R'$ is removed from \candset{}, and the loop continues. Finally, the updated \candset{} is returned in line 8.

    \begin{lemma}\label{lem:complexity_of_reprinsert}
        \reprinsert{} terminates in $O(k \cdot |\relation{}|^k)$, where $k$ is the domain size of the sensitive attribute and $|\relation{}|$ is the input relation size.
    \end{lemma}
    \begin{proof}
        In line 1, we iterate over a set of candidates in $O(|\relation{}|^k)$ as implied by \Cref{lemma:candset_size}. Within the loop, we sequentially check the conditions (line 2 and line 4) in $O(k)$. Therefore, the overall time complexity is $O(k \cdot |\relation{}|^k)$.     
    \end{proof}
    

\subsection{Proof of \Cref{lem:reduce-correct}}\label{subsec:proof_reduce}
In \Cref{lem:reduce-correct}, we prove the correctness of \reduce{}, which relies on the correctness of \consensusreduction{} and the correctness of \commonlhsreduction{}. \Cref{fig:consensusreduction,fig:commonlhsreduction} show the workflow of these two reductions and repeat the notations used in the algorithms.
\begin{figure}[!ht]
    \centering
    \includegraphics[width=\columnwidth]{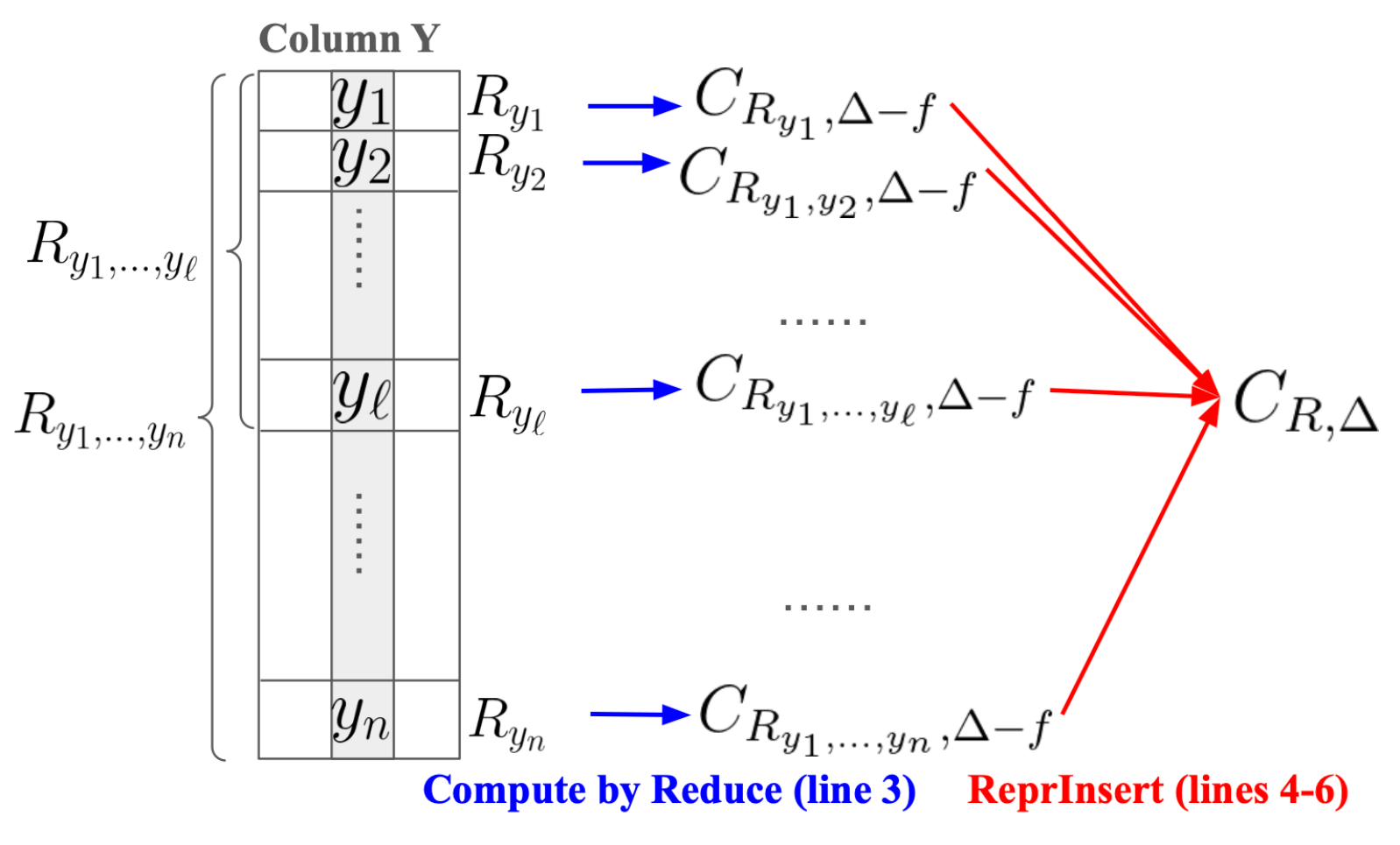}
    \caption{Notations and Workflow of \consensusreduction{}. $f$ is the \consensusfd{} and $Y$ is its RHS. $y_1, \dots, y_n$ are distinct values of column $Y$. $\relation{}_{y_{\ell}} = \sigma_{Y = y_{\ell}}{\relation{}}$ for every $\ell \in [1,n]$. $\candset{\relation{y_\ell}, \fdset{} - f}$ is the \candidateset{} of relation $\relation{y_{\ell}}$ and FD set $\fdset{} - f$. $\candset{\relation{}, \fdset{}}$ is the output of \consensusreduction{}.}
    \label{fig:consensusreduction}
    \Description[]{}
\end{figure}
\begin{figure}[!ht]
    \centering
    \includegraphics[width=\columnwidth]{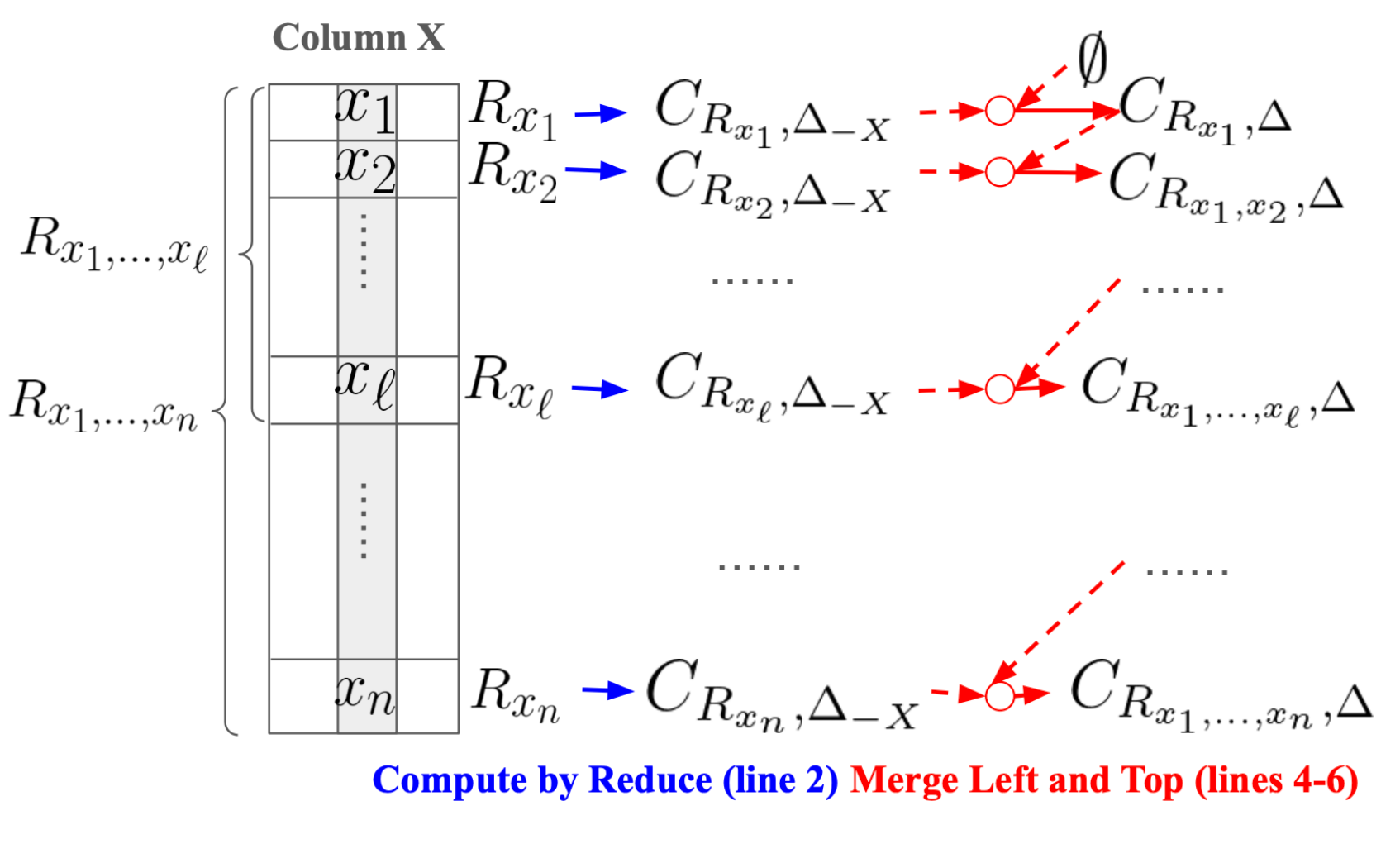}
    \caption{Notations and Workflow of \commonlhsreduction{}. $X$ is the \commonlhs{} column. $x_1, \dots, x_n$ are the distinct values of column $X$. $\relation{}_{x_{\ell}} = \sigma_{X = x_{\ell}}{\relation{}}$ for every $\ell \in [1,n]$. $\relation{x_1, \dots, x_{\ell}} = \sigma_{X = x_1 \lor \dots \lor X = x_{\ell}}{\relation{}}$ for every $\ell \in [1, n]$. $\candset{\relation{x_{\ell}}, \fdset{}_X}$ is the \candidateset{} of relation $\relation{x_{\ell}}$ and FD set $\fdset{}$. $\candset{\relation{x_1, \dots, x_{\ell}}, \fdset{}_{-X}}$ is the \candidateset{} of relation $\relation{x_1, \dots, x_{\ell}}$ and FD set $\fdset{}_{-X}$ (i.e. removing $X$ from \fdset{}). \candset{\relation{}, \fdset{}} is the output of \commonlhsreduction{}.}
    \label{fig:commonlhsreduction}
    \Description[]{}
\end{figure}
\subsubsection{Correctness of \reduce{} (\Cref{lem:reduce-correct})}
\reducecorrectlem*
\begin{proof}
    We prove it by induction on $w$, the number of reductions applied to \fdset{} by \reduce{}. Note that because \reduce{} works on \lhschain{}s, $\fdset{}$ must be able to reduced to the $\emptyset$ with a finite $w$. Assume that at each step of $w$, \reduce{} works on \relation{} and \fdset{}, which can be the same as inputted or different, i.e. a sub-relation and a portion of the input FD set.

    In the base case of induction, when $w = 0$, the FD set \fdset{} must be empty, so $\relation{}$ is the only candidate in $\candset{\relation{}, \fdset{}}$, which will be returned by \reduce{}.

    For the inductive step, assume \reduce{} works for all $w \in [0, \beta - 1]$, we prove that it also works for $w = \beta$. In this case, $\reduce{(\relation{}, \fdset{})}$ will apply one of the following reductions:
    
    \paratitle{(If $\fdset{}$ has a consensus FD $f: \emptyset \rightarrow Y$)} the condition of line 6 is satisfied and the sub-routine \consensusreduction{} is called. Similarly, \consensusreduction{} depends on $\reduce{(\relation_{y_{\ell}}, \fdset{} - f)}$ for every $y_{\ell} \in \dom{(Y)}$, whom work because they applies $w - 1$ reductions. Next, we argue that \consensusreduction{} derives a $\candset{\relation{}, \fdset{}}$ from the results of $\reduce{(\relation_{y_{\ell}}, \fdset{} - f)}$ for every $y_{\ell}$ in $\dom{(Y)}$: we first discuss the relationship between $\candset{\relation{}, \fdset{}}$ and $\bigcup\limits_{\forall y_{\ell} \in \dom{(Y)}} \candset{\relation{}_{y_{\ell}}, \fdset{} - f}$. For any candidate $R'$ from $\candset{\relation{}, \fdset{}}$, since $R'$ is a \srepair{} of $(\relation{}, \fdset{})$:
    \begin{itemize}
        \item $R' \fdsatisfies \{f\}$, implying all tuples in $R'$ agree on the same $Y$-value (say $y_{\ell}$ w.l.o.g.) and consequently $R'$ is also an \srepair{} of $\relation{}_{y_{\ell}}$. 
        \item We claim that $R'$ is not dominated by any other candidates: otherwise if there exists an $R''$ from $\candset{\relation{}_{y_{\ell}}, \fdset{} - f}$ such that $R'' \reprdom R'$, then $R''$ must also be a candidate of $\candset{\relation{}, \fdset{}}$, contradicting the assumption that $R'$ is a candidate of $\candset{\relation{}, \fdset{}}$.
        \item $R' \fdsatisfies \fdset{} - f$.
    \end{itemize}
    Combining these three points, we claim that $R'$ either is a member of $\candset{\relation{}_{y_{\ell}}, \fdset{} - f}$ (and consequently $\bigcup\limits_{\forall y_{\ell} \in \dom{(Y)}} \candset{\relation{}_{y_{\ell}}, \fdset{} - f}$) or shares the same distribution of $A_s$ values with another candidate in $\candset{\relation{}_{y_{\ell}}, \fdset{} - f}$ (and consequently $\bigcup\limits_{\forall y_{\ell} \in \dom{(Y)}} \candset{\relation{}_{y_{\ell}}, \fdset{} - f}$).

    Hence, $\candset{\relation{}, \fdset{}}$ can be derived from $\bigcup\limits_{\forall y_{\ell} \in \dom{(Y)}} \candset{\relation{}_{y_{\ell}}, \fdset{} - f}$ by inserting each candidate through \reprinsert{}, because
    \begin{itemize}
        \item there is neither representative dominance nor representative equivalence in $\candset{\relation{}, \fdset{}}$ (guaranteed by \reprinsert{}),
        \item any candidate $R'' \in \allsrepair{\relation{}, \fdset{}} \setminus \candset{\relation{}, \fdset{}}$ must be either representatively dominated by or representatively equivalent to one candidate in $\candset{\relation{}, \fdset{}}$, because $R'' \in \candset{\relation{}_{\pi_{Y}{R''}}, \fdset{} - f} \subseteq \bigcup\limits_{\forall y_{\ell} \in \dom{(Y)}} \candset{\relation{}_{y_{\ell}}, \fdset{} - f}$ as $R'' \fdsatisfies f$. Therefore, according to \reprinsert{}, $R''$, as a candidate from  $\bigcup\limits_{\forall y_{\ell} \in \dom{(Y)}} \candset{\relation{}_{y_{\ell}}, \fdset{} - f}$, does not exists in $\candset{\relation{}, \fdset{}}$ only because of one the two reasons mentioned above.
    \end{itemize}

    \paratitle{(If $\fdset{}$ has a common LHS $X$)} \commonlhsreduction{} is called (corresponding to lines 3-4). By the induction hypothesis, \commonlhsreduction${(\relation{}, \fdset{})}$ relies on \reduce$(\relation{}_{x_{\ell}}, \fdset{}_{-X})$ for each $x_{\ell} \in \dom{(X)}$. Note that those smaller instances work because they applies $w - 1$ reductions. Furthermore, we argue that \commonlhsreduction{} derives a $\candset{\relation{}, \fdset{}}$ from the results of $\reduce{(\relation{}_{x_{\ell}}, \fdset{}_{-X})}$ for every $x_{\ell}$ in $\dom{(X)}$. To achieve this, for each $\ell \in [1, n]$, $\candset{\relation{}_{x_1, \dots, x_{\ell}}, \fdset{}}$ is derived from $\candset{\relation{}_{x_1, \dots, x_{\ell - 1}}, \fdset{}}$ and $\candset{\relation{}_{x_{\ell}}, \fdset_{- X}}$ by inserting each candidate in
    \begin{equation*}
        \{R' \cup R'' \mid R' \in \candset{\relation{}_{x_{\ell}}, \fdset_{- X}}, R'' \in \candset{\relation{}_{x_1, \dots, x_{\ell - 1}}, \fdset{}}\}
    \end{equation*}
into $\candset{\relation{}_{x_1, \dots, x_{\ell}}, \fdset{}}$ using \reprinsert{} by \commonlhsreduction{}. We prove it by induction on  $\ell$. 
    
When $\ell = 1$, we argue that for any candidate $R'_1$ from $\candset{\relation_{x_1}, \fdset{}}$, there must exist an $R'_2$ from $\candset{\relation_{x_1}, \fdset{}_{-X}}$ where $R'_1 \repreq R'_2$, because:
\begin{itemize}
    \item $R'_1$ is an \srepair{} of $\relation_{x_1}$ as $R'_1$ has no violation on $\fdset{}$ as well as $\fdset{}_{-X}$. 
    \item $R'_1$ is not dominated by any candidate in $\candset{\relation_{x_1}, \fdset{}_{-X}}$. Otherwise, there must exist a candidate from $\candset{\relation_{x_1}, \fdset{}}$ also dominates $R'_1$, which conflicts with the assumption.
    \item $R'_1$ is not in $\candset{\relation_{x_1}, \fdset{}_{-X}}$ only if $R'_1$ shares the same value distribution of $A_s$ with someone else.
\end{itemize}
This step guarantees the properties of a \candidateset{} by \reprinsert{} as well.

    When $(\ell > 1)$, for any candidate $\bar{R}$ from $\candset{\relation{}_{x_1, \dots, x_{\ell}}, \fdset{}}$, we argue that:
    \begin{itemize}
        \item $\bar{R} = (\bar{R} \cap \relation_{x_{\ell}}) \cup (\bar{R} \cap \relation_{x_1, \dots, x_{\ell - 1}})$, because $\relation_{x_{\ell}}$ and $\relation_{x_1, \dots, x_{\ell - 1}}$ are two separate parts of $\relation{}$.
        \item $\bar{R} \cap \relation_{x_{\ell}}$ is a candidate of $(\relation_{x_{\ell}}, \fdset{}_{-X})$ (or shares a value distribution of $A_s$ with another candidate),
        \item and $\bar{R} \cap \relation_{x_1, \dots, x_{\ell - 1}}$ is a candidate of $(\relation_{x_1, \dots, x_{\ell - 1}}, \fdset{})$ (or shares a value distribution of $A_s$ with another candidate).
    \end{itemize}
    
    These three bullet points are proved by: firstly, an intersection with $\bar{R}$ is equivalent to taking a subset of $\bar{R}$, and we know that a subset of an \srepair{} is still an \srepair{} over the entire relation $R$ as well as the sub-relation $\relation_{x_{\ell}}$ (resp. $\relation_{x_1, \dots, x_{\ell - 1}}$). Secondly, $\bar{R} \cap \relation_{x_{\ell}}$ (resp. $\bar{R} \cap \relation_{x_1, \dots, x_{\ell - 1}}$) is not representatively dominated by any other candidate in $\candset{\relation{}_{x_{\ell}}, \fdset{}_{-X}}$ (resp. $\candset{\relation{}_{x_1, \dots, x_{\ell - 1}}, \fdset{}}$). Otherwise, if such a candidate $R_0$ exists in $\candset{\relation{}_{x_{\ell}}, \fdset{}_{-X}}$ (resp. $\candset{\relation{}_{x_1, \dots, x_{\ell - 1}}, \fdset{}}$), then $R_0 \cup (\bar{R} - \relation_{x_{\ell}})$ must be a candidate of $\candset{\relation{}_{x_1, \dots, x_{\ell}}, \fdset{}}$ that representatively dominates $\bar{R}$, resulting in the contradiction with the definition of $\bar{R}$.  Finally, the only reason that $\bar{R}$ is not included in $$\{R' \cup R'' \mid \forall R' \in \candset{\relation{}_{x_{\ell}}, \fdset{}_{-X}}~\forall R'' \in \candset{\relation{}_{x_1, \dots, x_{\ell - 1}}, \fdset{}}\}$$ is that it shares the same distribution of $A_s$ values with another candidate in this \candidateset{}.  Similarly, this step guarantees the properties of a \candidateset{} by \reprinsert{} as well.

    Combining $\ell = 1$ and $\ell > 1$, we prove that \commonlhsreduction{} computes $\candset{\relation{}_{x_1, \dots, x_{\ell}}, \fdset{}}$ for every $\ell \in [1,n]$. Since $\relation{}_{x_1, \dots, x_{n}} = \relation{}$, we have $\candset{\relation{}, \fdset{}} \gets \candset{\relation{}_{x_1, \dots, x_n}, \fdset{}}$.

    \smallskip
    Therefore, \reduce{} is correct when $w = \beta$.
    
\end{proof}

\subsection{Proof of Theorem~\ref{thm:poly-time-chain}}\label{subsec:proof_thm_poly_time_chain}
\polytimechainthm*
\begin{proof}
    We first prove the correctness and then the time complexity of \dpalgo.
    \paragraph*{(Correctness)} We argue that \dpalgo{} computes the optimal \rsrepair{} when \rc{} is fixed and \fdset{} forms an LHS chain. The proof consists of two parts: (a) \reduce{} correctly computes the \candidateset{} $\candset{\relation{}, \Delta}$ (proved in \Cref{lem:reduce-correct}) and (b) Lines 2-3 of \dpalgo{} compute the optimal \rsrepair{} by applying \postclean{} to $\candset{\relation{}, \Delta}$ (proved in \Cref{lemma:candset}).
    \paragraph*{(Time complexity)} The time complexity of \dpalgo{} is $O(m \cdot |\fdset{}| \cdot k \cdot |\relation{}|^{3k + 2})$ as implied by \Cref{lem:dpalgo_complexity}.
\end{proof}

\section{Details from Section~5}
\subsection{LP Rounding-based Algorithms}\label{subsec:lp}
Besides \lprepr{}, we also show another polynomial-time rounding procedure, \lpgreedy{}, in \Cref{subsec:lpgreedy}. They differs in whether the rounding considers representation or not. From the experiments, we see that they have close performance while \lprepr{} is better in scenarios where the dataset contains a large number of errors.

\subsubsection{Representation-Aware Rounding (\lprepr{})}
To address the issue in \lpgreedy{}, we propose a representation-aware greedy algorithm, which is called \lprepr{}. Here, for rounding fractional $x_i$s to 1, tuples from underrepresented groups are prioritized. Inspired by stratified sampling \cite{neyman1992two}, we use the ratio $\frac{\sum_{i:t_i[A_s] = a_{\ell}}{x_i}}{p_\ell}$ to prioritize $x_i$ for rounding, where $p_\ell = \rc(a_\ell)$.
A group with smaller ratio is less representative in the retained tuples. From these less representative groups, \lprepr{} chooses $x_i$ associated with the fewest conflict constraints, then employing similar techniques as \lpgreedy{} does, including applying a final \postclean{}. 
While the computational cost is usually dominated by solving the LP, the rounding for \lprepr{} takes slightly longer than that in \lpgreedy{} due to the additional representation considerations in this step.
\begin{algorithm}
\caption{\lprepr$(\relation{}, \fdset{}, \rc{})$}\label{alg:lprepr}
    \begin{algorithmic}[1]
        \State Build and optimize LP;
        \If{Find a solution $\Vec{x} \in \mathbb{R}^{|\relation{}|}$}
            \While {$\exists~ 0 < x_i < 1$}
                \State $a_{\ell} \gets \arg \min_{a_{\ell} \in \dom{(A_s)}} \frac{\sum_{t_i[A_s] = a_{\ell}}{x_i}}{p_{\ell}}$;
                \State Identify $x_i$ with $t_i[A_s] = a_{\ell}$ and involved in the smallest number of LP constraints;
                \State $x_i \gets 1$ and $x_j \gets 0$ for all $j$ where there is a constraint $x_i + x_j \leq 1$ exists;
            \EndWhile
            \State $R' \coloneq \{t_i \mid x_i = 1, \forall i \in [1, |\relation{}|]\}$;
            \State \Return \postclean{$(R', \rc{})$};
        \EndIf
        \State \Return $\emptyset$;
    \end{algorithmic}
\end{algorithm}
To alleviate this issue in \lpgreedy{}, we propose a representat-ion-aware enhancement in \lprepr{} as outlined in \cref{alg:lprepr}. The primary modification lies in lines 4 and 5, where considerations of representation are incorporated by prioritizing the selection of tuples from underrepresented sensitive groups. \cut{during the selection of undecided $x_i$ values.} Inspired by stratified sampling in \cite{neyman1992two}, we introduce $\frac{\sum_{t_i[A_s] = a_{\ell}}{x_i}}{p_c}$, representing the ratio of the tuple with a given $A_s$-value $a_{\ell}$ to the expected proportion $p_{\ell}$. A smaller ratio indicates a less representative group in the retained tuples. From these less representative groups, \lprepr{} chooses $x_i$ associated with the fewest conflict constraints, then employing similar techniques as \lpgreedy{} does. While the computational cost remains primarily on solving the LP for most cases, the rounding takes slightly longer than that in \lpgreedy{} due to the additional representation considerations in this step. Similarly, the complexity of this rounding is $O(|\relation{}|^2)$.

\subsubsection{Greedy Rounding (\lpgreedy{})}\label{subsec:lpgreedy}



The greedy rounding algorithm \lpgreedy\
greedily chooses a fractional $x_i$, where $0 < x_i < 1$, that participates in the smallest number of LP constraints.
Then it rounds $x_i$ to 1 and for all constraints $x_i + x_j \leq 1$, it rounds $x_j$ to 0, thereby satisfying the linear constraints and consequently resolving the FD violations. 
Once all the FDs in $\fdset$ are satisfied, it calls \postclean{} (\Cref{subsec:postclean}) to ensure that $\rc$ is also satisfied. 
When dealing with datasets containing a large number of errors, this simple greedy rounding approach does not work well as it fails to give priority to sub-populations that are under-representative compared to the desired proportions specified by \rc{}.

\begin{algorithm}
\caption{\lpgreedy$(\relation{}, \fdset{}, \rc{})$}\label{alg:lpgreedy}
    \begin{algorithmic}[1]
        \State Build and optimize LP;
        \If{Find a solution $\Vec{x} \in \mathbb{R}^{|\relation{}|}$}
            \While {$\exists~0 < x_i < 1$}
                \State Identify $x_i$ that is involved in the smallest number of LP constraints;
                \State $x_i \gets 1$ and $x_j \gets 0$ for all $j$ where there is a constraint $x_i + x_j \leq 1$ exists;
            \EndWhile
            \State $R' \gets \{t_i \mid x_i = 1, \forall i \in [1, |\relation{}|]\}$;
            \State \Return \postclean{$(R', \rc{})$};
        \EndIf
        \State \Return $\emptyset$;
    \end{algorithmic}
\end{algorithm}

It is important to note that representation is not considered before line 8. The time cost consists of two parts. The time cost of building and optimizing the LP depends on the number of constraints, which is at most $O(|\relation{}|^2)$ in our case. On the other hand, the rounding step consists of an enumeration of undecided variables and their neighbors, where the complexity of the rounding step is $O(|\relation{}|^2)$.

However, we have observed performance degradation with the greedy rounding approach, especially when dealing with datasets containing a large number of errors. The greedy rounding fails to preserve the minor sub-groups because the central focus of rounding is primarily on maximizing the number of tuples retained.

\section{More Experimental Results}\label{sec:appendix_experiment}
\subsection{More results for \Cref{ex:intro-distribution}}\label{sec:more_intro_example}
\Cref{fig:combined_motivating_example} extend the \Cref{ex:intro-distribution}, showing results across different relative distributions. 
\begin{figure}[!tbp]
\centering
\begin{minipage}{1\linewidth}
\centering
\begin{subfigure}{0.45\linewidth}
\includegraphics[width=\linewidth]{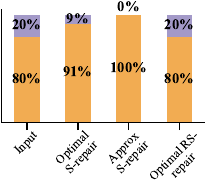}
\caption{Distribution of disability (20\%-80\%)}
\label{fig:motivating_example_a_appendix_0}
\end{subfigure}%
\hfill
\begin{subfigure}{0.5\linewidth}
\includegraphics[width=\linewidth]{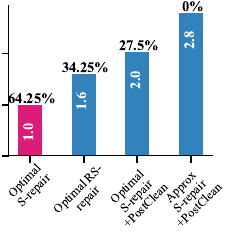}
\caption{Cost of representation (20\%-80\%)}
\label{fig:motivating_example_b_appendix_0}
\end{subfigure}%
\end{minipage}
\begin{minipage}{1\linewidth}
\centering
\begin{subfigure}{0.45\linewidth}
\includegraphics[width=\linewidth]{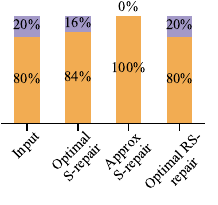}
\caption{Distribution of disability (33\%-67\%)}
\label{fig:motivating_example_a_appendix_1}
\end{subfigure}%
\hfill
\begin{subfigure}{0.5\linewidth}
\includegraphics[width=\linewidth]{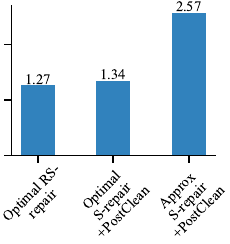}
\caption{Cost of representation (33\%-67\%)}
\label{fig:motivating_example_b_appendix_1}
\end{subfigure}%
\end{minipage}
\begin{minipage}{1\linewidth}
\centering
\begin{subfigure}{0.45\linewidth}
\includegraphics[width=\linewidth]{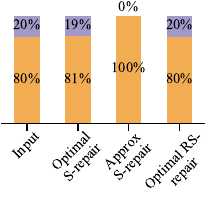}
\caption{Distribution of disability (50\%-50\%)}
\label{fig:motivating_example_a_appendix_2}
\end{subfigure}%
\hfill
\begin{subfigure}{0.5\linewidth}
\includegraphics[width=\linewidth]{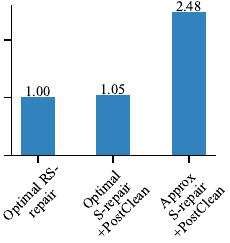}
\caption{Cost of representation (50\%-50\%)}
\label{fig:motivating_example_b_appendix_2}
\end{subfigure}%
\end{minipage}
\begin{minipage}{1\linewidth}
\centering
\includegraphics[width=0.5\linewidth]{figures/intro_example/intro_distribution_legend.pdf}
\end{minipage}

\caption{Cost of representation for ACS data and disability status across different relative noise distributions.}
\label{fig:combined_motivating_example}
\Description[]{}
\end{figure}

\subsection{Additional experiments from \Cref{subsec:exp_cost_repr}}
In this section, we present the deletion overhead, deletion percentage of \srepair{} and \rsrepair{} for varying representation and relative noise distribution for COMPAS data (5\% overall noise level, 10K tuples, chain and non-chain FDs) in \Cref{tab:del_overhead_comparison_appendix}. It demonstrates similar trends as the ACS data shown in \Cref{tab:del_overhead_comparison} in \Cref{subsec:exp_cost_repr}.
\begin{table*}[!htbp]
\centering
\caption{Deletion overhead for varying representations (80\%-20\% and 50\%-50\%) and relative noise distributions (COMPAS data, 5\% overall noise, 10k tuples). ``S-rep. \%'' and ``RS-rep. \%'' stand for deletion percentage of \srepair{} and and for \rsrepair{}, respectively.}
\label{tab:del_overhead_comparison_appendix}
\setlength{\tabcolsep}{2pt} 
\footnotesize
\begin{subtable}[!htp]{0.24\textwidth}
    \centering
    \caption{Chain FD set: 80\%-20\%}
    \label{tab:del_overhead_a_appendix}
    \begin{tabularx}{\textwidth}{|>{\centering\arraybackslash}p{0.9cm}|XXX|}
        \hline
        noise & del. ratio & S-rep. \% & RS-rep. \% \\
        \hline
        20\%-80\% & 2.459 & 8.31 & 20.45 \\
        40\%-60\% & 1.377 & 8.35 & 11.50 \\
        50\%-50\% & 1.016 & 8.35 & 8.48 \\
        60\%-40\% & 1.072 & 8.26 & 8.85 \\
        80\%-20\% & 1.178 & 8.24 & 9.70 \\
        \hline
    \end{tabularx}
\end{subtable}%
\hfill
\begin{subtable}[!htp]{0.24\textwidth}
    \centering
    \caption{Chain FD set: 50\%-50\%}
    \label{tab:del_overhead_b_appendix}
    \begin{tabularx}{\textwidth}{|>{\centering\arraybackslash}p{0.9cm}|XXX|}
        \hline
        noise (\%) & del. ratio & S-rep. \% & RS-rep. \% \\
        \hline
        20\%-80\% & 1.584 & 8.24 & 13.05 \\
        40\%-60\% & 1.191 & 8.39 & 10.00 \\
        50\%-50\% & 1.011 & 8.33 & 8.42 \\
        60\%-40\% & 1.186 & 8.35 & 9.90 \\
        80\%-20\% & 1.578 & 8.28 & 13.07 \\
        \hline
    \end{tabularx}
\end{subtable}%
\hfill
\begin{subtable}[!htp]{0.24\textwidth}
    \centering
    \caption{Non-chain FD set: 80\%-20\%}
    \label{tab:del_overhead_c_appendix}
    \begin{tabularx}{\textwidth}{|>{\centering\arraybackslash}p{0.9cm}|XXX|}
        \hline
        noise (\%) & del. ratio & S-rep. \% & RS-rep. \% \\
        \hline
        20\%-80\% & 1.965 & 24.40 & 47.95 \\
        40\%-60\% & 1.000 & 26.07 & 26.08 \\
        50\%-50\% & 1.077 & 26.03 & 28.05 \\
        60\%-40\% & 1.157 & 26.07 & 30.15 \\
        80\%-20\% & 1.289 & 25.57 & 32.95 \\
        \hline
    \end{tabularx}
\end{subtable}
\hfill
\begin{subtable}[!htp]{0.24\textwidth}
    \centering
    \caption{Non-chain FD set: 50\%-50\%}
    \label{tab:del_overhead_d_appendix}
    \begin{tabularx}{\textwidth}{|>{\centering\arraybackslash}p{0.9cm}|XXX|}
        \hline
        noise (\%) & del. ratio & S-rep. \% & RS-rep. \% \\
        \hline
        20\%-80\% & 1.635 & 24.74 & 40.45 \\
        40\%-60\% & 1.186 & 25.83 & 30.64 \\
        50\%-50\% & 1.000 & 25.68 & 25.68 \\
        60\%-40\% & 1.154 & 25.67 & 29.63 \\
        80\%-20\% & 1.624 & 24.56 & 39.89 \\
        \hline
    \end{tabularx}
\end{subtable}
\end{table*}

\subsection{Additional experiments from \Cref{subsec:exp_quality}}
We present the repair quality for ACS and COMPAS data of different value distributions of the sensitive attribute for 5\% noise level in \Cref{tab:appendix_quality}. It also shows similar trends to the results for 10\% noise level results in \Cref{tab:quality_acs_compas}. 

\begin{table*}[!htp]
\centering
\caption{Repair quality of different algorithms for ACS data and COMPAS data (5\% noise level) with chain and non-chain FD sets.}
\label{tab:appendix_quality}
\footnotesize
\setlength{\tabcolsep}{0pt} 
\renewcommand{\arraystretch}{1} 
\begin{subtable}{\textwidth}
    \centering
    \begin{tabularx}{\textwidth}{|l|*{4}{>{\centering\arraybackslash}X}| *{4}{>{\centering\arraybackslash}X}|*{4}{>{\centering\arraybackslash}X}|}
    \hline
    \multicolumn{13}{|c|}{\cellcolor{gray!50}\textbf{(a) ACS: chain FD set}}\\\hline
    {\bf Sensitive attribute distribution} & \multicolumn{4}{c|}{\textbf{80\%-20\%}} & \multicolumn{4}{c|}{\textbf{60\%-40\%}} & \multicolumn{4}{c|}{\textbf{50\%-50\%}} \\
    \hline
    \textbf{Algorithm / Size (K)} & \textbf{4} & \textbf{6} & \textbf{8} & \textbf{10} & \textbf{4} & \textbf{6} & \textbf{8} & \textbf{10} & \textbf{4} & \textbf{6} & \textbf{8} & \textbf{10} \\
    \hline
    \expglobalilp{} & 100 & 100 & 100 & 100 & 100 & 100 & 100 & 100 & 100 & 100 & 100 & 100 \\
    \rowcolor{gray!30} \textbf{\expdpalgo{} (= \expscalableheuristic{})} & \textbf{100} & \textbf{100} & \textbf{100} & \textbf{100} & \textbf{100} & \textbf{100} & \textbf{100} & \textbf{100} & \textbf{100} & \textbf{100} & \textbf{100} & \textbf{100} \\
    \bilp{}+\bpostclean{} & 93.48 & 93.73 & 94.62 & 94.65 & 98.95 & 99.30 & 99.21 & 99.28 & 99.94 & 99.95 & 99.99 & 99.93 \\
    \bdp{}+\bpostclean{} & 93.48 & 93.73 & 94.62 & 94.65 & 98.95 & 99.30 & 99.21 & 99.28 & 99.94 & 99.95 & 99.99 & 99.93 \\
    \bmuse{}+\bpostclean{} & 92.87 & 94.31 & 94.79 & - & - & - & - & - & - & - & - & -\\ 
\hline
\end{tabularx}
\end{subtable}

\begin{subtable}{\textwidth}
    \centering
    \begin{tabularx}{\textwidth}{|l|*{4}{>{\centering\arraybackslash}X}| *{4}{>{\centering\arraybackslash}X}|*{4}{>{\centering\arraybackslash}X}|}
    \hline
    
    \multicolumn{13}{|c|}{\cellcolor{gray!50}\textbf{(b) ACS: non-chain FD set}}\\\hline
    {\bf Sensitive attribute distribution} & \multicolumn{4}{c|}{\textbf{80\%-20\%}} & \multicolumn{4}{c|}{\textbf{60\%-40\%}} & \multicolumn{4}{c|}{\textbf{50\%-50\%}} \\
    \hline
    \textbf{Algorithm / Size (K)} & \textbf{4} & \textbf{6} & \textbf{8} & \textbf{10} & \textbf{4} & \textbf{6} & \textbf{8} & \textbf{10} & \textbf{4} & \textbf{6} & \textbf{8} & \textbf{10} \\
    \hline
    \expglobalilp{} & 100 & 100 & 100 & 100 & 100 & 100 & 100 & 100 & 100 & 100 & 100 & 100 \\
    \rowcolor{gray!30} \textbf{\expscalableheuristic{}} & \textbf{96.68} & \textbf{98.63} & \textbf{98.75} & \textbf{99.15} & \textbf{97.96} & \textbf{97.89} & \textbf{97.76} & \textbf{98.00} & \textbf{98.04} & \textbf{98.79} & \textbf{98.84} & \textbf{98.98} \\
    \explprepr{} & 81.21 & 90.07 & 90.71 & 90.51 & 93.78 & 88.28 & 87.11 & 92.73 & {\bf 99.36} & 98.67 & {\bf 99.45} & {\bf 99.08} \\
    \bilp{}+\bpostclean{} & 79.34 & 84.93 & 85.43 & 86.01 & 92.94 & 93.67 & 93.60 & 94.31 & {\bf 98.26} & 98.69 & 98.55 & {\bf 99.13}  \\
    \bapprox{}+\bpostclean{} & 0.43 & 6.34 & 6.72 & 5.09 & 36.68 & 37.84 & 36.74 & 36.94 & 48.86 & 50.74 & 49.72 & 50.50  \\
    \hline
    \end{tabularx}
    \end{subtable}

\begin{subtable}{\textwidth}
    \centering
    \begin{tabularx}{\textwidth}{|l|*{4}{>{\centering\arraybackslash}X}|*{4}{>{\centering\arraybackslash}X}|*{3}{>{\centering\arraybackslash}X}|}
    \hline
    \multicolumn{12}{|c|}{\cellcolor{gray!50}\textbf{(c) COMPAS: chain FD set}}\\\hline
    {\bf Sensitive attribute distribution} & \multicolumn{4}{c|}{\textbf{80\%-20\%}} & \multicolumn{4}{c|}{\textbf{60\%-40\%}} & \multicolumn{3}{c|}{\textbf{50\%-50\%}} \\
    \hline
    \textbf{Algorithm / Size (K)} & \textbf{4} & \textbf{10} & \textbf{20} & \textbf{30} & \textbf{4} & \textbf{10} & \textbf{20} & \textbf{30} & \textbf{4} & \textbf{10} & \textbf{20} \\
    \hline
    \expglobalilp{} & 100 & 100 & 100 & 100 & 100 & 100 & 100 & 100 & 100 & 100 & 100 \\
    \rowcolor{gray!30} \textbf{\expdpalgo{}} & \textbf{100} & \textbf{100} & \textbf{100} & \textbf{100} & \textbf{100} & \textbf{100} & \textbf{100} & \textbf{100} & \textbf{100} & \textbf{100} & \textbf{100} \\
    \bilp{}+\bpostclean{} & 99.61 & 100 & 100 & 100 & 100 & 100 & 100 & 100 & 100 & 100 & 100 \\
    \bdp{}+\bpostclean{} & 99.61 & 100 & 100 & 100 & 100 & 100 & 100 & 100 & 100 & 100 & 100 \\
    \bmuse{}+\bpostclean{} & 98.83 & - & - & - & - & - & - & - & - & - & - \\
    \hline
    \end{tabularx}
\end{subtable}

\begin{subtable}{\textwidth}
    \centering
    \begin{tabularx}{\textwidth}{|l|*{4}{>{\centering\arraybackslash}X}|*{4}{>{\centering\arraybackslash}X}|*{3}{>{\centering\arraybackslash}X}|}
    \hline
    \multicolumn{12}{|c|}{\cellcolor{gray!50}\textbf{(d) COMPAS: non-chain FD set}}\\\hline
    {\bf Sensitive attribute distribution} & \multicolumn{4}{c|}{\textbf{80\%-20\%}} & \multicolumn{4}{c|}{\textbf{60\%-40\%}} & \multicolumn{3}{c|}{\textbf{50\%-50\%}} \\
    \hline
    \textbf{Algorithm / Size (K)} & \textbf{4} & \textbf{10} & \textbf{20} & \textbf{30} & \textbf{4} & \textbf{10} & \textbf{20} & \textbf{30} & \textbf{4} & \textbf{10} & \textbf{20} \\
    \hline
    \expglobalilp{} & 100 & 100 & 100 & 100 & 100 & 100 & 100 & 100 & 100 & 100 & 100\\
    \rowcolor{gray!30} \textbf{\expscalableheuristic{}} & \textbf{99.54} & \textbf{96.06} & \textbf{98.15} & \textbf{97.44} & \textbf{99.28} & \textbf{98.70} & \textbf{97.60} & \textbf{97.08} & \textbf{99.13} & \textbf{99.60} & \textbf{99.33} \\
    \explprepr{}  & 100 & 100 & 100 & 100 & 100 & 100 & 100 & 100 & 100 & 100 & 100\\
    \bilp{}+\bpostclean{} & 99.08 & 95.87 & 93.78 & 91.35 & 99.19 & 98.59 & 97.47 & 96.98 & 99.80 & 99.73 & 99.55 \\
    \bapprox{}+\bpostclean{} & 16.09 & 13.26 & 10.44 & 8.34 & 33.12 & 30.99 & 30.86 & 29.97 & 43.62 & 46.07 & 44.76\\
    \hline
    \end{tabularx}
\end{subtable}
\end{table*}

\subsection{Additional experiments from \Cref{subsec:exp_scalability}}
\paragraph{Runtime on ACS and COMPAS Data}
\Cref{tab:runtime_cost_appendix} demonstrates the runtime for ACS and COMPAS data of different value distributions of sensitive attribute (5\% and 10\% noise level, chain and non-chain FD sets).

\begin{table*}[!htp]
\centering
\caption{Runtime (sec) of different algorithms for ACS data (5\% and 10\% noise levels) with chain and non-chain FD sets}
\label{tab:runtime_cost_appendix}
\footnotesize
\setlength{\tabcolsep}{3pt} 
\renewcommand{\arraystretch}{1.1} 
\begin{subtable}{\textwidth}
    \centering
    \begin{tabularx}{\textwidth}{|l|*{4}{>{\centering\arraybackslash}X}|*{4}{>{\centering\arraybackslash}X}|*{4}{>{\centering\arraybackslash}X}|}
    \hline
    \multicolumn{13}{|c|}{\cellcolor{gray!50}\textbf{(a) ACS: non-chain FD set (5\% noise)}}\\\hline
    {\bf Sensitive attribute distribution} & \multicolumn{4}{c|}{\textbf{80\%-20\%}} & \multicolumn{4}{c|}{\textbf{60\%-40\%}} & \multicolumn{4}{c|}{\textbf{50\%-50\%}} \\
    \hline
    \textbf{Algorithm / Size (K)} & \textbf{4} & \textbf{6} & \textbf{8} & \textbf{10} & \textbf{4} & \textbf{6} & \textbf{8} & \textbf{10} & \textbf{4} & \textbf{6} & \textbf{8} & \textbf{10} \\
    \hline
    \expglobalilp{} & 313.49 & 1003.48 & 3188.59 & 17840.45 & 64.97 & 260.79 & 1113.37 & 3781.58 & 30.08 & 161.76 & 903.63 & 2017.56 \\
    {\bf \expscalableheuristic{}} & 2.53 & 2.43 & 3.57 & 4.33 & 4.60 & 7.87 & 12.09 & 15.23 & 5.72 & 10.36 & 13.47 & 27.20 \\
    \explprepr{} & 31.51 & 69.82 & 117.88 & 206.39 & 21.69 & 52.70 & 93.31 & 174.77 & 11.36 & 28.30 & 48.38 & 133.10 \\
    \bilp{}+\bpostclean{} & 62.60 & 101.64 & 250.12 & 619.45 & 34.40 & 100.44 & 209.47 & 475.07 & 24.58 & 124.47 & 182.43 & 518.43 \\
    \bapprox{}+\bpostclean{} & 3.20 & 3.82 & 7.38 & 14.02 & 2.08 & 4.77 & 6.87 & 14.82 & 2.11 & 3.39 & 5.90 & 27.98 \\
    \hline
    \end{tabularx}
\end{subtable}

\begin{subtable}{\textwidth}
    \centering
    \begin{tabularx}{\textwidth}{|l|*{4}{>{\centering\arraybackslash}X}|*{4}{>{\centering\arraybackslash}X}|*{4}{>{\centering\arraybackslash}X}|}
    \hline
    \multicolumn{13}{|c|}{\cellcolor{gray!50}\textbf{(b) ACS: non-chain FD (10\% noise)}}\\\hline
    {\bf Sensitive attribute distribution} & \multicolumn{4}{c|}{\textbf{80\%-20\%}} & \multicolumn{4}{c|}{\textbf{60\%-40\%}} & \multicolumn{4}{c|}{\textbf{50\%-50\%}} \\
    \hline
    \textbf{Algorithm / Size (K)} & \textbf{4} & \textbf{6} & \textbf{8} & \textbf{10} & \textbf{4} & \textbf{6} & \textbf{8} & \textbf{10} & \textbf{4} & \textbf{6} & \textbf{8} & \textbf{10} \\
    \hline
    \expglobalilp{} & 1869.85 & 5072.64 & 58977.39 & 64039.06 & 261.93 & 1344.63 & 4293.89 & 44326.58 & 139.27 & 855.29 & 1227.35 & 3328.55 \\
    {\bf \expscalableheuristic{}} & 2.29 & 2.48 & 2.52 & 2.36 & 5.72 & 10.33 & 12.52 & 16.64 & 9.50 & 12.09 & 13.85 & 23.14 \\
    \explprepr{} & 33.83 & 87.84 & 165.72 & 278.52 & 39.40 & 97.13 & 248.89 & 397.02 & 44.67 & 114.99 & 117.65 & 250.03 \\
    \bilp{}+\bpostclean{} & 135.41 & 412.93 & 869.86 & 673.26 & 104.51 & 307.65 & 1000.11 & 2961.43 & 116.17 & 358.94 & 909.23 & 2404.93 \\
    \bapprox{}+\bpostclean{} & 3.19 & 7.13 & 13.88 & 20.72 & 3.34 & 7.31 & 13.37 & 19.87 & 3.84 & 11.57 & 11.42 & 18.89 \\
    \hline
    \end{tabularx}
\end{subtable}

\begin{subtable}{\textwidth}
    \centering
    \begin{tabularx}{\textwidth}{|l|*{4}{>{\centering\arraybackslash}X}|*{4}{>{\centering\arraybackslash}X}|*{4}{>{\centering\arraybackslash}X}|}
    \hline
    \multicolumn{13}{|c|}{\cellcolor{gray!50}\textbf{(c) ACS: chain FD (5\% noise)}}\\\hline
    {\bf Sensitive attribute distribution} & \multicolumn{4}{c|}{\textbf{80\%-20\%}} & \multicolumn{4}{c|}{\textbf{60\%-40\%}} & \multicolumn{4}{c|}{\textbf{50\%-50\%}} \\
    \hline
    \textbf{Algorithm / Size (K)} & \textbf{4} & \textbf{6} & \textbf{8} & \textbf{10} & \textbf{4} & \textbf{6} & \textbf{8} & \textbf{10} & \textbf{4} & \textbf{6} & \textbf{8} & \textbf{10} \\
    \hline
    \expglobalilp{} & 5.02 & 13.37 & 94.12 & 90.16 & 3.27 & 14.11 & 47.89 & 308.90 & 1.92 & 4.53 & 14.30 & 49.97 \\
    \rowcolor{gray!30} \textbf{\expdpalgo{} (= \expscalableheuristic{})} & 3.49 & 4.80 & 8.41 & 9.06 & 3.73 & 5.88 & 7.14 & 14.61 & 4.78 & 7.72 & 8.24 & 11.75 \\
    \bdp{}+\bpostclean{} & 0.41 & 0.50 & 0.59 & 0.62 & 0.81 & 0.52 & 0.56 & 0.88 & 0.53 & 0.55 & 0.60 & 0.69 \\
    \bilp{}+\bpostclean{} & 1.17 & 2.74 & 5.78 & 8.15 & 1.49 & 2.82 & 5.48 & 9.76 & 1.59 & 3.51 & 6.65 & 13.74 \\
    \bmuse{}+\bpostclean{} & 460.05 & 3904.73 & 15161.43 & - & - & - & - & - & 599.76 & 3850.74 & 14413.26 & - \\
    \hline
    \end{tabularx}
\end{subtable}

\begin{subtable}{\textwidth}
    \centering
    \begin{tabularx}{\textwidth}{|l|*{4}{>{\centering\arraybackslash}X}|*{4}{>{\centering\arraybackslash}X}|*{4}{>{\centering\arraybackslash}X}|}
    \hline
    \multicolumn{13}{|c|}{\cellcolor{gray!50}\textbf{(d) ACS: chain FD (10\% noise)}}\\\hline
    {\bf Sensitive attribute distribution} & \multicolumn{4}{c|}{\textbf{80\%-20\%}} & \multicolumn{4}{c|}{\textbf{60\%-40\%}} & \multicolumn{4}{c|}{\textbf{50\%-50\%}} \\
    \hline
    \textbf{Algorithm / Size (K)} & \textbf{4} & \textbf{6} & \textbf{8} & \textbf{10} & \textbf{4} & \textbf{6} & \textbf{8} & \textbf{10} & \textbf{4} & \textbf{6} & \textbf{8} & \textbf{10} \\
    \hline
    \expglobalilp{} & 5.74 & 64.79 & 413.00 & 1045.78 & 11.40 & 209.54 & 544.95 & 657.28 & 9.97 & 13.10 & 116.37 & 129.90 \\
    \rowcolor{gray!30} \textbf{\expdpalgo{} (= \expscalableheuristic{})} & 5.64 & 6.75 & 9.84 & 12.12 & 5.32 & 8.71 & 11.08 & 28.91 & 7.11 & 9.62 & 12.32 & 14.35 \\
    \bdp{}+\bpostclean{} & 0.55 & 0.64 & 0.69 & 0.71 & 0.55 & 0.64 & 0.86 & 1.02 & 0.63 & 0.67 & 0.73 & 0.91 \\
    \bilp{}+\bpostclean{} & 2.24 & 4.85 & 9.65 & 15.95 & 2.28 & 5.94 & 12.61 & 19.60 & 2.86 & 7.10 & 12.68 & 26.69 \\
    \bmuse{}+\bpostclean{} & 4558.15 & 16187.80 & - & - & 3075.89 & 23414.92 & - & - & 12328.36 & 14197.39 & - & - \\
    \hline
    \end{tabularx}
\end{subtable}

\begin{subtable}{\textwidth}
    \centering
    \begin{tabularx}{\textwidth}{|l|*{4}{>{\centering\arraybackslash}X}|*{4}{>{\centering\arraybackslash}X}|*{3}{>{\centering\arraybackslash}X}|}
    \hline
    \multicolumn{12}{|c|}{\cellcolor{gray!50}\textbf{(e) COMPAS: non-chain FD set (5\% noise)}}\\\hline
    {\bf Sensitive attribute distribution} & \multicolumn{4}{c|}{\textbf{80\%-20\%}} & \multicolumn{4}{c|}{\textbf{60\%-40\%}} & \multicolumn{3}{c|}{\textbf{50\%-50\%}} \\
    \hline
    \textbf{Algorithm / Size (K)} & \textbf{4} & \textbf{10} & \textbf{20} & \textbf{30} & \textbf{4} & \textbf{10} & \textbf{20} & \textbf{30} & \textbf{4} & \textbf{10} & \textbf{20} \\
    \hline
    \expglobalilp{} & 124.48 & 1724.06 & 3038.94 & 8083.23 & 114.57 & 1145.06 & 3689.82 & 6727.07 & 130.72 & 1625.59 & 6766.93\\
    \textbf{\expscalableheuristic{}} & 18.73 & 236.34 & 1413.43 & 5237.60 & 24.86 & 240.79 & 1826.81 & 5966.77 & 31.47 & 278.60 & 1985.86 \\
    \explprepr{} & 23.48 & 148.77 & 601.35 & 1434.94 & 26.45 & 186.49 & 1656.01 & 3558.59 & 30.14 & 186.87 & 1104.72\\
    \bilp{}+\bpostclean{} & 131.51 & 1428.27 & 5331.15 & 12498.97 & 101.97 & 1046.61 & 5898.54 & 13983.64 & 123.99 & 1280.14 & 11003.02 \\
    \bapprox{}+\bpostclean{} & 5.41 & 34.79 & 148.39 & 311.44 & 6.92 & 40.30 & 316.29 & 469.01 & 8.35 & 34.84 & 359.49 \\
    \hline
    \end{tabularx}
\end{subtable}

\begin{subtable}{\textwidth}
    \centering
    \begin{tabularx}{\textwidth}{|l|*{4}{>{\centering\arraybackslash}X}|*{4}{>{\centering\arraybackslash}X}|*{3}{>{\centering\arraybackslash}X}|}
    \hline
    \multicolumn{12}{|c|}{\cellcolor{gray!50}\textbf{(e) COMPAS: non-chain FD set (10\% noise)}}\\\hline
    {\bf Sensitive attribute distribution} & \multicolumn{4}{c|}{\textbf{80\%-20\%}} & \multicolumn{4}{c|}{\textbf{60\%-40\%}} & \multicolumn{3}{c|}{\textbf{50\%-50\%}} \\
    \hline
    \textbf{Algorithm / Size (K)} & \textbf{4} & \textbf{10} & \textbf{20} & \textbf{30} & \textbf{4} & \textbf{10} & \textbf{20} & \textbf{30} & \textbf{4} & \textbf{10} & \textbf{20} \\
    \hline
    \expglobalilp{} & 294.09 & 1620.48 & 5289.97 & 10235.06 & 221.40 & 1672.16 & 5294.01 & 10742.42 & 242.20 & 1415.72 & 8104.47 \\
    \textbf{\expscalableheuristic{}} & 13.70 & 82.28 & 646.44 & 2110.72 & 18.13 & 179.43 & 1518.88 & 3924.87 & 23.66 & 328.31 & 3365.94 \\
    \explprepr{} & 56.36 & 421.65 & 2881.62 & 18485.24 & 65.67 & 373.81 & 5688.91 & 19661.11 & 47.12 & 309.59 & 1789.07 \\
    \bilp{}+\bpostclean{} & 294.09 & 1620.48 & 5289.97 & 10235.06 & 221.40 & 1672.16 & 5294.01 & 10742.42 & 242.20 & 1415.72 & 8104.47 \\
    \bapprox{}+\bpostclean{} & 9.58 & 59.64 & 254.31 & 632.66 & 8.80 & 62.63 & 286.61 & 909.84 & 8.95 & 56.94 & 271.85 \\
    \hline
    \end{tabularx}
\end{subtable}

\begin{subtable}{\textwidth}
    \centering
    \begin{tabularx}{\textwidth}{|l|*{4}{>{\centering\arraybackslash}X}|*{4}{>{\centering\arraybackslash}X}|*{3}{>{\centering\arraybackslash}X}|}
    \hline
    \multicolumn{12}{|c|}{\cellcolor{gray!50}\textbf{(c) COMPAS: chain FD set (5\% noise)}}\\\hline
    {\bf Sensitive attribute distribution} & \multicolumn{4}{c|}{\textbf{80\%-20\%}} & \multicolumn{4}{c|}{\textbf{60\%-40\%}} & \multicolumn{3}{c|}{\textbf{50\%-50\%}} \\
    \hline
    \textbf{Algorithm / Size (K)} & \textbf{4} & \textbf{10} & \textbf{20} & \textbf{30} & \textbf{4} & \textbf{10} & \textbf{20} & \textbf{30} & \textbf{4} & \textbf{10} & \textbf{20} \\
    \hline
    \expglobalilp{} & 2.48 & 21.12 & 96.14 & 211.90 & 2.86 & 25.25 & 122.54 & 234.58 & 2.49 & 26.85 & 114.82 \\
    \textbf{\expdpalgo{} (= \expscalableheuristic{})} & 0.04 & 0.06 & 0.07 & 0.09 & 0.04 & 0.06 & 0.09 & 0.10 & 0.04 & 0.05 & 0.07 \\
    \bilp{}+\bpostclean{} & 2.09 & 18.06 & 86.95 & 177.53 & 2.81 & 20.09 & 121.65 & 214.16 & 2.58 & 19.78 & 95.52 \\
    \bdp{}+\bpostclean{} & 0.03 & 0.04 & 0.05 & 0.07 & 0.04 & 0.04 & 0.06 & 0.09 & 0.05 & 0.04 & 0.05 \\
    \bmuse{}+\bpostclean{} & 11071.22 & - & - & - & - & - & - & - & - & - & - \\
    \hline
    \end{tabularx}
\end{subtable}

\begin{subtable}{\textwidth}
    \centering
    \begin{tabularx}{\textwidth}{|l|*{4}{>{\centering\arraybackslash}X}|*{4}{>{\centering\arraybackslash}X}|*{3}{>{\centering\arraybackslash}X}|}
    \hline
    \multicolumn{12}{|c|}{\cellcolor{gray!50}\textbf{(c) COMPAS: chain FD set (10\% noise)}}\\\hline
    {\bf Sensitive attribute distribution} & \multicolumn{4}{c|}{\textbf{80\%-20\%}} & \multicolumn{4}{c|}{\textbf{60\%-40\%}} & \multicolumn{3}{c|}{\textbf{50\%-50\%}} \\
    \hline
    \textbf{Algorithm / Size (K)} & \textbf{4} & \textbf{10} & \textbf{20} & \textbf{30} & \textbf{4} & \textbf{10} & \textbf{20} & \textbf{30} & \textbf{4} & \textbf{10} & \textbf{20} \\
    \hline
    \expglobalilp{} & 5.32 & 65.48 & 226.84 & 443.41 & 7.18 & 63.17 & 234.84 & 417.15 & 5.97 & 89.14 & 244.40 \\
    \textbf{\expdpalgo{} (= \expscalableheuristic{})} & 0.05 & 0.07 & 0.08 & 0.16 & 0.05 & 0.06 & 0.07 & 0.08 & 0.04 & 0.05 & 0.07 \\
    \bilp{}+\bpostclean{} & 4.31 & 43.77 & 187.56 & 401.77 & 5.83 & 45.29 & 241.47 & 364.47 & 4.88 & 43.49 & 217.83 \\
    \bdp{}+\bpostclean{} & 0.03 & 0.04 & 0.05 & 0.08 & 0.04 & 0.04 & 0.05 & 0.07 & 0.06 & 0.04 & 0.05 \\
    \bmuse{}+\bpostclean{} & 38218.90 & - & - & - & - & - & - & - & - & - & - \\
    \hline
    \end{tabularx}
\end{subtable}
\end{table*}

\paragraph{Runtime on Flight Data}
\begin{figure}[!htbp]
    \centering
    \begin{minipage}{\linewidth}
        \begin{subfigure}{0.5\linewidth}
        \captionsetup{aboveskip=0pt}
        \includegraphics[width=\linewidth]{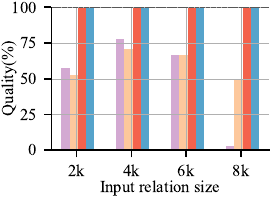}
        \caption{Repair quality}
            \label{fig:quality_flight}
        \end{subfigure}
        \begin{subfigure}{0.5\linewidth}
        \captionsetup{aboveskip=0pt}
            \includegraphics[width=\linewidth]{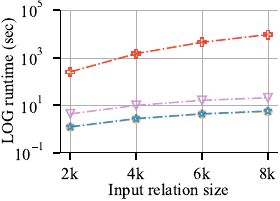}
            \caption{Runtime}
            \label{fig:time_flight}
        \end{subfigure}
    \end{minipage}
    \begin{minipage}{\linewidth}
        \centering
        \includegraphics[width=0.8\linewidth]{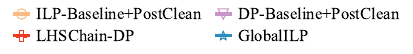}
    \end{minipage}
    \caption{Repair quality and runtime for Flight data (chain FDs)}
    \label{fig:quality_time_flight}
    \Description[]{}
\end{figure}

\Cref{fig:time_flight} shows the runtimes for the Flight dataset that has chain FDs. 
While both \dpalgo{} and \globalilp{} provides the optimal \rsrepair{}, 
\globalilp{} is faster than \dpalgo{}  (and other baselines) for smaller datasets \globalilp{}.
Since the tuples are quite distinct on the attribute "DF" (only 2 or 3 tuples sharing the same key), \dpalgo{} takes a much longer runtime to loop over all new repairs that are potential to be a candidate in the step of \commonlhsreduction{}. 
\dpalgo{} seems to take a longer runtime when the value of \commonlhs{} has a sparse domain.

\subsection{Experiments on Credit Card Transaction Dataset}\label{subsec:creditcard_experiment}
\subsubsection{Setup}
Credit card transaction~\cite{choksi_credit_2023} dataset provides detailed transaction records, including associated personal and merchant information, which can be used for fraud detection. 7 attributes are involved in our experiments.

We consider a non-chain FD set here: \{$\fd{Merchant}{Category}$;\:\\
$\fd{City}{State}$;\: $\fd{First,Last}{Gender}$;\: $\fd{First,Last}{City}$\}. we use an RC on the sensitive attribute $\att{Gender}$, $\{\%\texttt{Male} = \frac{20}{100}, \%\texttt{Female} = \frac{80}{100}\}$. The method of noise input generation is the same as that used in ACS and COMPAS data (\Cref{sec:expt-noise}).

\subsubsection{RS-repair Quality and Runtime Cost}

\begin{figure}[!htp]
\centering
\begin{minipage}{1\linewidth}
\centering
\begin{subfigure}{0.5\linewidth}
\includegraphics[width=\linewidth]{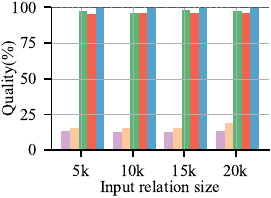}
\centering
\caption{Repair quality}
\label{fig:quality_credit}
\end{subfigure}%
\begin{subfigure}{0.5\linewidth}
\includegraphics[width=\linewidth]{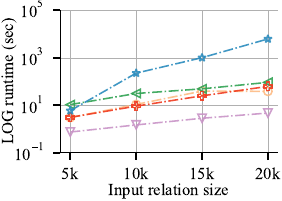}
\centering
\caption{Runtime}
\label{fig:runtime_credit}
\end{subfigure}%
\end{minipage}
\centering
\begin{minipage}{1\linewidth}
\centering
\includegraphics[width=0.9\linewidth]{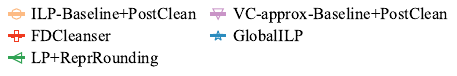}
\end{minipage}
\caption{Repair quality and runtime of different algorithms for Credit card transaction data (5\% noise)}
\label{fig:credicard_results}
\end{figure}
\Cref{fig:credicard_results} presents the repair quality and runtime of different algorithms in the Credit card transaction dataset under the noise level $5\%$. As shown in \Cref{fig:quality_credit}, both \explprepr{} and \expscalableheuristic{} achieve high quality above $95\%$, while \explprepr{} performs slightly better than \expscalableheuristic{}. However, the repair quality of \bilp{}+\\ \bpostclean{} and \bapprox{}+\bpostclean{} achieved is significantly lower than the optimal $100\%$ repair quality provided by \expglobalilp{}.

\Cref{fig:runtime_credit} demonstrates that the runtime of \expglobalilp{} increases exponentially as the input size increases, but \explprepr{} and \expdpalgo{} provide much more efficient alternatives. Moreover, \expdpalgo{} takes less time than \explprepr{} while achieving high repair quality as well.

\subsection{Additional experiments varying noise levels}

\begin{table*}[ht]
\centering
\caption{COMPAS: (Chain FD set) Repair Quality}
\label{tab:my_label4}
\footnotesize
\begin{tabularx}{\textwidth}{l|XXXXXXX|XXXXXXX}
\toprule
noise & \multicolumn{7}{c|}{15\%} & \multicolumn{7}{c}{20\%} \\
size & 4k & 6k & 8k & 10k & 20k & 30k & 40k & 4k & 6k & 8k & 10k & 20k & 30k & 40k \\
\midrule
\textbf{GlobalILP} & 100.00 & 100.00 & 100.00 & 100.00 & 100.00 & 100.00 & 100.00 & 100.00 & 100.00 & 100.00 & 100.00 & 100.00 & 100.00 & 100.00 \\
\textbf{LHSChain-DP} & 100.00 & 100.00 & 100.00 & 100.00 & 100.00 & 100.00 & 100.00 & 100.00 & 100.00 & 100.00 & 100.00 & 100.00 & 100.00 & 100.00 \\
\textbf{ILP-Baseline+PostClean} & 93.84 & 96.46 & 95.04 & 95.77 & 95.29 & 96.95 & 97.48 & 91.04 & 93.01 & 93.37 & 92.35 & 92.40 & 93.50 & 94.55 \\
\textbf{DP-Baseline+PostClean} & 93.84 & 96.46 & 95.04 & 95.77 & 95.29 & 96.95 & 97.48 & 91.04 & 93.01 & 93.37 & 92.35 & 92.40 & 93.50 & 94.55 \\
\textbf{MuSe-Baseline+PostClean} & 95.50 &   &   &   &   &   &   &   &   &   &   &   &   &   \\
\bottomrule
\end{tabularx}
\end{table*}

\begin{table*}[ht]
\centering
\caption{COMPAS: (Chain FD set) Run time cost}
\label{tab:my_label5}
\footnotesize
\begin{tabularx}{\textwidth}{l|XXXXXXX|XXXXXXX}
\toprule
noise & \multicolumn{7}{c|}{15\%} & \multicolumn{7}{c}{20\%} \\
size & 4k & 6k & 8k & 10k & 20k & 30k & 40k & 4k & 6k & 8k & 10k & 20k & 30k & 40k \\
\midrule
\textbf{GlobalILP} & 8.33 & 30.18 & 85.78 & 179.88 & 463.69 & 1376.47 & 1114.64 & 11.72 & 47.65 & 100.34 & 202.19 & 390.46 & 1590.58 & 2390.69 \\
\textbf{LHSChain-DP} & 0.05 & 0.06 & 0.08 & 0.07 & 0.11 & 0.15 & 0.18 & 0.06 & 0.08 & 0.08 & 0.11 & 0.11 & 0.28 & 0.36 \\
\textbf{ILP-Baseline+PostClean} & 6.70 & 18.10 & 36.64 & 65.74 & 336.20 & 770.82 & 1811.51 & 8.92 & 32.38 & 49.07 & 79.95 & 331.32 & 932.78 & 2463.53 \\
\textbf{DP-Baseline+PostClean} & 0.03 & 0.03 & 0.04 & 0.04 & 0.06 & 0.11 & 0.23 & 0.03 & 0.03 & 0.04 & 0.04 & 0.05 & 0.09 & 0.12 \\
\textbf{MuSe-Baseline+PostClean} & 74113.50 &   &   &   &   &   &   &   &   &   &   &   &   &   \\

\bottomrule
\end{tabularx}
\end{table*}

This section presents the quality and runtime for COMPAS data with larger noise levels, 15\% and 20\% for chain FDs in \Cref{fig:varynoise_quality,tab:my_label4,tab:my_label5}.
\begin{figure}[!htp]
\centering
\begin{minipage}{1\linewidth}
\centering
\begin{subfigure}{0.5\linewidth}
\includegraphics[width=\linewidth]{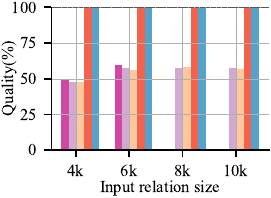}
\centering
\caption{ACS: chain FD set (15\% noise)}
\end{subfigure}%
\begin{subfigure}{0.5\linewidth}
\includegraphics[width=\linewidth]{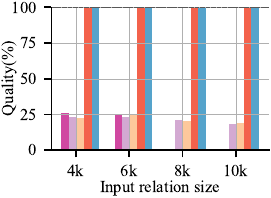}
\centering
\caption{ACS: chain FD set (20\% noise)}
\end{subfigure}%
\end{minipage}
\centering
\begin{minipage}{1\linewidth}
\centering
\includegraphics[width=0.8\linewidth]{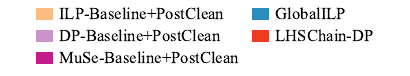}
\end{minipage}
\caption{Repair Quality while varying the noise level}
\label{fig:varynoise_quality}
\end{figure}

\section{Details from Section~8}
\subsection{Discussion on Multiple Sensitive Attributes}

We discuss several next steps for the extension of multiple sensitive attributes. First, a few definitions need to be extended: {(i)} $b$ sensitive attributes, denoted by $A_{s_1}, \dots, A_{s_b}$, of $\schema{}$ and $b$ associated RCs, denoted by $\rc{}_1, \dots, \rc{}_b$; {(ii)} the second item in \Cref{def:repr_optsrepair}, i.e., the definition of \rsrepair{}, will be changed to \textit{$R'$ satisfies all $b$ RCs, i.e., $\rc{}_1, \dots, \rc{}_b$.} In other words, a \rsrepair{} preserves the marginal distribution, given by the RC, for every sensitive attribute. 
Second, cases where handling a single sensitive attribute is NP-hard remain NP-hard when extended to multiple sensitive attributes. The rationale behind this is that the single sensitive attribute problem can be viewed as a special instance of the multiple sensitive attributes problem. 
Third, \globalilp{} and its ILP (\Cref{ilp:global}) can be extended to handle the problems with $b$ RCs, $\rc{}_1, \dots, \rc{}_b$. The constraint in the middle will be changed to:
$$\sum_{i: t_i[A_{s_j}] = a_\ell}{\hspace{-1em}x_i} ~~\geq~~ \rc{}_{j}(a_{\ell}) \times \sum_{i}{x_i} \quad \text{for all } j \in [1, b] \text{ and } a_\ell \in \dom(A_{s_j})$$
Fourth, the representation of a tuple becomes multi-dimensional, so the impact of deleting it or preserving it is also multi-dimensional, which further complicates the extension of \postclean{} and \candidateset{}.
To be specific, in the cases with a single sensitive attribute, \postclean{} can greedily decide on the tuples to delete until it satisfies the RC, but it is unclear that which set of tuples to remove when there are multiple RCs. The cause is, for multiple sensitive attributes, deleting a tuple reduces the representation of a combination of sensitive values. For example, considering a tuple with two sensitive values, $\texttt{Asian}$ and $\texttt{Female}$, deleting it might not be ideal if $\texttt{Asian}$ is over-represented on $\att{Race}$ but $\texttt{Female}$ is under-represented in $\att{Sex}$. Hence, the extension of DP is possible but not trivial. Essentially, we can keep track of more information for all the sensitive attributes, but it is unclear how to extend the definition of \candidateset{} (\Cref{def:candset}), which means it is uncertain to declare if one \srepair{} is representatively equivalent to or representatively dominates the other. Therefore, the size of the candidate set is unknown and consequently the complexity of DP is unknown.

\end{document}